\documentclass[onecolumn,pra,nofootinbib]{revtex4-2}

\usepackage[dvips]{graphicx} 
\usepackage{amsfonts,amssymb,amscd,amsmath,amsthm}
\usepackage{textcomp}
\usepackage{enumerate}
\usepackage{epsfig}
\usepackage{subfigure}
\usepackage{xcolor}
\usepackage[outdir=./]{epstopdf}
\usepackage[colorlinks = true]{hyperref}
\usepackage{physics}
\usepackage{appendix}
\usepackage{comment, apptools, thmtools, thm-restate}
\usepackage{diagbox}
\usepackage[most]{tcolorbox}
\usepackage{ulem}
\usepackage{bm}

\usepackage{tikz}
\usepackage[]{qcircuit}
\usetikzlibrary{arrows}
\usetikzlibrary{shapes,fadings,snakes}
\usetikzlibrary{decorations.pathmorphing,patterns}
\usetikzlibrary{calc}
\usetikzlibrary{positioning}

\newtheorem{theorem}{Theorem}
\newtheorem{lemma}{Lemma}

\newtheorem{proposition}{Proposition}
\newtheorem{definition}{Definition}

\newtcolorbox[auto counter]{mybox}[2][]{
	enhanced,
	breakable,
	colback=blue!5!white,
	colframe=blue!75!black,
	fonttitle=\bfseries,
	title=Box \thetcbcounter: #2,#1
}

\graphicspath{{./figure/}}

\newcommand{\lket}[1]{\vert #1 \rangle\!\rangle}
\newcommand{\lbra}[1]{\langle\!\langle #1 \vert}
\newcommand{\lbraket}[2]{\langle\!\langle #1 \vert #2 \rangle\!\rangle}

\begin{document}

\title{Scalable fast benchmarking for individual quantum gates with local twirling}
\author{Yihong Zhang}
\affiliation{Center for Quantum Information, Institute for Interdisciplinary Information Sciences, Tsinghua University, Beijing 100084, China}
\author{Wenjun Yu}
\affiliation{Center for Quantum Information, Institute for Interdisciplinary Information Sciences, Tsinghua University, Beijing 100084, China}
\author{Pei Zeng}
\affiliation{Center for Quantum Information, Institute for Interdisciplinary Information Sciences, Tsinghua University, Beijing 100084, China}
\author{Guoding Liu}
\affiliation{Center for Quantum Information, Institute for Interdisciplinary Information Sciences, Tsinghua University, Beijing 100084, China}
\author{Xiongfeng Ma}
\email{xma@tsinghua.edu.cn}
\affiliation{Center for Quantum Information, Institute for Interdisciplinary Information Sciences, Tsinghua University, Beijing 100084, China}

\begin{abstract}
With the development of controllable quantum systems, fast and practical characterization for multi-qubit gates is essential for building high-fidelity quantum computing devices. The usual way to fulfill this requirement via randomized benchmarking asks for the complicated implementation of numerous multi-qubit twirling gates. How to efficiently and reliably estimate the fidelity of a quantum process remains an open problem. In this work, we propose a character-cycle benchmarking protocol and a character-average benchmarking protocol only using local twirling gates to estimate the process fidelity of an individual multi-qubit operation. Our protocols can characterize a large class of quantum gates including and beyond the Clifford group via the local gauge transformation, which forms a universal gate set for quantum computing. We numerically demonstrate our protocols for a non-Clifford gate --- controlled-$(TX)$ and a Clifford gate --- five-qubit quantum error-correcting encoding circuit. The numerical results show that our protocols can efficiently and reliably characterize the gate process fidelities. Compared with the cross-entropy benchmarking, the simulation results show that the character-average benchmarking achieves three orders of magnitude improvements in terms of sampling complexity.
\end{abstract}


\maketitle

\section{Introduction}
Characterizing a quantum process has great importance in both the fundamental study and practical application of quantum information science. With the recent advent of noisy intermediate-scale quantum computing \cite{Preskill2018NISQ}, benchmarking quantum operations is critical for quantum control \cite{RevModPhys.76.1037, Chu2002} as it provides an indicator to assess the experimental devices. It is essential for the development of high-precision quantum information processing instruments. Accurate benchmarking can reliably characterize the noise levels of the quantum operations and plays a critical role in promoting fault-tolerant universal quantum computing \cite{RevModPhys.87.307, Campbell2017}. In practice, we need to evaluate the performance of a quantum circuit to verify whether a quantum algorithm or an error-correcting code is properly implemented in a quantum system.

Numerous approaches have been proposed to characterize quantum processes. Conventional methods like quantum process tomography \cite{Chuang1997tomo} provide a full description of a channel. However, these methods are impractical for large-scale quantum systems as the required experimental resources increase exponentially with the number of qubits, even with state-of-the-art techniques such as compressed sensing \cite{Gross2010prlCompressedSensing, Flammia2012CompressedSensing}. Direct fidelity estimation \cite{Flammia2011prlDirectFidelity} tackles the scaling problem and characterizes the quantum process in terms of average fidelity. Unfortunately, the result inevitably contains extra errors from the state preparation and measurement (SPAM) and hence often over-estimates the noise levels. In reality, SPAM errors usually grow rapidly with the system size so that it is hard to characterize the quantum process accurately for large-scale quantum systems with direct fidelity estimation.

Randomized benchmarking (RB) and variants there of are proposed to avoid both the scaling problem and SPAM errors at the same time \cite{Emerson2005, Emerson2007science, Knill2008pra, Emerson2011prl, Emerson2012pra,10.1145/3408039,PRXQuantum.2.010322}. Standard RB estimates the average error rate of a specific gate set under the assumption of gate-independent or weakly-dependent noise. The gate set is normally chosen to be the Clifford group and has been widely implemented in experiments \cite{Chow2009prlRB, Gaebler2012prlRB, Laflamme2012prl, Barends2014surface, Lu2015prl, Ballance2016prlRB, Gaebler2016RBion, proctor2021scalable}. Otherwise, in order to characterize a specific Clifford gate, a variant called interleaved RB was proposed and utilizes random Clifford gates interleaved with the target gate \cite{Magesan2012interleavedRB}. The random gates here are considered as the twirling gates for reference, whose fidelity should be measured separately to infer the fidelity of the target gate. The interleaved RB method is efficient and scalable in principle. However, it suffers from two severe problems in practice. The first is the compiling overhead for twirling operations. In reality, any operation needs to be compiled to one- and two-qubit gates native to the quantum system. Note that twirling gates are randomly picked from a gate set, like the Clifford group. In general, the average number of native gates used for compiling a single sample grows dramatically with the system size.
The second is the gate-dependent noises introduced by twirling gates. Note that different twirling operations in a gate group can vary a lot in the depths of compiled circuits. For example, a local operation like a Pauli gate can be implemented by a single layer circuit,
while a complex entangling operation requires a deep circuit with massive native gates. The strong gate-dependent noises caused by the uneven compilations may bring inaccuracy to the fidelity estimation \cite{Wallman2018quantum, merkel2021RB}. As a result, the compiling overhead and gate-dependent noises introduced by twirling gates limit the scalability of the RB method in experiments.

Recently, there are several variants of RB attempting to address the two compiling problems. For example, character benchmarking employs the character theory so that the quality parameters can be extracted from the local twirling operations \cite{Helsen2019characterRB}.
Unfortunately, for the gate groups with exponentially increasing number of quality parameters, this method requires an exponential amount of SPAM settings. 
Besides, character benchmarking is still caught in the aforementioned compiling problems for the final inverse gate and can be hardly applied for a generic multi-qubit quantum operation. 
Another inspiring attempt called cycle benchmarking aims to estimate the fidelity of the target gate by interleaving it with the Pauli gate set. However, it is restricted to the Clifford gates~\cite{Erhard2019cycleRB}. 
Also, cycle benchmarking requires numerous repetitions for the gates with large cyclic numbers, which is common for multi-qubit gates. Hence, this method cannot efficiently benchmark a wide class of gates.
The cross-entropy benchmarking (XEB) characterize the fidelity of a generic quantum gate reflected by linear cross-entropy using local Clifford gate twirling~\cite{XEB2019google}.
However, the Haar measure assumption in XEB may lead to poor fidelity estimation when the size of the target gate is large. How to efficiently and reliably estimate the fidelity of a large-scale quantum process from a universal gate set remains an open problem.

In this work, we propose two scalable and efficient protocols to tackle the compiling problems as well as the SPAM error issues simultaneously, which we call character-cycle benchmarking (CCB) and character-average benchmarking (CAB). The protocols utilize local twirling to reliably characterize the fidelity of an individual multi-qubit quantum operation. We employ the Pauli and the local Clifford gates for twirling and extend the applicable gate set to non-Clifford gates via the local gauge transformation. The efficiency and reliability of the protocols are shown by rigorous mathematical derivations and by numerical simulations under realistic physical assumptions.

\section{Character cycle benchmarking}
Denote the quantum operation of a unitary matrix $U$ acting on an $n$-qubit quantum state, $\rho$, by the calligraphic letter $\mathcal{U}$, i.e., $\mathcal{U}(\rho) = U\rho U^{-1}$, and the noisy implementation by $\tilde{\mathcal{U}}$. One can evaluate the quality of $\tilde{\mathcal{U}}$ by the process fidelity of the noise channel $\Lambda = \mathcal{U}^{-1} \circ \tilde{\mathcal{U}}$,
\begin{equation}
F(\Lambda) = \frac{1}{d^2} \sum_{i=0}^{4^n-1} \lambda_i,
\end{equation}
where $\lambda_i = d^{-1} \Tr(P_i \Lambda(P_i))$ is the \textit{Pauli fidelity} associated with the Pauli operator $P_i\in\textsf{P}_n$ and $d = 2^n$ is the dimension of the quantum system. Here, $\textsf{P}_n$ denotes the $n$-qubit Pauli group, containing the tensor product of the identity operation $I$ and three Pauli matrices $X, Y, Z$.

In practice, it is costly to figure out all the parameters $\lambda_0, \lambda_1, \cdots, \lambda_{4^n-1}$ since their number increases exponentially with $n$. Instead, one can estimate the process fidelity via repeatable sampling of $\lambda_0, \lambda_1, \cdots, \lambda_{4^n-1}$. Concretely, one samples a sufficient number of Pauli operators $\{P_j\}$ and averaging the corresponding $\{\lambda_j\}$,
\begin{equation}\label{eq:F_sample}
F(\Lambda) \approx \frac{1}{M} \sum_{\{P_j\}} \lambda_j,
\end{equation}
where $M$ is the number of samples and the summation takes over the sample set.

Here, we propose a CCB protocol which employs the key techniques of the cycle benchmarking \cite{Erhard2019cycleRB} and character benchmarking \cite{Helsen2019characterRB}. Specifically, we extract different Pauli fidelities through applying specific initial states and measurements and utilize the character theory to fully separate the SPAM errors.
The schematic circuit of the CCB protocol is shown in Fig.~\ref{fig:circuit}(a). Let us start with the Clifford case, where the target gate belongs to the $n$-qubit Clifford group $\textsf{C}_n$. The inner random gate layer consists of the target gate $\mathcal{U}$ and its inverse gate $\mathcal{U}^{-1}$ interleaved with two random Pauli gate layers. The Pauli gates are the reference gates employed to perform local Pauli twirling over the generic quantum noise channel $\Lambda$ and turns it into
\begin{equation}
\begin{split}
\Lambda_\textsf{P} &= \frac{1}{4^n}\sum_{P_j \in \textsf{P}_n} \mathcal{P}_j^{-1} \circ \Lambda \circ \mathcal{P}_j,
\end{split}
\end{equation}
where $\Lambda$ contains the errors of Pauli gates and target gate $\mathcal{U}$; $\Lambda_\textsf{P}$ is a Pauli channel satisfying $\Lambda_\textsf{P}(\rho) = \sum_j p_j P_j\rho P_j$, and $p_j$ is the Pauli error rate related to $P_j$.

Note that the introduction of the inverse target gate $\mathcal{U}^{-1}$ is the major difference between CCB and cycle benchmarking. In cycle benchmarking, we need to repeat $\mathcal{U}$ for multiples of $l$ times, where $l$ is the cyclic number of $U$, i.e., $\mathcal{U}^l = I$. In general, $l$ can be quite large for a wide class of Clifford gates which is prohibitive for the experiments. For example, the five-qubit quantum error-correcting encoding circuit requires $l=124$. The CCB protocol improves the efficiency and application scope via substituting a single $\mathcal{U}^{-1}$ for multiple repetitions of $\mathcal{U}$ in cycle benchmarking. In many quantum platforms, such as superconducting quantum processors, the inverse gates of the native gates are also native. Typical examples include single-qubit gates, CZ, and iSWAP. Thus, the inverse gates of native gates are normally easy to implement. More generally, if $\mathcal{U}$ is composed of several native gates, the difficulty to implement $\mathcal{U}$ and $\mathcal{U}^{-1}$ is often the same. Based on this consideration, the introduction of $\mathcal{U}^{-1}$ does not increase the implementation difficulty of CCB in most cases.

In CCB, the randomization of Pauli gates in the inner gate layers will generate a composite channel $\mathcal{U}^{-1} \circ \Lambda_\textsf{P}^{(-)} \circ \mathcal{U} \circ \Lambda_\textsf{P}$, where $\Lambda_\textsf{P}$ and $\Lambda_\textsf{P}^{(-)}$ are the Pauli-twirled channels corresponding to gates $\mathcal{U}$ and $\mathcal{U}^{-1}$, respectively. Note that this composite channel is a Pauli channel for Clifford gate $\mathcal{U}$. The fidelity we aim to estimate in the CCB protocol is defined as the CCB fidelity,
\begin{equation}\label{eq:Fccb0}
F_{\mathrm{ccb}} = F(\sqrt{\mathcal{U}^{-1} \circ \Lambda_\textsf{P}^{(-)} \circ \mathcal{U} \circ \Lambda_\textsf{P}}),
\end{equation}
which contains the fidelities of $\mathcal{U}$ and $\mathcal{U}^{-1}$. For the case that the noise channel of $\mathcal{U}$ is the same as that of $\mathcal{U}^{-1}$, which is valid for most of the experimental platforms, the CCB fidelity is simplifies as
\begin{equation}\label{eq:Fccb}
F_{\mathrm{ccb}}(\Lambda) = F(\sqrt{\mathcal{U}^{-1} \circ \Lambda_\textsf{P} \circ \mathcal{U} \circ \Lambda_\textsf{P}}).
\end{equation}
Eq.~\eqref{eq:Fccb} is a lower bound of the process fidelity $F(\Lambda)$ in terms of the expectation value, as proved in Appendix \ref{append:CCBfidelity}. In the following context, we will employ Eq.~\eqref{eq:Fccb} as our CCB fidelity metric model and our arguments apply to the general model of Eq.~\eqref{eq:Fccb0} as well. The difference between $F_{\mathrm{ccb}}$ and $F$ is normally small since the physical realizations of the qubits in one experimental platform are similar and the qualities of these qubits will not differ too much. Note that if $\Lambda_\textsf{P}$ is a depolarizing channel, then $F_{\mathrm{ccb}}=F$. Thus, the CCB fidelity can be seen as a reliable metric for the noise channel $\Lambda$.

\begin{figure}[htbp!]
\centering
\includegraphics[width=1\textwidth]{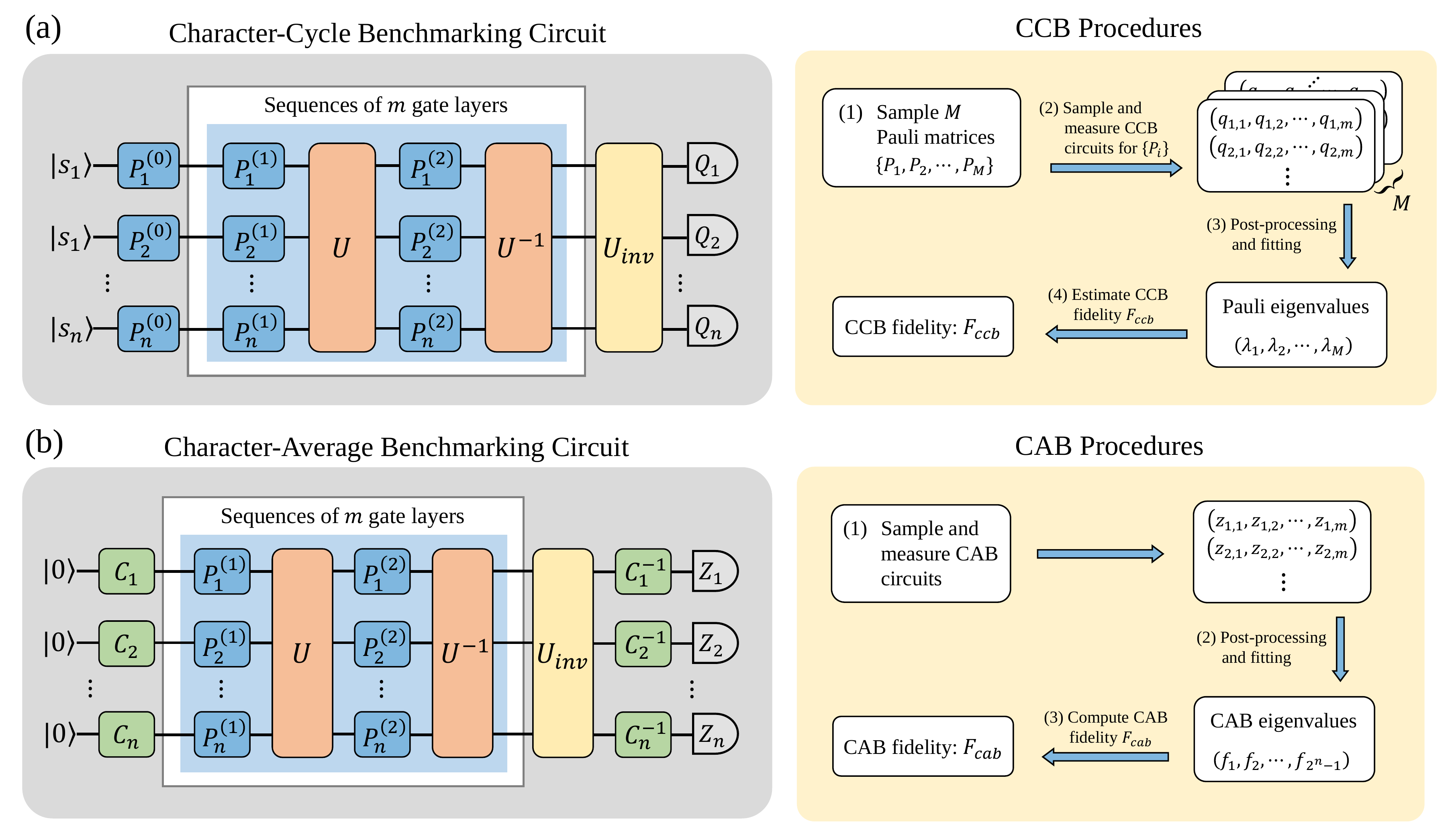}
\caption{Illustrations of circuit and procedures used in (a) CCB and (b) CAB protocols. The orange boxes represent the target gate $\mathcal{U}$ and its inverse gate $\mathcal{U}^{-1}$. The blue and green boxes represent the random Pauli gate and random local Clifford gate. The yellow boxes denote the inverse gate for the $m$ inner gate layers in the light blue box. Here, $P^{(i)}_k$ is a single-qubit Pauli gate on qubit $k$ and $P^{(i)} = P^{(i)}_1 \otimes \cdots P^{(i)}_n$ is a $n$-qubit Pauli gate.}
\label{fig:circuit}
\end{figure}

The procedure of the CCB protocol runs as follows:
\begin{enumerate}
    \item sample a Pauli operator $P_j$ and initialize state $\ket{s}$ such that $P_j\ket{s} = \ket{s}$;
    
    \item apply a gate sequence $\mathcal{S}_\mathrm{ccb}$ composed of a Pauli gate $\mathcal{P}^{(0)}$, $m$ inner gate layers denoted by $\mathcal{S}_m$, and inverse gate $\mathcal{S}_m^{-1}$;
    
    \item perform measurement $P_j$ and then calculate the $P_j$-weighted survival probability $f_j(m, \mathcal{S}_\mathrm{ccb}) = \chi_j(P^{(0)})\Tr(P_j \mathcal{S}_\mathrm{ccb}(\rho_s))$, where $\chi_j(P^{(0)}) = 1$ if $P_j$ commutes with $P^{(0)}$ and $-1$ otherwise;
    
    \item repeat steps (2)-(3) for several times for different $m$ and fit the $P_j$-weighted fidelity to $\hat{f}_j(m) = A_j\lambda_j^{2m}$;
    
    \item repeat steps (1)-(4) for several times and finally estimate the CCB fidelity as $F_{\mathrm{ccb}} = \mathrm{ave}_j \lambda_j$.
\end{enumerate}
Here, the estimated fidelity $F_\mathrm{ccb}$ includes the errors from the local reference gate set $\textsf{P}_n$. In order to remove these extra errors, one can employ the interleaved RB technique, by performing additional CCB with a target gate of identity $I$ to estimate the reference fidelity $F_\mathrm{ccb}^I$. Then, one can infer the fidelity of the target gate as $F_\mathrm{ccb} / F_\mathrm{ccb}^I$. In practice, the errors of local gates are often negligible and hence we focus on $F_\mathrm{ccb}$ in the following discussions.

Note that our inverse gate $\mathcal{S}_m^{-1}$ is a Pauli gate and hence will not introduce extra gate compiling overhead. As a contrast, character benchmarking for a single multi-qubit Clifford gate~\cite{Helsen2019characterRB} requires a global inverse gate and a complicated compiling process. This may cause strong gate-dependent errors and lead to inaccuracy for fidelity estimation, especially for multi-qubit quantum operations.
The CCB protocol maintains the local structure of reference and inverse gates and thus avoids the compiling problems.

In the CCB protocol, one needs to average Pauli fidelities $\lambda_j$ to estimate $F_{\mathrm{ccb}}$. The sampling complexity for the CCB protocol is given by the following theorem.
\begin{theorem}[informal version]\label{thm:CCB_sample}
For an $n$-qubit quantum noise channel, in order to estimate the CCB fidelity within the confidence interval $[\hat{F}_\mathrm{ccb} - \epsilon_M - \epsilon_b, \hat{F}_\mathrm{ccb} + \epsilon_M + \epsilon_b]$ with probability greater than $1-\delta$, one needs to sample $M$
Pauli fidelities where each Pauli fidelity is estimated via $K$ random sequences. 
The confidence probability of the estimation is given by,
\begin{equation}
    \mathrm{Pr}(|\hat{F}_\mathrm{ccb} - \bar{F}_\mathrm{ccb}| \leq \epsilon_M + \epsilon_b) \geq 1 - \delta,
\end{equation}
where $\epsilon_M \leq \mathcal{O}(\frac{-\log\delta}{M})$ and $\epsilon_b \leq \mathcal{O}(K^{-1})+\mathcal{O}((\frac{-\log\delta}{K})^{3/2})$.
\end{theorem}
Here, the total number of samples, or sample complexity, depends on $M$ and $K$. If the number of random sequences for each Pauli fidelity is the same, then the sample complexity is simply given by $MK$. Theorem \ref{thm:CCB_sample} shows that the sample complexity only depends on fidelity precision $\epsilon_M, \epsilon_b$ and confidence level $\delta$. The independence on system size $n$ reflects the strong scalability of the CCB protocol. A more detailed description of the result is shown in Theorem \ref{thm:M_number}.

\section{Local gauge transformation}
Now, let us extend the applicable gates for the CCB protocol to non-Clifford gates. One can introduce local gauge transformation $L$ to the twirling gate set, $\textsf{P}_n \rightarrow L \textsf{P}_n L^{-1}$, where $L$ is an arbitrary local unitary operation $L \in \textsf{U}(2)^{\otimes n}$. Note that the transformed twirling gate set $L\textsf{P}_n L^{-1}$ is still local. Then, we can show that the applicable target gate set becomes $L \textsf{C}_n L^{-1}$, where $\textsf{C}_n$ is the $n$-qubit Clifford gate set.

To benchmark a gate $LUL^{-1}$ from gate group $L{\sf{C}}_nL^{-1}$, we insert local gates $L$ and $L^{-1}$ between the twirling gates and the target gates in the original CCB circuit, as shown in Fig.~\ref{fig:gaugefreedom}(a). Here, $L=\bigotimes_{i=1}^n L_i$, where $L_i$ can be an arbitrary single-qubit gate. In practice, the local gates are absorbed into twirling gates and target gates and do not need to be implemented individually as manifested in Fig.~\ref{fig:gaugefreedom}(b). The character gate $LP^{(0)}$ and the twirling gate $LP^{(1)}L^{-1}$ will be merged into a single gate in implementation as well. Details of the derivation are shown in Appendix \ref{app:LocalGauge}.

\begin{figure}[htbp!]
	\centering
	\includegraphics[width=1.0\textwidth]{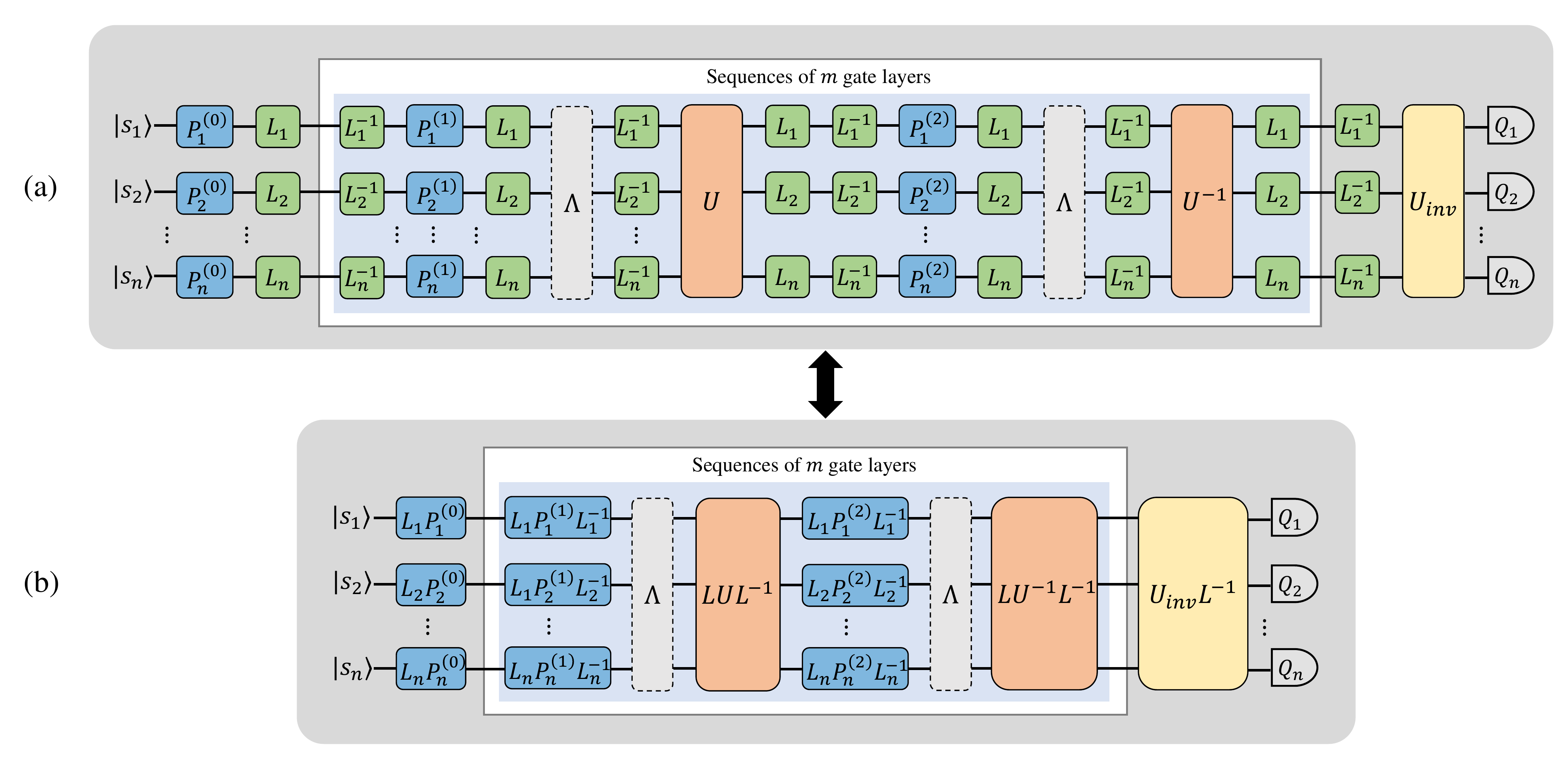}
	\caption{Illustrations of the noisy CCB circuit with local gauge transformation. For simplicity, we show the case that the local gates are noiseless. The grey dashed boxes denote the noise channel $\Lambda$. The orange boxes represent the target gate $\mathcal{U}$ and its inverse gate $\mathcal{U}^{-1}$. The blue boxes represent the random Pauli gates. The green boxes represent the inserted local gates $L$ and $L^{-1}$, where $L=\bigotimes_{i=1}^n L_i$. The yellow box denotes the inverse gate for the $m$ inner gate layers in the light blue boxes. In practice, we implement gates in circuit (b) while absorbing local gates $L$ and $L^{-1}$ into twirling gates and target gates. Here, the target gate after gauge transformation becomes $LUL^{-1}$. Note that Circuit (a) is equivalent to a CCB circuit with target gate $U$ and noise channel $\mathcal{L}^{-1}\Lambda \mathcal{L}$. Thus, it can be implemented to estimate $F_{\mathrm{ccb}}(\mathcal{L}^{-1}\Lambda \mathcal{L})$, which is close to $F(\Lambda)$.}
	\label{fig:gaugefreedom}
\end{figure}



As shown in Fig.~\ref{fig:gaugefreedom}, the CCB circuit with local gauge transformation $L$ and noise channel $\Lambda$ is equivalent to the original CCB circuit with noise channel $\mathcal{L}^{-1}\Lambda \mathcal{L}$. Thus, one can obtain $F_{\mathrm{ccb}}(\mathcal{L}^{-1}\Lambda \mathcal{L})$, which is close to the process fidelity $F(\mathcal{L}^{-1}\Lambda \mathcal{L})$. As process fidelity is gauge-invariant, that is, $F(\mathcal{L}^{-1}\Lambda \mathcal{L}) = F(\Lambda)$, one can estimate the process fidelity of $\Lambda$ as the performance indicator of gate $LUL^{-1}$.

Now, let us check out what kinds of quantum gates belong to the set $\mathcal{S} = \{LUL^{-1}|L\in \textsf{U}(2)^{\otimes n}, U\in \textsf{C}_n\}$.

First, notice that if a unitary $U\in \mathcal{S}$, then for any $L\in \textsf{U}(2)^{\otimes n}$, $LUL^{-1}\in \mathcal{S}$. As any unitary is generated by a Hamiltonian, that is $U = e^{i H}$ where $H$ is hermitian, one can conclude that if $e^{i H}\in \mathcal{S}$, then for any $L\in \textsf{U}(2)^{\otimes n}$, $Le^{i H}L^{-1} = e^{i LHL^{-1}}\in \mathcal{S}$.

Take a step forward, if a controlled-$e^{i H}\in \mathcal{S}$, then through local gauge transformation $I\otimes L$, $(I\otimes L) \text{controlled-}e^{i H} (I\otimes L)^{-1} = \text{controlled-}e^{i LHL^{-1}}\in \mathcal{S}$. The arguments also apply to the case of multi-controlled gates.

The two observations inspire us to first represent Clifford gates in the form of $e^{iH}$ or multiple controlled-$e^{iH}$, then replace $H$ with $LHL^{-1}$ to find other gates in $\mathcal{S}$. Take $CZ$ as an example. $CZ = e^{i\pi \ketbra{11}}$ = controlled-$e^{-i\frac{\pi}{2} Z}$. Through local gauge transformation, one can transform $\ket{11}$ to any product state $\ket{\psi\phi}$ and transform $e^{-i\frac{\pi}{2} Z}$ to any $\pi$-rotation $e^{-i\frac{\pi}{2} \Vec{\sigma}\cdot \Vec{\theta}}$, where $\Vec{\sigma} = (X,Y,Z)$ and $\Vec{\theta}$ is a unit vector. Thus, for any two-qubit product state $\ket{\psi\phi}$, we have $e^{-i\pi \ketbra{\psi\phi}}\in \mathcal{S}$. Also, any controlled-$\pi$ rotation, such as controlled-$H$ and controlled-$TX$, belongs to $\mathcal{S}$.

Reversely, controlled-$S$ = controlled-$e^{-i\frac{\pi}{4} Z}$ is a controlled-$\frac{\pi}{2}$ rotation. As any controlled-$\frac{\pi}{2}$ rotation is not Clifford, one can conclude that controlled-$S$ does not belong to $\mathcal{S}$. Similarly, Tofolli = controlled-controlled-$e^{-i\frac{\pi}{2} X}$ is a controlled-controlled-$\pi$ rotation. As any controlled-controlled-$\pi$ rotation is not Clifford, one can conclude that Tofolli does not belong to $\mathcal{S}$ either. It is an interesting question to decide whether a quantum gate belongs to $\mathcal{S}$ in a more general case and we leave it for future work.

\section{character-average benchmarking}
We can take the CCB protocol one step further. Observe that in CCB, one needs to implement the fitting procedures for each sampled Pauli operator $P_j$ to estimate Pauli fidelity $\lambda_j$. Each estimation requires specific initial state, measurement, and independent randomization procedures. We can further simplify these procedures by introducing the local Clifford group $\textsf{C}_1^{\otimes n}$. Recall that in a qubit system, the Clifford twirling depolarizes a channel via averaging the error rates in $X, Y, Z$ bases \cite{Emerson2011prl}. Then for an $n$-qubit system, the twirling over $\textsf{C}_1^{\otimes n}$ would partially depolarize a channel and average out $4^n$ Pauli fidelities $\lambda_j$ into $2^n$ terms. These $2^n$ values can be obtained from the $Z^{\otimes n}$ basis measurement only with additional data post-processing.

Based on the local Clifford twirling, we propose the CAB protocol as an improvement of the CCB protocol. The schematic circuit of CAB is shown in Fig.~\ref{fig:circuit}(b), with the detailed procedures described in Box \ref{box:CAB}. Like the CCB protocol, we can extend the target gate set beyond the Clifford group by employing local gauge transformation. Here, in order to suppress statistical fluctuations, we remove the character technique. Detailed description and analysis of the CCB and CAB protocols are presented in Appendix \ref{sec:CCBandCAB}.

\begin{mybox}[label={box:CAB}]{{Procedures for character-average benchmarking}}
\begin{enumerate}
\item
Sample a gate sequence $(C, P^{(1)}, P^{(2)}, \cdots, P^{(2m)})$, where $C$ and $P^{(i)} (1 \leq i \leq 2m)$ are sampled uniformly at random from the local groups, $\textsf{C}_1^{\otimes n}, \textsf{P}_n$, respectively.

\item
Initialize the state $\ket{\psi} = \ket{0}^{\otimes n}$ and apply the gate sequence as shown in  Fig.~\ref{fig:circuit}(b),
\begin{equation}\label{eq:sequence1}
\mathcal{S}_\mathrm{cab} = \mathcal{C}^{-1} \circ \mathcal{U}_{\mathrm{inv}} \circ \mathcal{U}^{-1} \circ \mathcal{P}^{(2m)} \circ \cdots \circ \mathcal{U} \circ \mathcal{P}^{(1)} \circ \mathcal{C},
\end{equation}
where the inverse gate $\mathcal{U}_{\mathrm{inv}} = \mathcal{P}^{(1)} \circ \mathcal{U}^{-1} \circ\cdots\circ \mathcal{P}^{(2m)}\circ\mathcal{U}$ is a local gate as well.

\item Measure in $Z^{\otimes n}$ basis and compute the survival probability for each measurement observable $Q_k \in \{I, Z\}^{\otimes n}$,
\begin{equation}
f_k(m, \mathcal{S}_\mathrm{cab}) = \Tr[Q_k \mathcal{S}_\mathrm{cab}(\rho_\psi)],
\end{equation}
where $\rho_\psi$ is the noisy preparation of the initial state $\ket{\psi}$.

\item
Repeat for a sufficient number of sequences and estimate the average value
\begin{equation}\label{eq:CABfm}
f_k(m) = \mathop{\mathbb{E}}_{\mathcal{S}_\mathrm{cab}} f_k(m, \mathcal{S}_{\mathrm{cab}}).
\end{equation}

\item Repeat for different $m$ and fit to the function
\begin{equation}\label{eq:CABfitting_model}
f_k(m) = A_k \mu_k^{2m},
\end{equation}
where $A_k$ and $\mu_k$ are fitting parameters.

\item Estimate the CAB fidelity
\begin{equation}\label{eq:Fcab}
F_{\mathrm{cab}} = \frac{1}{d^2}\sum_k d_k \mu_k,
\end{equation}
where $d_k = 3^{\pi(Q_k)}$, $\pi(Q_k)$ is the number of $Z$ in $Q_k$.
\end{enumerate}
\end{mybox}

Similar to the CCB protocol, the randomization over $m$ gate layers inside the blue box in Fig.~\ref{fig:circuit}(b) will generate a Pauli channel $\Lambda_\textsf{P}(m) = (\mathcal{U}^{-1} \circ \Lambda_\textsf{P} \circ \mathcal{U} \circ \Lambda_\textsf{P})^m$. The local Clifford gates in the beginning and end of the circuit jointly perform local unitary 2-design twirling, which transforms the Pauli channel $\Lambda_\textsf{P}(m)$ into a partially depolarizing channel $\Lambda_\textsf{C}(m)$. Here, the quantum channel $\Lambda_\textsf{C}(m)$ contains less independent parameters than the original $\Lambda_P(m)$. It holds the unique value of fidelity for every disjoint Pauli subset in $R^n=\{\{I\},\{X,Y,Z\}\}^{\otimes n}$.
The Pauli fidelities in $\Lambda_\textsf{C}(m)$ can be seen as the average values of those in $\Lambda_\textsf{P}(m)$, $\mu_k^{2m} = \sum_{P_j\in \sigma_k} \lambda_j^{2m} / |\sigma_k|$, where $\{\lambda_j^{2m}\}$ are the Pauli fidelities of the channel $\Lambda_\textsf{P}(m)$. The local Clifford twirling here averages
multi exponential decays into one exponential decay and captures all the information of the noise channel, 
as shown in Eq.~\eqref{eq:Fcab}. The comparison between $F_{\mathrm{cab}}$ and $F_{\mathrm{ccb}}$ is shown in Lemma 3 in Appendix \ref{append:CABfidelity}. While the CCB protocol employs a sampling method as in Eq.~\eqref{eq:F_sample}, which only contains partial information of the noise channel. Thus, one can intuitively conclude that the CAB protocol is more efficient than the CCB protocol, as demonstrated in later simulations.


\section{Simulation}
In numerical simulations, we characterize a two-qubit controlled-($TX$) gate and a five-qubit quantum error correcting encoding circuit, respectively. We simulate the noise channel for the target gate with a realistic error model that contains: a Pauli channel, an amplitude damping channel, and a correlation channel. In the simulation, the Pauli fidelities of the Pauli channel are randomly sampled from a normal distribution $\mathcal{N}(\mu, \sigma)$, which we call the $\mathcal{N}(\mu, \sigma)$-Pauli channel. Here, the error parameter $\mu$ reflects the quality of the Pauli channel and $\sigma$ implies the discrepancy of the channel, i.e., the differences among Pauli fidelities. The detailed descriptions for the error models and simulations are presented in Appendix \ref{sec:simulation}.

For the controlled-$(TX)$ gate, we take $(I \otimes \sqrt{T}) \textsf{P}_n (I \otimes \sqrt{T}^{-1})$ as the twirling gate set, where $T$ is the $\pi/8$-phase gate, $T=\exp(-i\pi Z/8)$, and $I \otimes \sqrt{T}$ is the local gauge transformation. We simulate the CAB and CCB protocols on the controlled-(TX) gate with 8 different noise channels.
For each noise channel, we take 40 independent simulations for both CAB and CCB protocols. In CCB simulations, we sample $M = 10$ Pauli operators to estimate $F_\mathrm{ccb}$.

Figure \ref{fig:CTXa} shows $F_\mathrm{cab}$ and $F_\mathrm{ccb}$ versus the error rate $r = 1 - F$ for the   controlled-$(TX)$ gate with different noise channels. We observe that when the standard deviation of error parameters $\sigma$ grows, the error bars of $F_\mathrm{cab}$ and $F_\mathrm{ccb}$ become larger. Intuitively, the discrepancy of the Pauli fidelities is one of the key reasons for the fluctuations of $F_\mathrm{cab}$ and $F_\mathrm{ccb}$. The fluctuations for the estimations will reach the  minimum level when the noise channel is completely depolarizing. Besides, the error bar of CAB is smaller than the error bar of CCB. This shows that under the same estimation accuracy, the \textit{sampling complexity}, i.e., the amount of sampling sequences in total, of the CAB protocol is smaller than that of the CCB protocol, especially when the discrepancy of the noise channel is large. In Fig.~\ref{fig:CTXb}, we take one of the 8 noise channels as an example and show the three fitting curves of Eq.~\eqref{eq:CABfitting_model} for the CAB protocol. The resulted CAB fidelity is $F_\mathrm{cab} = 95.99\%$, which is very close to the theoretical value of process fidelity $F = 95.98\%$.

\begin{figure}[htbp!]
\centering

\subfigure[]{
	\centering
	\includegraphics[width=3.4in]{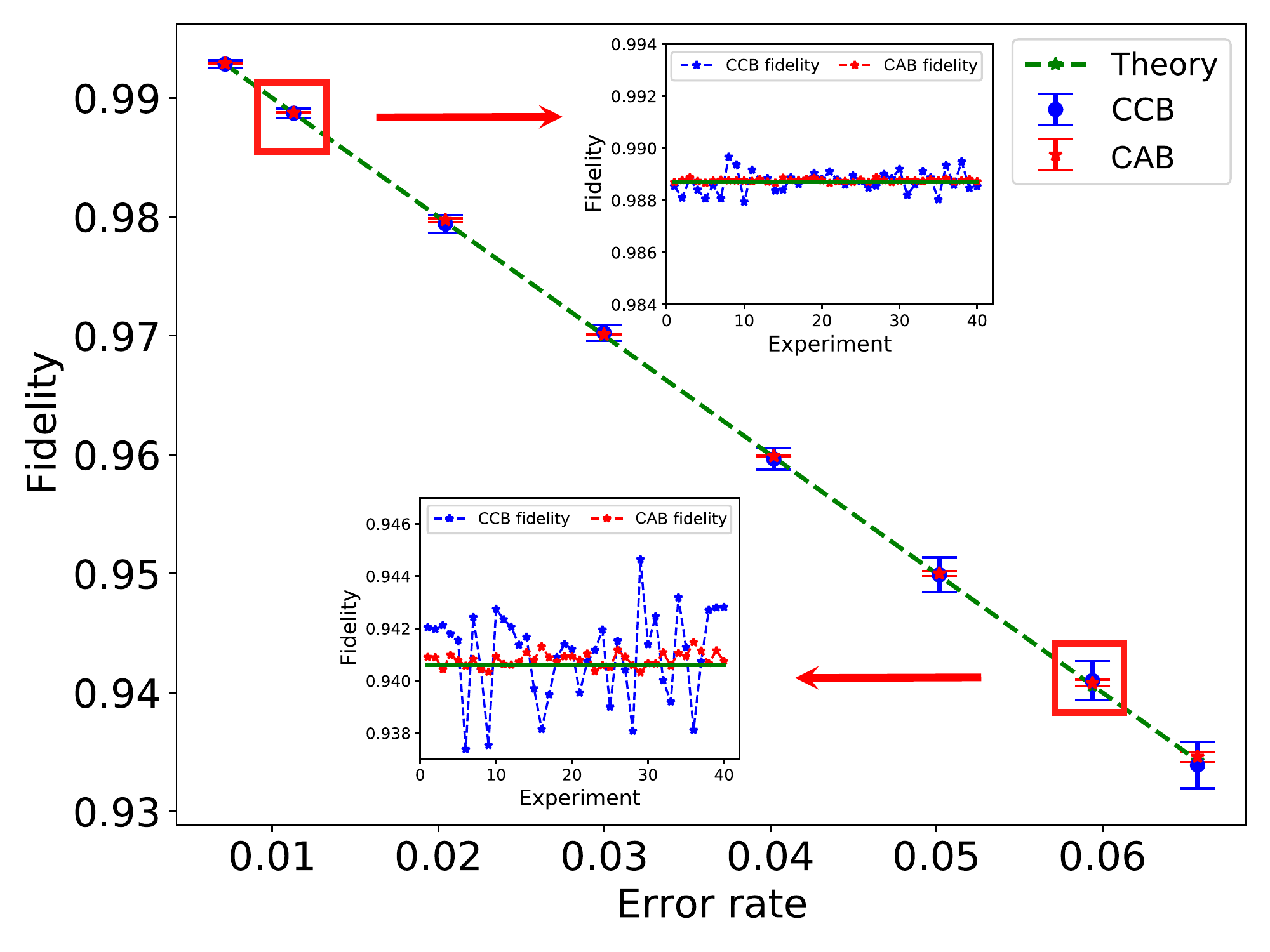}
	\label{fig:CTXa}
}
\subfigure[]{
	\centering
	\includegraphics[width=3.4in]{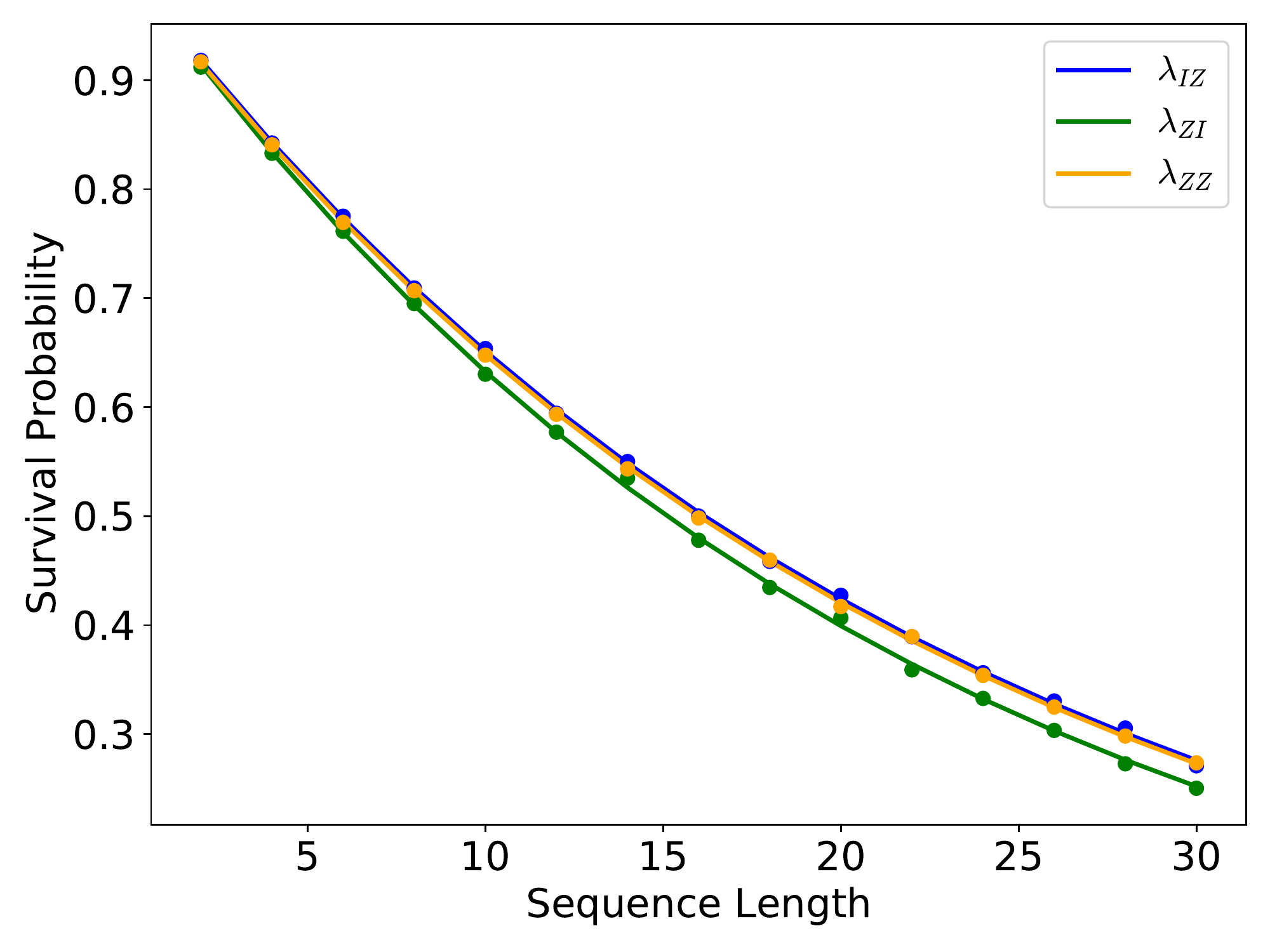}
	\label{fig:CTXb}
}
\caption{Simulation results for the controlled-$(TX)$ gate with 8 different noise channels, each containing a $\mathcal{N}(\mu, \sigma)$-Pauli channel, an amplitude damping channel, and a correlation channel. The parameters for the amplitude damping channels and correlation channels set to be the same for the 8 channels. For the Pauli channel, the error parameters are set to $\{(\mu, \sigma)\} =$ \{(0.995, 0.001), (0.990, 0.002), (0.980, 0.003), (0.970, 0.004), (0.960, 0.005), (0.950, 0.006), (0.940, 0.007), (0.930, 0.008)\}. (a) The fidelity estimations with different error rates in 40 independent simulations. The green dashed line represents the theoretical fidelities. The two insert scatter plots show the fluctuations of estimated fidelities over different simulations. (b) Take the fifth simulation with process fidelity of $95.98\%$ as an example. The fitting curves of Eq.~\eqref{eq:CABfitting_model} for $Q_k=IZ, ZI, ZZ$. The decay parameters derived from the curves are $\lambda_{IZ} = 0.9580, \lambda_{ZI} = 0.9550, \lambda_{ZZ} = 0.9577$.}
\label{fig:controlled-TX}
\end{figure}

For the 5-qubit error correcting encoding circuit, which is a Clifford gate, we take the Pauli group $\textsf{P}_n$ as the twirling gate set. For the simplicity of simulation, we set the Pauli channel to be a depolarizing channel $\Lambda_\mathrm{dep}(\rho) = p\rho + (1-p)I/d$ where $p = 0.98$. The setting of the amplitude damping and correlation channels remain the same. We simulate the CAB and XEB protocols to characterize the noisy 5-qubit encoding circuit. For each protocol, we run 40 independent simulations. In each simulation, we take the sampling number of gate sequences as $K = 50, \cdots, 500$ for each sequence length $m$. The box plot of $F_\mathrm{cab}$ versus $K$ is shown in Fig.~\ref{fig:QEC-a}. We can see that when $K$ grows, the fluctuations of $F_\mathrm{cab}$ become smaller. When $K$ is not too large, like $K = 50$, the fluctuation is already small enough, which implies that the CAB protocol works well with few sampling sequences needed.

\begin{figure}[htbp!]
\centering

\subfigure[]{
\centering
\includegraphics[width=3.4in]{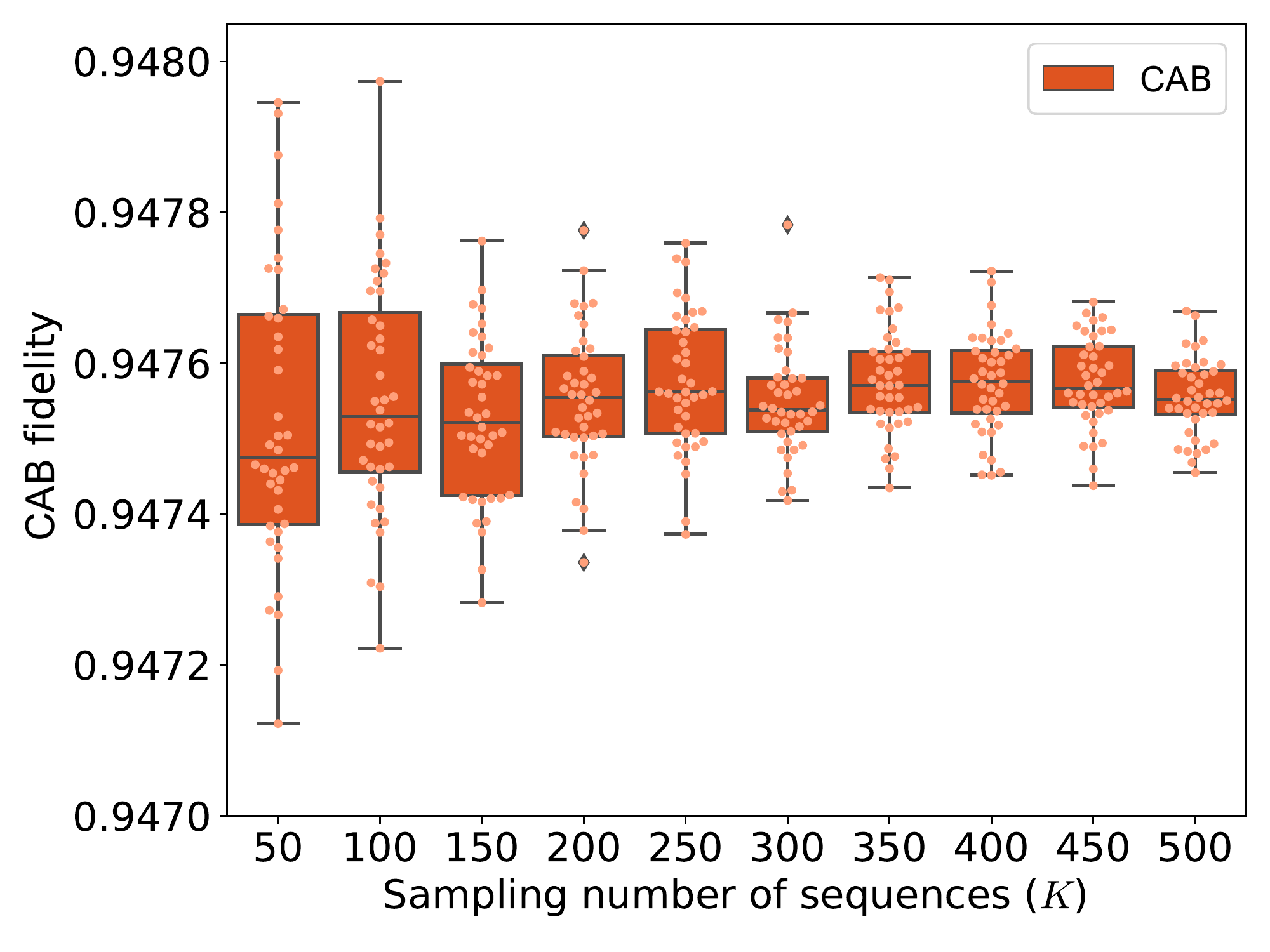}
\label{fig:QEC-a}
}
\subfigure[]{
\centering
\includegraphics[width=3.4in]{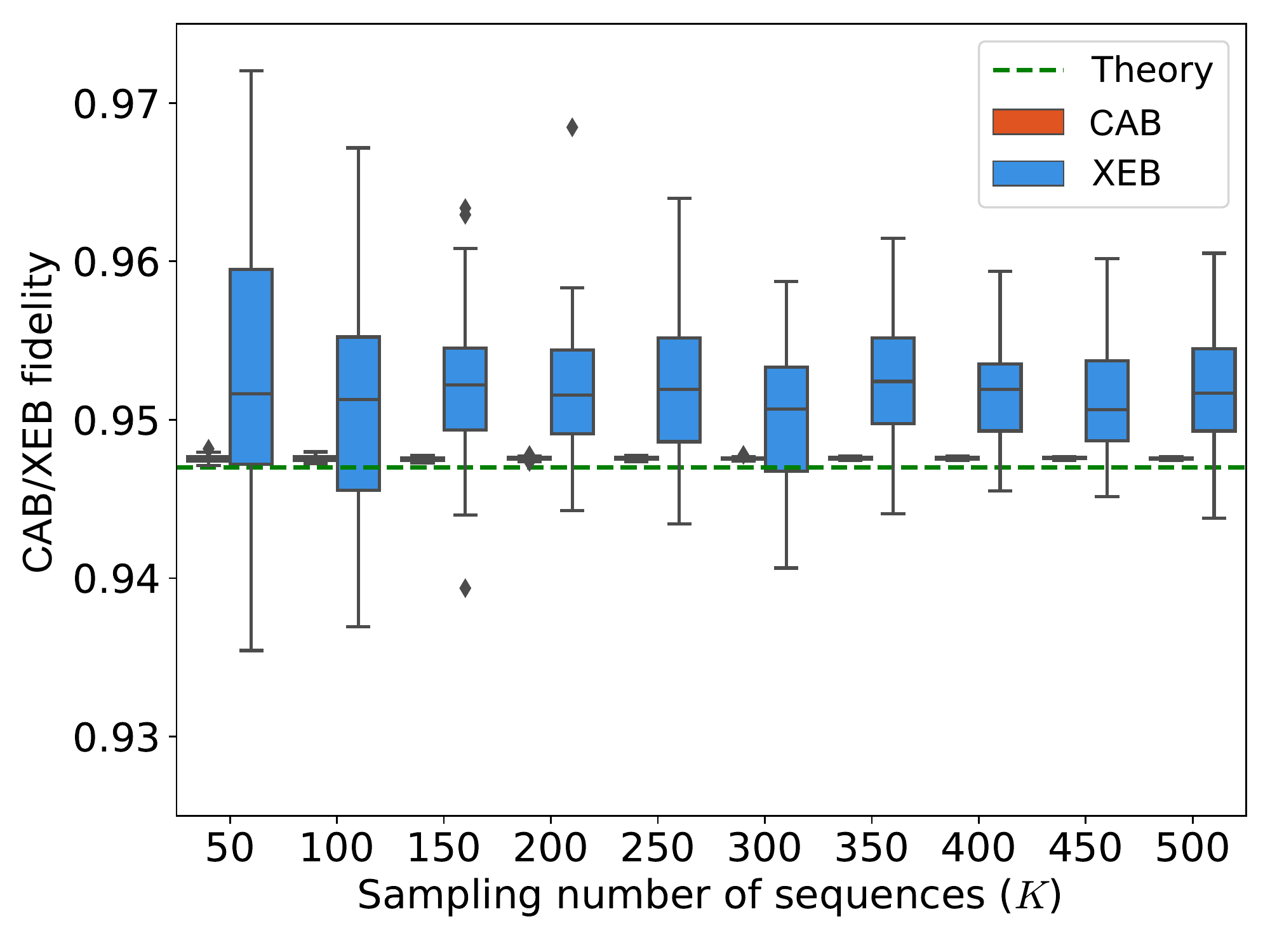}
\label{fig:QEC-b}
}
\caption{Simulation results for the 5-qubit quantum error correcting encoding circuit with a noise channel composed of a depolarizing channel $\Lambda_\mathrm{dep}(\rho) = p\rho + (1-p)I/d$ where $p = 0.98$, an amplitute damping channel, and a correlation channel. The theoretical process fidelity is $F = 94.70\%$. (a) Box plot of the CAB fidelities versus sampling number $K$ with 40 independent simulations. The red boxes represent the distributions of the CAB fidelity estimations with respect to $K$. The orange points represent the distributions of CAB fidelities. (b) Box plots of the CAB and XEB fidelities versus $K$ with 40 independent simulations. The green dashed line represents the theoretical process fidelity. The plots of CAB in (a) and (b) are the same with different scaling.}
\label{fig:QEC}
\end{figure}

In Fig.~\ref{fig:QEC-b}, we show the box plots of $F_\mathrm{cab}$ and XEB fidelities $F_\mathrm{xeb}$ versus the sampling number $K$. It is clear to see that compared with $F_\mathrm{xeb}$, $F_\mathrm{cab}$ is much closer to the theoretical process fidelity $F = 94.7\%$. Meanwhile, the convergence of $F_\mathrm{cab}$ is much better than that of $F_\mathrm{xeb}$. This implies that the required $K$ for CAB is much smaller than that of XEB under the same estimation accuracy.

To give a concrete example, we take 20 CAB simulations and 20 XEB simulations under the same noise channel. From the simulation results, we find that for CAB, when $K_\mathrm{cab} = 20$, the standard deviation over the 20 simulations is $\sigma_\mathrm{cab} = 3.25 \times 10^{-4}$; while for XEB, when $K_\mathrm{xeb} = 20000$, the standard deviation is $\sigma_\mathrm{xeb} = 4.29 \times 10^{-4}$. This shows that to estimate the fidelities with standard deviations around $4\times 10^{-4}$, the required $K_\mathrm{xeb}$ is over 1000 times larger than $K_\mathrm{cab}$. Thus, we can conclude that the performance of CAB protocol is three orders of magnitudes better than that of XEB protocol in terms of the sampling complexity.

The simulation results reveal the strong scalability and reliability of our protocols, especially the CAB protocol. The fluctuation of estimated CAB fidelity is small even when for multi-qubit gates. We believe the CAB protocols can provide fast feedback in experimental designs and promote the development of universal fault-tolerant quantum computing. 

\section{Conclusion and discussion}
Characterization of large-scale individual quantum processes is crucial to the development of near-term quantum devices. However, there does not exist scalable and practical methods that can benchmark multi-qubit universal gate-set currently. In this work, we propose and demonstrated efficient and scalable randomized benchmarking protocols --- CCB and CAB that can individually characterize a wide class of quantum gates including and beyond the Clifford set. The key technique of ours protocols is using the local reference gate-set for twirling, which avoid the inaccuracy of the estimation caused by gate-compiling overhead and gate-dependent noises. The method of local gauge transformation offers a tool for characterizing non-Clifford gates. The sampling and measurement complexity are independent of the qubit number of gate, which means our benchmarking protocols can be generalized to large-scale quantum systems.

Our protocols maintain the simplicity and robustness of the conventional RB method, and estimate the quantity of most interest --- process fidelity of the target gates. We believe our protocols will promote the development of universal fault-tolerant quantum computing. Furthermore, it would also be interesting to extend our randomization and estimation methods for characterizing other properties like unitarity and coherence, which we leave for future research.


\section{Acknowledgement}
We acknowledge B.~Chen for the insightful discussions. This work was supported by the National Natural Science Foundation of China Grants No.~11875173 and No.~12174216 and the National Key Research and Development Program of China Grants No.~2019QY0702 and No.~2017YFA0303903.

\appendix
\section{Preliminaries}\label{sec:Preliminary}
\subsection{Representation theory}\label{sec:representation}
The representation theory works as a general analysis of every representation for abstract groups.
Informally, the representations of a group can reflect its block-diagonal structures. Let $\textsf{G}$ be a finite group and $g \in \textsf{G}$ be a group element. The representation of $\textsf{G}$ is defined as follows.

\begin{definition}[Group representation] \label{append:definition_group}
Map $\phi$ is said to be a representation of group $\textsf{G}$ on a linear space $V$ if it is a group homomorphism from $\textsf{G}$ to $GL(V)$,
\begin{equation}
	\begin{split}
		&\phi:{\sf{G}} \rightarrow GL(V), \\
		&g \mapsto \phi(g), \ \forall g \in {\sf{G}};
	\end{split}
\end{equation}
where $GL(V)$ is the general linear group of $V$, such that $\forall g_1, g_2 \in {\sf{G}}$,
\begin{equation}
	\phi(g_1)\phi(g_2) = \phi(g_1 g_2).
\end{equation}
\end{definition}

Given representation $\phi$ on $V$, a linear subspace $W \subseteq V$ is called \textit{invariant} if $\forall w \in W$ and $\forall g \in \textsf{G}$,
\begin{equation}
	\phi(g)w \in W.
\end{equation}
The restriction of $\phi$ to the invariant subspace $W$ is known as a \textit{subrepresentation} of $\textsf{G}$ on $W$. One can further define the \textit{irreducible representation} (or irrep for short) as follows.

\begin{definition}[Irreducible representation]
	Representation $\phi$ of group $\textsf{G}$ on linear space $V$ is irreducible if it merely has trivial subrepresentations, i.e., the invariant subspaces for $V$ are only $\{0\}$ and $V$ itself.
\end{definition}

The Maschke's theorem provides an interesting property that each representation $\phi$ of a finite group $\textsf{G}$ can be decomposed to the irreducible representations, $\forall g \in \textsf{G}$,
\begin{gather}\label{eq:directsum}
	\phi(g) \simeq \bigoplus_{\sigma\in R_{\textsf{G}}} \sigma(g)^{m_\sigma},
\end{gather}
where $R_{\textsf{G}} = \{\sigma\}$ denotes the set of all the irreps of representation $\phi$ and $m_\sigma$ is the \textit{multiplicity} of the equivalent irreps of $\sigma$. In this paper, we will focus on the non-degenerate representation case, i.e., $m_\sigma = 1,\ \forall \sigma$.

\begin{definition}[Character function]
	Let $\sigma$ be a representation over group ${\sf{G}}$, the character of $\sigma$ is the function $\chi_\sigma: {\sf{G}} \rightarrow \mathbb{C}$ given by $\forall g \in {\sf{G}}$,
	\begin{equation}
		\chi_\sigma(g) = \Tr[\sigma(g)].
	\end{equation}
\end{definition}

With the character function, we introduce the \textit{generalized projection formula} used in character randomized benchmarking \cite{Helsen2019characterRB}.

\begin{lemma} [Generalized projection formula \cite{fultonRepresentation}] \label{lemma:project_formula}
	Given a finite group, ${\sf{G}}$, and its representation, $\phi$, denote $\sigma$ to be an irreducible representation contained in $\phi$ with its character function $\chi_\sigma: {\sf{G}} \rightarrow \mathbb{C}$. The projector onto the support space of $\sigma$ can be written as,
	\begin{equation}\label{eq:project_formula}
		\Pi_\sigma = \frac{d_\sigma}{|\textsf{G}|}\sum_{g \in \textsf{G}} \chi_\sigma(g) \phi(g),
	\end{equation}
	where $d_\sigma$ is the dimension of $\sigma$, $d_\sigma = \chi(e)$ with $e$ being the identity element in ${\sf{G}}$.
\end{lemma}

Next, we will introduce twirling over a group $\textsf{G}$.

\begin{definition} [Twirling]
	For representation $\phi$ of group ${\sf{G}}$ on linear space $V$, a random twirling for a linear map $\Lambda: V \rightarrow V$ over ${\sf{G}}$ is defined as
	\begin{equation}
		\Lambda_{{\sf{G}}} = \frac{1}{\abs{{\sf{G}}}} \sum_{g\in{\sf{G}}} \phi(g)^\dagger\Lambda\phi(g).
	\end{equation}
\end{definition}

Using Schur's Lemma \cite{fultonRepresentation}, one can show the following proposition.

\begin{proposition}\label{prop:twirling}
	For any linear map $\Lambda: V \rightarrow V$, the twirling over group ${\sf{G}}$ and its representation $\phi$ can be written as
	\begin{equation}\label{eq:twirling2}
		\Lambda_{{\sf{G}}} = \sum_{\sigma \in \mathcal{R}_{{\sf{G}}}}\frac{\Tr(\Lambda\Pi_\sigma)}{\Tr(\Pi_\sigma)}\Pi_\sigma,
	\end{equation}
	where $\Pi_\sigma$ denotes the projector onto the support space of $\sigma$, $\mathcal{R}_{\sf{G}}$ denotes the set of all irreps of $\phi$.
\end{proposition}

\subsection{Representation for quantum channel}\label{sec:channel}
Here, we introduce the quantum channel and three frequently-used channel representations which our main results rely on. Denote the Hilbert space for $n$ qubits as $\mathcal{H}$ and the set of linear operators on $\mathcal{H}$ as $\mathcal{L}(\mathcal{H})$. Quantum channels are defined as completely positive and trace-preserving (CPTP) linear maps on $\mathcal{L}(\mathcal{H})$. Given any quantum channel $\mathcal{E}:\mathcal{L}(\mathcal{H})\rightarrow\mathcal{L}(\mathcal{H})$, we can represent it in \textit{Kraus representation}, $\forall\ O\in \mathcal{L}(\mathcal{H})$,
\begin{gather}
	\mathcal{E}(O)=\sum_{l=1}^m K_l O K_l^\dagger,
\end{gather}
where $\{K_l\}$ are the Kraus operators satisfying
\begin{gather}
	\sum_{l=1}^mK_l^\dagger K_l=I.
\end{gather}
With the Kraus representation, the concatenation of the quantum channels or quantum gates is given by
\begin{gather}
	\mathcal{E}_2\circ\mathcal{E}_1(O)=\sum_{l_2=1}^{m_2}\sum_{l_1=1}^{m_1}K_{l_2}K_{l_1}O K_{l_1}^\dagger K_{l_2}^\dagger,
\end{gather}
where $\{K_1\}$ and $\{K_2\}$ are the Kraus operators for $\mathcal{E}_1$ and $\mathcal{E}_2$, respectively.

To describe a long quantum circuit, the Kraus representation is not convenient. Here, we introduce another widely-used representation --- \textit{Liouville representation}. The Liouville representation is defined on a set of trace-orthonormal basis on $\mathcal{L}(\mathcal{H})$. Often, we use the \textit{normalized Pauli group}, i.e., Pauli group with a normalization factor. The $n$-qubit Pauli group is given by
\begin{gather}
    \textsf{P}_n = \bigotimes_{i=1}^n\{I_i,X_i,Y_i,Z_i\},
\end{gather}
where $X_i,Y_i,Z_i$ are the single-qubit Pauli matrices. Then the normalized Pauli group is given by
\begin{equation}
    \left\{\sigma_{i} = \frac{1}{\sqrt{2^n}}P_{i} |P_{i}\in \textsf{P}_n\right\}.
\end{equation}
Each pair of elements $(\sigma_i, \sigma_j)$ in this group satisfies the following constraints under the Hilbert-Schmidt inner product, $\forall\, \sigma_i, \sigma_j \in \textsf{P}_n/\sqrt{2^n}$,
\begin{gather}
    \Tr(\sigma_i^\dagger \sigma_j) = \delta_{ij}.
\end{gather}
Any $n$-qubit operator can be decomposed over the $4^n$ normalized Pauli operators. We can rewrite the density operator $O$ on $\mathcal{L}(\mathcal{H})$ in a vector form,
\begin{gather}
	\lket{O}=\sum_{i\in\textsf{P}^n}\Tr(\sigma_i^\dagger O)\sigma_i,
\end{gather}
Moreover, any quantum channel can be represented as a matrix in the Liouville representation. To be specific, we can represent an arbitrary channel $\Lambda$ acting on an operator $O$ as follows,
\begin{gather}
	\lket{\Lambda(O)}=\Lambda\lket{O}.
\end{gather}
We can see the element of this matrix is given by
\begin{gather}
	\Lambda_{ij}=\lbra{\sigma_i}\Lambda\lket{\sigma_j}=\Tr(\sigma_i\Lambda(\sigma_j)).
\end{gather}
Consequently, in the Liouville representation, the concatenation of two channels can be depicted as the product of two matrices,
\begin{gather}
	\lket{\Lambda_2\circ\Lambda_1(O)}=\Lambda_2\lket{\Lambda_1(O)}=\Lambda_2\Lambda_1\lket{O}.
\end{gather}
The measurement operator can also be vectorized with the Liouville bra-notation according to the definition of the Hilbert-Schmidt inner product. For example, the measurement probability of a state $\rho$ on a positive operator-valued measure (POVM) $\{F_i\}$ is given by,
\begin{gather}
	p_i=\lbraket{F_i}{\rho}=\Tr(F_i^\dagger\rho).
\end{gather}
We call such a Pauli-Liouville representation as the Pauli Transfer Matrix (PTM) representation. 

An $n$-qubit quantum channel can also be described in the $\chi$-matrix representation,
\begin{equation}
	\Lambda(\rho) = d \sum_{i, j} \chi_{ij} \sigma_i \rho \sigma^\dag_j.
\end{equation}
The process matrix $\chi$ is uniquely determined by the orthonormal operator basis $\{\sigma_j\}$ where the first element is proportional to the identity matrix $\sigma_0 = I / \sqrt{d}$ and $d = 2^n$ is the dimension of the quantum system. Often, we take the normalized Pauli operators as the basis of the $\chi$-matrix representation. If a channel is diagonal in this representation, we call it \textit{Pauli channel}.

\subsection{Quantum channel fidelity}
The process fidelity of a channel $\Lambda$ can be defined with its $\chi$-matrix representation,
\begin{equation}
F(\Lambda) = \chi_{00}(\Lambda).
\end{equation}
Here, the quantity $\chi_{00}$ is independent with the choices of operator basis $\{\sigma_j\}$. In the following, we set the operator basis to be normalized Pauli group. Define the \textit{Pauli fidelity} of a quantum channel $\Lambda$,
\begin{equation}
	\lambda_j = d^{-1} \Tr(P_j \Lambda(P_j)),
\end{equation}
which is the diagonal term of the PTM representation of $\Lambda$. We can relate Pauli fidelities to the diagonal terms of $\chi$-matrix representation via Walsh-Hadamard transformation,
\begin{equation}\label{eq:WalshHtrans}
\lambda_j = \sum_i (-1)^{\langle i, j \rangle }\chi_{ii}.
\end{equation}
Here, $\langle i, j \rangle = 0$ if $P_i$ commutes with $P_j$ and $\langle i, j \rangle = 1$ otherwise. Then one can derive the process fidelity from Eq.~\eqref{eq:WalshHtrans},
\begin{equation}
	F(\Lambda) = \chi_{00}(\Lambda) = \frac{1}{d^2} \sum_j \lambda_j,
\end{equation}
which can be viewed as another definition of process fidelity.

There is a relation between the commonly-used average fidelity $F_{ave}$ and the process fidelity~\cite{horodecki1999general},
\begin{equation}
F_{ave} = (dF+1)/(d+1).
\end{equation}
Average fidelity is defined as
\begin{equation}
F_{ave} = \int d\psi \tr(\ketbra{\psi} \Lambda(\ketbra{\psi})),
\end{equation}
where the integral is implemented over Haar measure. These two fidelity measures are both well-defined metrics to quantify the closeness of a channel to the identity.

\subsection{Representation theory in randomized benchmarking}
Now, we can use the representation theory to analyze the randomized benchmarking (RB) procedures. Let us start from a quick review of the standard RB protocol. Considering an $n$-qubit gate set $\textsf{G}$, RB is performed via sampling random gate sequences,
\begin{equation}
	\mathcal{S}_\mathrm{rb} = \mathcal{G}_{m+1} \mathcal{G}_m \cdots \mathcal{G}_1,
\end{equation}
where for $1 \leq i \leq m+1$, $\mathcal{G}_i$ denotes the quantum operation in the PTM representation of a unitary matrix $g_i \in \textsf{G}$. Here, for $1 \leq i \leq m$, $g_i$ is randomly sampled from group $\textsf{G}$ and $g_{m + 1} = (g_m \cdots g_2g_1)^{-1}$. Then, one applies the random gate sequence $\mathcal{S}_\mathrm{rb}$ to the input state $\rho_{\psi0}$ and perform measurement $Q_0$ for a sufficient number of times to estimate the average survival probability,
\begin{equation}
	f(m) = \mathop{\mathbb{E}}_{\mathcal{G}_1 \cdots \mathcal{G}_m} \lbra{Q_0}\tilde{\mathcal{G}}_{m+1} \tilde{\mathcal{G}}_m\cdots\tilde{\mathcal{G}}_2 \tilde{\mathcal{G}}_1 \lket{\rho_{\psi 0}},
\end{equation}
where $\tilde{\mathcal{G}}_i$ denotes the noisy implementation of $\mathcal{G}_i$, $\rho_{\psi0}$ is the noisy preparation of the initial state $\ket{\psi}$, and measurement $Q_0$ also includes errors. Here, we employ the gate independent noise assumption as used in most RB protocols. That is, the noise channels attached to the gate set $\{\mathcal{G}_i\}$ are the same, $\forall i$,
\begin{equation}
	\tilde{\mathcal{G}}_i = \Lambda_L\mathcal{G}_i\Lambda_R,
\end{equation}
where $\Lambda_L$ and $\Lambda_R$ are left and right noise channels. The randomization over the gate sequence $\mathcal{S}_\mathrm{rb}$ can be seen as performing twirling operation for the noise channels between $\mathcal{G}_i$ and $\mathcal{G}_{i + 1}$. Denote $\Lambda = \Lambda_R \Lambda_L$, then we have
\begin{equation}
	f(m) = \lbra{Q}(\mathop{\mathbb{E}}_{g \in \textsf{G}} \mathcal{G}^\dag \Lambda \mathcal{G})^m \lket{\rho_\psi},
\end{equation}
where the right noise channel $\Lambda_R$ of the first gate is absorbed into the state preparation error, $\lket{\rho_\psi} = \Lambda_R \lket{\rho_{\psi 0}}$ and the left noise channel $\Lambda_L$ of the last gate is absorbed into the measurement error, $\lbra{Q} = \lbra{Q_0} \Lambda_L$. According to Proposition \ref{prop:twirling}, one can express $f_m$ in a more elegant manner,
\begin{equation}\label{eq:pm2}
	f(m) = \sum_{\sigma\in R_{\textsf{G}}} \lbra{Q}\Pi_\sigma ^m \lket{\rho_{\psi}} \lambda_\sigma^m,
\end{equation}
where $\Pi_\sigma$ denotes the projector onto the irreducible subspace associated with $\sigma$ and $\lambda_\sigma$ contains the trace information of the channel $\Lambda$ on the subspace,
\begin{equation}
	\lambda_\sigma = \frac{\Tr(\Lambda\Pi_\sigma)}{\Tr(\Pi_\sigma)}.
\end{equation}
Since the twirling operation will not change the trace value of $\Lambda$, the process fidelity of $\Lambda$ can be given by,
\begin{equation}
	F(\Lambda, I) = 4^{-n} \sum_\sigma d_\sigma \lambda_\sigma,
\end{equation}
where $d_\sigma$ is the dimension of the irrep $\sigma$ with $d_\sigma = \Tr(\Pi_\sigma)$. Here, $F(\Lambda, I)$ is also known as the entanglement fidelity, or $\chi_{00}$. We call $\{\lambda_\sigma\}$ the quality parameters since they reflect the noisy level of a channel.

In the conventional RB protocol, the $n$-qubit Clifford group $\textsf{C}_n$ is often picked as the target gate set, which we call Clifford RB. Any Clifford operation $C \in \textsf{C}_n$ satisfies $\{C P_i C^{-1} \mathrm{\ for\ all\ } i\} = \textsf{P}_n$, which is a transformation permuting Pauli operators. Note that in the PTM representation, $\textsf{C}_n$ has only one nontrivial irrep, thus we only need to solve one single quality parameter. This is rather convenient, but it is hard to extend the conventional RB scheme to other group with multiple nontrivial irreps due to the multi-variable fitting problem, as shown in Eq.~\eqref{eq:pm2}, which has poor confidence intervals for $\{\Lambda_\sigma\}$.

To solve the fitting problem, in the following discussions, we employ the technique of character randomized benchmarking, which utilizes the generalized projection formula of Lemma \ref{lemma:project_formula} in the character theory.

\section{Benchmarking protocols}\label{sec:CCBandCAB}
\subsection{Character cycle benchmarking}\label{append:CCB}
Here, we present further technical details of the character cycle benchmarking (CCB) protocol. As shown in Lemma \ref{lemma:project_formula}, we can rewrite the projection equation of Eq.~\eqref{eq:project_formula} for a quantum operation group $\{\mathcal{G}\}$ in the PTM representation,
\begin{equation}
\Pi_\sigma = \frac{d_\sigma}{|\textsf{G}|}\sum_{g \in \textsf{G}} \chi_\sigma(g) \mathcal{G},
\end{equation}
then one can add an additional gate, as a character gate, to construct the projector for extracting the quality parameter $\lambda_\sigma$ associated with irrep $\sigma$. In CCB, we estimate the process fidelity of a target gate $U \in \textsf{C}_n$ via the twirling group $\textsf{G} = \textsf{P}_n$, which has $4^n$ irreps supported by the Pauli operator $\{P_j\}$, respectively. The schematic circuit is given in Fig.~1(a) in the main text. The detailed procedures of CCB protocol is given in Box \ref{box:CCB}.

\begin{mybox}[label={box:CCB}]{{Procedures for character cycle benchmarking}}
\begin{enumerate}
\item Sample a Pauli operator, $P_j \in \textsf{P}_n$, and initialize the state, $\ket{\psi_j}$, such that $\ket{\psi_j} = P_j\ket{\psi_j}$.

\item Sample a gate sequence, $(P^{(0)}, P^{(1)}, P^{(2)}, \cdots, P^{(2m)})$, where $P^{(i)} (0 \leq i \leq 2m)$ are sampled uniformly at random from Pauli group $\textsf{P}^n$.

\item Apply the gate sequence,
\begin{equation}\label{eq:sequence0}
\mathcal{S}_{\mathrm{ccb}} = \mathcal{U}_{\mathrm{inv}} \circ \mathcal{U}^{-1} \circ \mathcal{P}^{(2m)} \circ \mathcal{U} \circ \mathcal{P}^{(2m-1)} \circ \cdots \circ \mathcal{U}^{-1} \circ \mathcal{P}^{(2)} \circ \mathcal{U} \circ \mathcal{P}^{(1)} \circ \mathcal{P}^{(0)},
\end{equation}
where the inverse gate $\mathcal{U}_{\mathrm{inv}} = \mathcal{P}^{(1)} \circ \mathcal{U}^{-1} \circ\cdots\circ P^{(2m)} \circ \mathcal{U}$ is a local gate.

\item Measure in $P_j$ basis and compute the $P_j$-weighted survival probability,
\begin{equation}
f_j(m, \mathcal{S}_{\mathrm{ccb}}) = \chi_j(P^{(0)})d_j \Tr(Q_j \mathcal{S}(\rho_{\psi_j})),
\end{equation}
where the dimension of the representation associated with $P_j$ is $d_j = 1$, the corresponding character is $\chi_j(P^{(0)}) = 1$ if $P_j$ commutes with $P^{(0)}$ and -1 otherwise, and $Q_j$ and $\rho_{\psi_j}$ are the noisy implementations of the state preparation and measurement, respectively.

\item Repeat steps 2 to 4 for a sufficient number of sequences and estimate the average value
\begin{equation}\label{eq:CCBf_jm}
f_j(m) = \mathop{\mathbb{E}}_{\mathcal{S}_{\mathrm{ccb}}} f_j(m, \mathcal{S}_{\mathrm{ccb}}).
\end{equation}

\item Repeat steps 2 to 5 for different $m$ and fit to the function
\begin{equation}\label{eq:CCBfitting_model}
f_j(m) = A_j \lambda_j^{2m},
\end{equation}
where $A_j$ and $\lambda_j$ are fitting parameters.

\item Sample $M$ Pauli operators $\{P_j\}$ in step 1, and for each $P_j$, repeat steps 2 to 6. Finally estimate the CCB fidelity
\begin{equation}\label{eq:Fccb2}
F_{\mathrm{ccb}} = \frac{1}{M}\sum_{\{P_j\}} \lambda_j.
\end{equation}
\end{enumerate}
\end{mybox}

Note that in the CCB procedures, $P^{(0)}$ is the character gate, which we will merge into the gate $P^{(1)}$ in practical implementation. Besides, $P^{(0)}$ is not included in computing the inverse gate $U_{\mathrm{inv}}$. The average $P_j$-weighted survival probability of Eq.~\eqref{eq:CCBf_jm} can be further evaluated by,
\begin{equation}\label{eq:CCBfm1}
f_j(m) = d_j \lbra{Q_j} \left( (\mathop{\mathbb{E}}_{P_t \in\textsf{P}_n} \mathcal{P}_t^{-1} \mathcal{U}^{-1}\Lambda\mathcal{U} \mathcal{P}_t) (\mathop{\mathbb{E}}_{P_r \in\textsf{P}_n} \mathcal{P}_r^{-1}\Lambda \mathcal{P}_r) \right)^m (\mathop{\mathbb{E}}_{P\in\textsf{P}_n} \chi_j(P) \mathcal{P}) \lket{\rho_{\psi_j}},
\end{equation}
According to Proposition \ref{prop:twirling}, we have
\begin{equation}\label{eq:Lambda_P}
\begin{split}
\Lambda_\textsf{P} &\equiv \mathop{\mathbb{E}}_{P_r \in\textsf{P}_n} \mathcal{P}_r^{-1}\Lambda \mathcal{P}_r = \sum_{P_r \in \textsf{P}_n} \lbra{P_r}\Lambda\lket{P_r} \Pi_{P_r}, \\
\Lambda_\textsf{P}^U &\equiv \mathop{\mathbb{E}}_{P_t \in\textsf{P}_n} \mathcal{P}_t^{-1}\mathcal{U}^{-1}\Lambda\mathcal{U} \mathcal{P}_t = \sum_{P_t \in \textsf{P}_n} \lbra{P_t}\mathcal{U}^{-1}\Lambda\mathcal{U}\lket{P_t} \Pi_{P_t} = \mathcal{U}^{-1}\Lambda_\textsf{P}\mathcal{U}
\end{split}
\end{equation}
here $\Pi_{P_{r}} = \lket{P_{r}}\lbra{P_{r}}$ and $\Pi_{P_{t}} = \lket{P_{t}}\lbra{P_{t}}$ are the projectors onto the support spaces of $P_{r}$ and $P_{t}$ in the PTM representation, respectively. Then Eq.~\eqref{eq:CCBfm1} can be further simplified to
\begin{equation}\label{eq:CCBfm2}
f_j(m) = \lbra{Q_j} \left(\Lambda_\textsf{P}^U \Lambda_\textsf{P}\right)^m \left(d_j \mathop{\mathbb{E}}_{P\in\textsf{P}_n} \chi_j(P) \mathcal{P}\right) \lket{\rho_s}.
\end{equation}
According to the generalized projection formula of Lemma \ref{lemma:project_formula}, we have
\begin{equation}\label{eq:P_j}
\Pi_j = \frac{d_j}{4^n} \sum_{P\in\textsf{P}_n}
\chi_j(P)\mathcal{P},
\end{equation}
where $\forall j$, the dimension of the subspace associated with the $P_j$ is $d_j = 1$. Substituting Eq.~\eqref{eq:P_j} and Eq.~\eqref{eq:Lambda_P} into Eq.~\eqref{eq:CCBfm2}, we have
\begin{equation}
\begin{split}
f_j(m) &= \lbra{Q_j} \left(\sum_{P_t\in\textsf{P}_n} \lbra{P_t}\mathcal{U}^{-1}\Lambda\mathcal{U}\lket{P_t}\Pi_{P_t} \sum_{P_r\in\textsf{P}_n}\lbra{P_r}\Lambda\lket{P_r}\Pi_{P_r}\right)^m \Pi_j \lket{\rho_s}\\
&= \lbra{Q_j}\Pi_j\ket{\rho_s} \left(\lbra{P_j}\mathcal{U}^{-1}\Lambda\mathcal{U}\lket{P_j}\lbra{P_j}\Lambda\lket{P_j}\right)^m,
\end{split}
\end{equation}
Fit the survival probability $f_j(m)$ to the function of Eq.~\eqref{eq:CCBfitting_model}, one can obtain the fitting parameters,
\begin{equation}\label{eq:CCBfiting_params}
\begin{split}
\lambda_j &= \sqrt{\lbra{P_j}\mathcal{U}^{-1}\Lambda\mathcal{U}\lket{P_j} \lbra{P_j}\Lambda\lket{P_j}}, \\
A_j &= \lbra{Q_j}\Pi_j\lket{\rho_s}.
\end{split}
\end{equation}
Note that quality parameter $\lambda_j$ is one of the Pauli fidelities of channel $\sqrt{\mathcal{U}^{-1}\Lambda_\textsf{P}\mathcal{U} \Lambda_\textsf{P} }$. Then we can use the CCB protocol to sample Pauli operators $\{P_j\}$ and estimate process fidelity $F(\sqrt{\mathcal{U}^{-1}\Lambda_\textsf{P}\mathcal{U} \Lambda_\textsf{P} }, I)$, which is close to $F(\Lambda, I)$, as we shall see in Section \ref{append:CCBfidelity}.

\subsection{CCB fidelity}\label{append:CCBfidelity}
Here, we shall explain the physical meaning of the estimated process fidelity in the CCB protocol. As shown in Section \ref{append:CCB}, the process fidelity we estimate in the CCB protocol is called the CCB fidelity,
\begin{equation}\label{eq:CCBfidelity}
F_{\mathrm{ccb}}(\Lambda) = F(\sqrt{(\mathcal{U}^{-1}\Lambda_\textsf{P}\mathcal{U}) \Lambda_\textsf{P}}, I),
\end{equation}
here $\Lambda_P$ is the noise channel $\Lambda$ twirled by the Pauli gate set defined in Eq.~\eqref{eq:Lambda_P}. In order to compute the deviation between $F_{\mathrm{ccb}}(\Lambda)$ and $F(\Lambda, I)$, we employ the rearrangement inequality.

\begin{lemma}[Rearrangement inequality \cite{Ineq1952}]
	Consider two sets of $n$ real numbers $\{x_1, x_2, \cdots, x_n\}$, $\{y_1, y_2, \cdots, y_n\}$, and $x_1 \leq x_2 \leq \cdots \leq x_n, y_1 \leq y_2 \leq \cdots \leq y_n$, one have
	\begin{equation}
		x_ny_1 + \cdots + x_1y_n \leq x_{\tau(1)}y_1 + \cdots + x_{\tau(n)}y_n \leq x_1y_1 + \cdots + x_ny_n,
	\end{equation}
	where $\{x_{\tau(i)}\}$ is an arbitrary permutation of $\{x_i\}$.
\end{lemma}

Any Clifford unitary $U \in \textsf{C}_n$ satisfies, $\forall i$,
\begin{equation}
	U P_i U^{-1} = P_{u(i)}, 
\end{equation}
where $P_i \in \textsf{P}_n$ and the index permutation $u(i)$ is determined by the operation $U$. Thus, the diagonal terms of the Pauli channel $\mathcal{U}^{-1}\Lambda_\textsf{P}\mathcal{U}$ are a permutation of those in channel $\Lambda_\textsf{P}$. 

Denote $\{\omega_i\}$ as the Pauli fidelities of $\Lambda_\textsf{P}$ and $\{\omega_{u(i)}\}$ as the Pauli fidelities of the channel $\mathcal{U}^{-1}\Lambda_\textsf{P}\mathcal{U}$. Here, we assume $\forall i$, $\omega_i\ge0$ throughout the paper. In practice, the values of $\omega_i$ are normally close to 1 for a high-fidelity gate implementation. Using the rearrangement inequality, one have
\begin{equation}
	\omega_{u(1)}\omega_1 + \cdots + \omega_{u(4^n)}\omega_{4^n} \leq \omega_1^2 + \cdots + \omega_{4^n}^2,
\end{equation}
then one can further derive that
\begin{equation}
\sqrt{\omega_{u(1)}\omega_1} + \cdots + \sqrt{\omega_{u(4^n)}\omega_{4^n}} \leq \omega_1 + \cdots + \omega_{4^n}.
\end{equation}
From the definition of the process fidelity and CCB fidelity, one can conclude that
\begin{equation}
	F_{\mathrm{ccb}}(\Lambda) \leq F(\Lambda, I).
\end{equation}

\subsection{CCB with local gauge freedom}\label{app:LocalGauge}

In this section ,we give detailed analysis for the CCB protocol with local gauge freedom. As shown in Fig.~\ref{fig:gaugefreedom2}(a), we insert local gates $L$ and $L^{-1}$ between the twirling gates $P_i$ and the target gates $U$, $U^{-1}$. We define the local gate as $L=\bigotimes_{i=1}^n L_i$, where $L_i$ can be an arbitrary single-qubit gate. Note that the local gates can be absorbed into $P_i, U, U^{-1}$, as shown in Fig.~\ref{fig:gaugefreedom2}(b), and thus do not need to be implemented individually. Besides, the initial character gate $LP^{(0)}$ and the twirling gate $LP^{(1)}L^{-1}$ can be treated as single gates in experiments as well.

\begin{figure}[htbp!]
	\centering
	\includegraphics[width=1.0\textwidth]{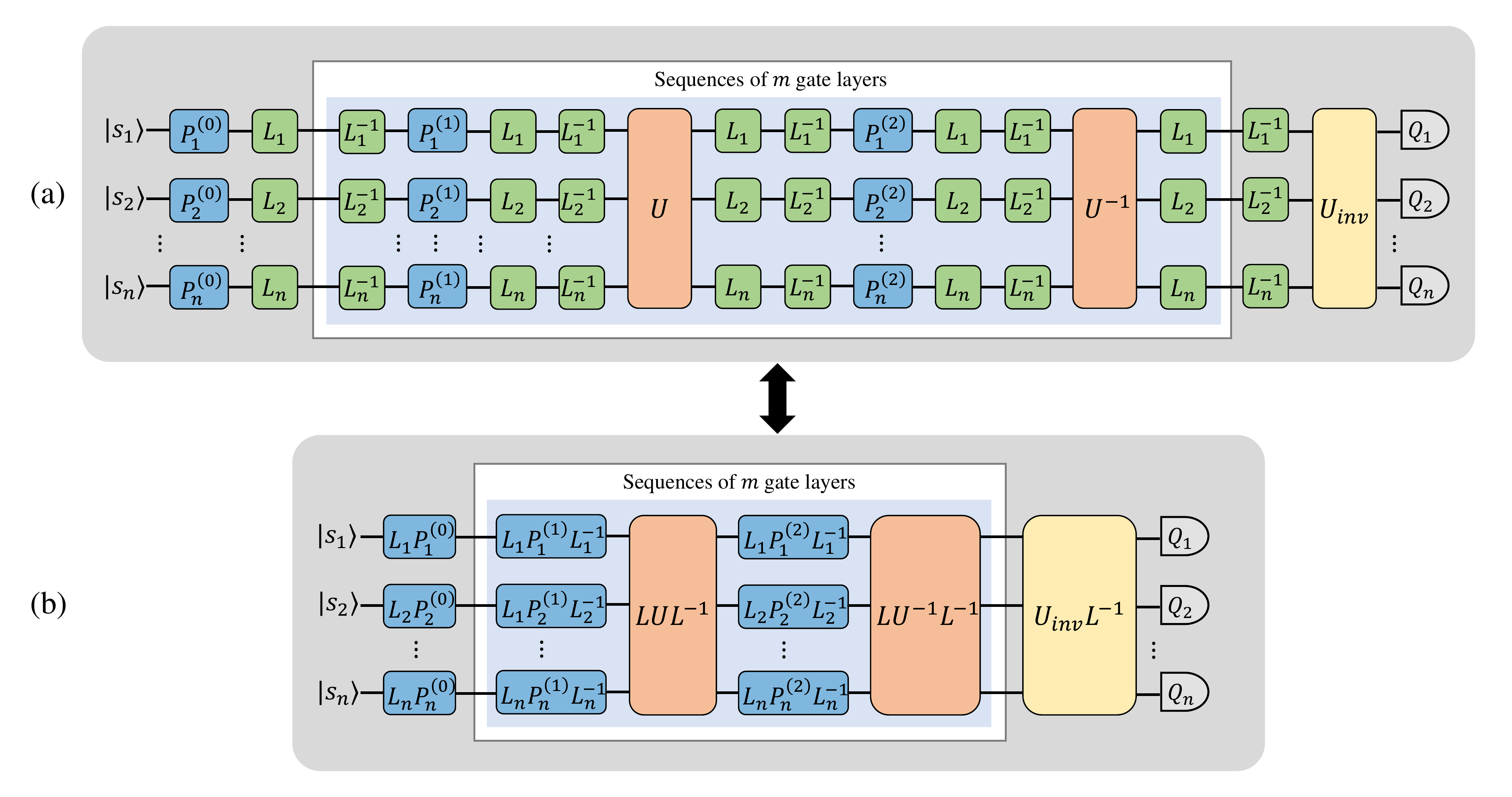}
	\caption{Illustrations of the CCB circuit with local gauge transformation. Circuit (a) is equivalent to the original CCB circuit in the main text if all of gates are ideal. The orange boxes represent the target gate $\mathcal{U}$ and its inverse gate $\mathcal{U}^{-1}$. The blue boxes represent the random Pauli gates. The green boxes represent the inserted local gates $L$ and $L^{-1}$, where $L=\bigotimes_{i=1}^n L_i$. The yellow box denote the inverse gate for the $m$ inner gate layers in the light blue box. In practice, we implement gates in circuit (b) while absorbing local gates $L$ and $L^{-1}$ into twirling gates and target gates. Here, the target gate after gauge transformation becomes $LUL^{-1}$.}
	\label{fig:gaugefreedom2}
\end{figure}

For simplicity, we assume that the local twirling gates are noiseless and the noise of $LUL^{-1}$ and $LU^{-1}L^{-1}$ are the same. Now, we can analyse the relationship between the sequence length of gate layers and the survival probability. Given the noise channel of the target gate $\Lambda$, the averaged $P_j$-weighted survival probability is given by,
\begin{equation}\label{eq:CCBlgf}
\begin{aligned}
f_j(m) &= d_j \lbra{Q_j} \left( \mathop{\mathbb{E}}_{P_t \in\textsf{P}_n} \mathcal{P}_t^{-1} \mathcal{L}^{-1} \mathcal{L} \mathcal{U}^{-1} \mathcal{L}^{-1} \Lambda \mathcal{L}\mathcal{U}\mathcal{L}^{-1} \mathcal{L} \mathcal{P}_t \right)^m \mathcal{L}^{-1} \mathcal{L} (\mathop{\mathbb{E}}_{P_r \in\textsf{P}_n} \mathcal{P}_r^{-1} \mathcal{L}^{-1}\Lambda\mathcal{L} \mathcal{P}_r) \mathcal{L}^{-1} \mathcal{L}  (\mathop{\mathbb{E}}_{P\in\textsf{P}_n} \chi_j(P) \mathcal{P}) \lket{\rho_s}\\
&= d_j \lbra{Q_j} \left( \mathop{\mathbb{E}}_{P_t \in\textsf{P}_n} \mathcal{P}_t^{-1}\mathcal{U}^{-1}\Lambda_L\mathcal{U}\mathcal{P}_t \right)^m (\mathop{\mathbb{E}}_{P_r \in\textsf{P}_n} \mathcal{P}_r^{-1}\Lambda_L \mathcal{P}_r) (\mathop{\mathbb{E}}_{P\in\textsf{P}_n} \chi_j(P) \mathcal{P}) \lket{\rho_s},
\end{aligned}
\end{equation}
where $\Lambda_L = \mathcal{L}^{-1} \Lambda \mathcal{L}$. Equation.~\eqref{eq:CCBlgf} is equivalent to Eq.~\eqref{eq:CCBfm1} except for substituting $\Lambda$ with $\Lambda_L$. That means running the CCB protocol with the circuit in Fig.~\ref{fig:gaugefreedom2} would provide an estimation of the process fidelity $F(\Lambda_L, I)$, which is equivalent to $F(\Lambda, I)$ as shown below,
\begin{equation}
F(\Lambda_L, I) = d^{-2}\Tr(\Lambda_L) = d^{-2}\Tr(\Lambda) = F(\Lambda, I).
\end{equation}
This accounts for the validity of the CCB protocol with local gauge freedom.

In the Pauli-Liouville representation, the off-diagonal terms of channel $\Lambda$ vanish after the Pauli twirling. When applying $L\textsf{P}_nL^{-1}$ for twirling, the off-diagonal terms of $\Lambda$ in the $L$-transformed Pauli-Liouville representation, defined on $L\textsf{P}_nL^{-1} / \sqrt{d}$, would vanish. Here, the local gauge transformation merely changes the representation of $\Lambda$, while maintains the exponential decay form of the survival probability. From another point of view, the CCB protocol with local gauge freedom is characterizing the diagonal terms of $\chi$-matrix of the noise channel under basis $L{\sf{P}}_nL^{-1}$ instead of $\sf{P}_n$. As the average fidelity is irrelevant to the representation basis, all the analysis in previous subsections still applies.

\subsection{Character-average benchmarking}\label{append:CAB}
Here, we shall give more details of the CAB protocol. The CAB protocol can be seen as an improvement based on the CCB protocol, which adds an additional local Clifford gate $C \in \textsf{C}_1^{\otimes n}$ to the beginning and $C^{-1}$ to the end of the inner gate sequence, respectively, as shown in Fig.~1(b) in the main text. The detailed procedures of CAB protocol is given in Box 1 in the main text.

Let us first analyze the irreps of $\textsf{C}_1^{\otimes n}$ in the PTM representation. According to Eq.~\eqref{eq:directsum}, the PTM representation of 1-qubit Clifford group $\textsf{C}_1$ is the direct sum of the trivial representation, $I$, associated with the identity element, and a nontrivial irrep, $\Upsilon$, supported by the subspace defined on the Pauli matrices $\{X, Y, Z\}$. For $n$-qubit local Clifford group $\textsf{C}_1^{\otimes n}$, there exists $2^n$ irreps
\begin{equation}
\mathcal{R}^n = \{I, \Upsilon\}^{\otimes n}.
\end{equation}

According to Proposition \ref{prop:twirling}, the twirling over $\textsf{C}_1^{\otimes n}$ can be written as
\begin{equation}\label{eq:LambdaCn}
\Lambda_\textsf{C} = \sum_{\sigma_k \in \mathcal{R}^n} \frac{\Tr(\Lambda \Pi_{\sigma_k})}{\Tr(\Pi_{\sigma_k})} \Pi_{\sigma_k},
\end{equation}
where $\Pi_{\sigma_k}$ denotes the projector onto the support space of $\sigma_k$ in the PTM representation, $\Tr(\Pi_{\sigma_k}) = 3^{\pi(\sigma_k)}$ is the dimension of $\sigma_k$, $\pi(\sigma_k)$ is the counting of $\Upsilon$ in $\sigma_k$. In the following, we will sometimes abuse the irrep notation and treat $\sigma_k$ as the Pauli operator set that defines the support space of the irrep $\sigma_k$ in the PTM representation.

It is obvious to see that channel $\Lambda_\textsf{C}$ is a partial depolarizing channel, which has the same diagonal values in the subspace associated with the irrep $\sigma_k$ for each $k$. We call the diagonal parameter as the \textit{local Clifford eigenvalue}.

\begin{definition}[Local Clifford eigenvalue]
The twirling channel $\Lambda_{\sf{C}}$ over $n$-qubit local Clifford group ${\sf{C}}_1^{\otimes}$ has $2^n$ quality parameters $\{\gamma_k\}$, which are defined as the local Clifford eigenvalues
\begin{equation}
\gamma_k = \frac{\Tr(\Lambda \Pi_{\sigma_k})}{\Tr(\Pi_{\sigma_k})} = \frac{1}{|\sigma_k|} \sum_{P_j \in \sigma_k} \lbra{P_j}\Lambda\lket{P_j},
\end{equation}
where $\gamma_k$ can be seen as an average value of the Pauli fidelities $\{\lbra{P_j}\Lambda\lket{P_j}|P_j \in \sigma_k\}$.
\end{definition}

We will give a concrete example to show how $\textsf{C}_1^{\otimes n}$ twirls an arbitrary channel. For a 2-qubit local Clifford group $\textsf{C}_1 \otimes \textsf{C}_1$, denote its irreps as
\begin{equation}
\begin{split}
\sigma_1 &= \{I \otimes I\} = \{II\}, \\
\sigma_2 &= \{I \otimes \Upsilon\} = \{IX, IY, IZ\}, \\
\sigma_3 &= \{\Upsilon \otimes I\} = \{XI, YI, ZI\}, \\
\sigma_4 &= \{\Upsilon \otimes \Upsilon\} = \{XX, XY, XZ, YX, YY, YZ, ZX, ZY, ZZ\},
\end{split}
\end{equation}
and the dimensions of these irreps are
\begin{equation}
\begin{split}
\Tr(\Pi_{\sigma_1}) &= |\sigma_1| = 3^0, \\
\Tr(\Pi_{\sigma_2}) &= |\sigma_2| = 3^1, \\
\Tr(\Pi_{\sigma_3}) &= |\sigma_3| = 3^1, \\
\Tr(\Pi_{\sigma_4}) &= |\sigma_4| = 3^2. \\
\end{split}
\end{equation}
Then the twirling of $\textsf{C}_1^{\otimes 2}$ over a channel, $\Lambda$, is given by Eq.~\eqref{eq:LambdaCn},
\begin{equation}
\Lambda_\textsf{C} = \Tr(\Lambda\Pi_{\sigma_1})\Pi_{\sigma_1} + \frac{\Tr(\Lambda \Pi_{\sigma_2})}{3}\Pi_{\sigma_2} + \frac{\Tr(\Lambda \Pi_{\sigma_3})}{3}\Pi_{\sigma_3} + \frac{\Tr(\Lambda \Pi_{\sigma_4})}{9}\Pi_{\sigma_4},
\end{equation}
which has 4 local Clifford eigenvalues $\{\Tr(\Lambda\Pi_{\sigma_1}), \Tr(\Lambda \Pi_{\sigma_2}) / 3, \Tr(\Lambda \Pi_{\sigma_3}) / 3, \Tr(\Lambda \Pi_{\sigma_4}) / 9\}$.

Now return to the $n$-qubit case. The randomization over an inner gate layer in CAB is the same as CCB and generate a composite Pauli channel,
\begin{equation}
\Lambda_\mathrm{in} = \mathcal{U}^{-1}\Lambda_\textsf{P}\mathcal{U} \Lambda_\textsf{P},
\end{equation}
where $\Lambda_\textsf{P}$ is a Pauli twirling channel defined in Eq.~\eqref{eq:Lambda_P}. The initial and last random local Clifford gates $\mathcal{C}, \mathcal{C}^{-1}$ in CAB together perform a local Clifford twirling over the inner Pauli channel,
\begin{equation}\label{eq:Lambda_C}
\begin{split}
\Lambda_\textsf{C}(m) &= \mathop{\mathbb{E}}_{C \in \textsf{C}_1^{\otimes n}} \mathcal{C}^{-1} \Lambda_\mathrm{in}^m \mathcal{C} \\
&= \sum_{\sigma_k \in \mathcal{R}^n} \frac{\Tr(\Lambda_\mathrm{in}^m \Pi_{\sigma_k})}{\Tr(\Pi_{\sigma_k})} \Pi_{\sigma_k}.
\end{split}
\end{equation}
Note that $\Lambda_\textsf{C}(m)$ is a partial depolarizing channel, as we mentioned above. Thus one can extract  local Clifford eigenvalues via performing corresponding measurement observables $P \in \sigma_k$. In the CAB protocol, by measuring in $Z^{\otimes n}$ basis, we can infer the measurement results of the $2^{n}$ observables $\{I, Z\}^{\otimes n}$, which span all the irrep spaces in $\mathcal{R}^n$.
We assume that the measurement $Z^{\otimes n}$ is performed with negligible errors that will not influence the fidelity estimations too much. If one wants completely remove the SPAM errors, an additional character gate from the gate set $\{I, Z\}^{\otimes n}$ can be added to the begin of the CAB gate sequence, which is similar to the CCB protocol.

The survival probability of Eq.~\eqref{eq:CABfm} in the main text can be derived as, for Pauli operator $Q_k \in \{I, Z\}^{\otimes n}$,
\begin{equation}\label{eq:CABfm1}
f_k(m) = \lbra{Q_k}\Lambda_\textsf{C}(m) \lket{\rho_s},
\end{equation}
which contains fidelity information for irrep $\sigma_k$ such that $Q_k \in \sigma_k$. Substituting Eqs.~\eqref{eq:Lambda_P} and \eqref{eq:Lambda_C} to Eq.~\eqref{eq:CABfm1}, we have
\begin{equation} \label{eq:fkm}
\begin{split}
f_k(m) &= \lbra{Q_k} \Pi_{Q_k}\Pi_{\sigma_k} \frac{\Tr(\Lambda_\mathrm{in}^m \Pi_{\sigma_k})}{\Tr(\Pi_{\sigma_k})} \lket{\rho_s} \\
&= \lbraket{Q_k}{\rho_s} \sum_{P_j \in \sigma_k} \frac{\left(\lbra{P_j}\mathcal{U}^{-1}\Lambda\mathcal{U}\lket{P_j} \lbra{P_j}\Lambda\lket{P_j}\right)^m}{\Tr(\Pi_{\sigma_k})} \\
&= \lbraket{Q_k}{\rho_s} \sum_{P_j \in \sigma_k} \frac{\lambda_j^{2m}}{\Tr(\Pi_{\sigma_k})},
\end{split}
\end{equation}
where $\lambda_j$ is the Pauli fidelity of the channel $\sqrt{\mathcal{U}^{-1}\Lambda_\textsf{P}\mathcal{U} \Lambda_\textsf{P} }$ defined in Eq.~\eqref{eq:CCBfiting_params}. Fit the survival probability $f_k(m)$ to the function $f_k(m) = A_k \mu_k^{2m}$, one can solve the quality parameters $\{\mu_k\}$ and estimate the CAB fidelity as
\begin{equation}\label{eq:CABfidelitysum}
F_\mathrm{cab} = 4^{-n}\sum_{\sigma_k \in \mathcal{R}^n} d_{\sigma_k} \mu_k,
\end{equation}
where $d_{\sigma_k} = \Tr(\Pi_{\sigma_k})$ is the dimension of irrep $\sigma_k$. 

Note that our fitting method averages multi-exponential decays $\{\lambda_j^{2m}\}$ into one exponential decay $\mu_k^{2m}$, which leads to $\mu_k \geq \mathrm{ave}_{P_j \in \sigma_k}\lambda_j$, as proved in Section \ref{append:CABfidelity}. We will further prove that the CAB fidelity is the upper bound of the CCB fidelity in Section \ref{append:CABfidelity}.

\subsection{Fitting analysis for CAB}\label{append:CABfidelity}
Here, we shall analyze the fitting results in the CAB protocol and explain the CAB fidelity $F_\mathrm{cab}$ in details. As shown in Eq.~\eqref{eq:fkm}, the survival probability in CAB is given by
\begin{equation}
	f_k(m) = \frac{\lbraket{Q_k}{\rho_s}}{|\sigma_k|} \sum_{P_j \in \sigma_k} \lambda_j^{2m},
\end{equation}
where $|\sigma_k| = \Tr(\Pi_{\sigma_k})$. Fit $f_k(m)$ to the function,
\begin{equation}\label{eq:CABfitting_func1}
	f_k(m) = A_k \mu_k^{2m},
\end{equation}
with $A_k >0, \mu_k >0$. Take the natural logarithm for both sides in Eq.~\eqref{eq:CABfitting_func1},
\begin{equation}\label{eq:CABfitting_func2}
	y = \beta_0  + \beta_1 m,
\end{equation}
where
\begin{equation}
	\begin{split}
		&y = \ln f_k(m), \\
		&\beta_0 = \ln A_k, \\
		&\beta_1 = 2\ln\mu_k.
	\end{split}
\end{equation}
Next we will employ the \textit{least-squares estimation} for the linear regression of Eq.~\eqref{eq:CABfitting_func2}. Set $\{m_1, m_2, \cdots, m_q\}$ as inputs, assume $m_1 < m_2 < \cdots < m_q$ without loss of generality. The regression matrix $\mathbf{M}$ and the observed values $\mathbf{Y}$ are given by
\begin{equation}
	\mathbf{M} = \begin{pmatrix} 1 & m_1\\ 1 & m_2\\ \vdots & \vdots \\ 1 & m_q \end{pmatrix}, \quad
	\mathbf{Y} = \begin{pmatrix} y_1 \\ y_2 \\ \vdots \\ y_q\end{pmatrix}.
\end{equation}
Using the least-squares estimation, one can solve the optimum parameters for the model,
\begin{equation}
\hat{\bm{\beta}} = \begin{pmatrix} \hat{\beta_0} \\ \hat{\beta_1} \end{pmatrix} = (\mathbf{M}^T\mathbf{M})^{-1} \mathbf{M}^T \mathbf{Y}.
\end{equation}
The fitting parameter of interest in CAB is $\mu_k$,
\begin{equation}\label{eq:mukbeta}
	\begin{split}
		\mu_k &= \exp(\frac{\hat{\beta}_1}{2}), \\
	\end{split}
\end{equation}
where
\begin{equation}
	\begin{split}
		\hat{\beta_1} &= \frac{q\sum_{i=1}^q m_i y_i - \sum_{i=1}^q m_i \sum_{i=1}^q y_i}{q\sum_{i=1}^q m_i^2 - (\sum_{i=1}^q m_i)^2}.
	\end{split}
\end{equation}
Then, we can have the following lemma.

\begin{lemma}\label{lemma:fitting}
	Denote $\{\mu_k\}$ as the fitting parameters we solve in the CAB protocol, given in Eq.~\eqref{eq:mukbeta}. For each $\mu_k$, we have
	\begin{equation}
		\mu_k \geq \frac{1}{|\sigma_k|}\sum_{P_j \in \sigma_k}\lambda_j,
	\end{equation}
	and the CAB fidelity is the upper bound of the CCB fidelity,
	\begin{equation}
		F_\mathrm{cab}(\Lambda) \geq F_\mathrm{ccb}(\Lambda).
	\end{equation}
\end{lemma}

\begin{proof}
For a simple linear regression using the least-squares estimation, there exist some observations above the fitting curve while the others are below the curve. Then one can conclude that there exists two adjacent observation whose slope in between is smaller than the fitting slope of the curve. Thus for the simple linear regression in CAB, the fitting slope holds
	\begin{equation}\label{eq:CABslop}
		\hat{\beta_1} \geq \min_i \left(\frac{y_{i + 1} - y_i}{m_{i + 1} - m_i}\right).
	\end{equation}
According to the Chebyshev sum inequality \cite{Ineq1952}, $\forall i$,
	\begin{equation}\label{eq:CABineq1}
		\begin{split}
			\frac{1}{|\sigma_k|}\sum_{P_j \in \sigma_k} \lambda_j^{2m_{i + 1}} &= \frac{1}{|\sigma_k|}\sum_{P_j \in \sigma_k} \lambda_j^{2m_i}  \lambda_j^{2(m_{i+1} - m_i)} \\
			&\geq \left(\frac{1}{|\sigma_k|}\sum_{P_j \in \sigma_k} \lambda_j^{2m_i}\right)  \left(\frac{1}{|\sigma_k|}\sum_{P_j \in \sigma_k} \lambda_j^{2(m_{i+1} - m_i)}\right).
		\end{split}
	\end{equation}
	According to Eq.~\eqref{eq:CABineq1}, we can further derive that
	\begin{equation}\label{eq:CABineq2}
		\begin{split}
			y_{i+1} - y_i &= \ln(\frac{\sum_{P_j\in\sigma_k} \lambda_j^{2m_{i+1}}}{\sum_{P_j\in\sigma_k} \lambda_j^{2m_i}}) \\
			&\geq \ln(\frac{1}{|\sigma_k|} \sum_{P_j\in\sigma_k} \lambda_j^{2(m_{i+1} - m_i)}) \\
			&\geq 2(m_{i+1} - m_i)\ln(\frac{1}{|\sigma_k|}\sum_{P_j\in\sigma_k}\lambda_j),
		\end{split}
	\end{equation}
where the last inequality comes from the convexity of function $x^{2(m_{i+1} - m_i)}$.
	
	Substituting Eq.~\eqref{eq:CABineq2} into Eq.~\eqref{eq:CABslop}, we have
	\begin{equation}
		\hat{\beta}_1 \geq 2\ln(\frac{1}{|\sigma_k|}\sum_{P_j\in\sigma_k}\lambda_j).
	\end{equation}
	Since $\hat{\beta}_1 = 2\ln\mu_k$, we can conclude that
	\begin{equation}
		\mu_k \geq \frac{1}{|\sigma_k|}\sum_{P_j \in \sigma_k}\lambda_j.
	\end{equation}
	Besides, the dimensions of irreps of the local Clifford gate $\textsf{C}_1^{\otimes n}$ satisfies
	\begin{equation}
		\sum_{\sigma_k \in \mathcal{R}^n} |\sigma_k| = \sum_{\sigma_k \in \mathcal{R}^n} \Tr(\Pi_{\sigma_k}) = 4^n.
	\end{equation}
Thus, 
	\begin{equation}
		\begin{split}
			F_\mathrm{cab} &= 4^{-n}\sum_{\sigma_k \in \mathcal{R}^n} |\sigma_k|\mu_k \\
			&\geq 4^{-n}\sum_{\sigma_k \in \mathcal{R}^n} \sum_{P_j \in \sigma_k}\lambda_j\\
			&= F_\mathrm{ccb}(\Lambda),
		\end{split}
	\end{equation}
	which completes the proof.
\end{proof}

\section{Statistical analysis}\label{sec:statis}
Here, we analyze the statistical fluctuation of the CCB protocol with finite sampling. Recall that the CCB fidelity is given by Eq.~\eqref{eq:CCBfidelity},
\begin{equation}
F_\mathrm{ccb}(\Lambda) = \frac{1}{4^n}\sum_j\lambda_j,
\end{equation}
where $\lambda_j$ is the Pauli fidelity of channel $\sqrt{\mathcal{U}^{-1}\Lambda_\textsf{P}\mathcal{U}\Lambda_\textsf{P}}$ related to the Pauli operator $P_j$. It is impractical to solve all the Pauli fidelities via the character RB method when qubit number $n$ grows large. One can sample a finite number of Pauli irreps to estimate the process fidelity of the channel. The reliability of these estimates is expressed by the confidence levels.

In the CCB protocol, the Pauli fidelity is obtained by fitting the survival probability and the gate sequence length, which is very hard for the statistical analysis. For simplicity, we take two points in the fitting diagram to analyze the fluctuation of the slope. In what follows, we use the notation that $\hat{x}$ is an estimator of a quantity, $\bar{x}$, where the bar denotes that either an expected value or a sample average has been taken over realizations of random variable $x$.

We begin by describing the CCB protocol in a statistical way.
\begin{enumerate}
\item 
Choose Pauli operator $P_j \in \textsf{P}_n$.

\item 
Choose positive integer $m_1$

\item 
Choose random gate sequence $s_1$ from gate set $\mathbb{S}_{m_1}$ and obtain an estimate of the $P_j$-weighted probability, $\hat{f}_j(m_1, s_1)$;

\item Repeat step 3 $K_1$ times to estimate
\begin{equation}
\hat{f}_j(m_1) = \frac{1}{K_1}\sum_{s_1} \hat{f}_j(m_1, s_1),
\end{equation}
where the $K_1$ gate sequences form a gate sequence set, $\mathbb{S}_1 \subset \mathbb{S}_{m_1}$.

\item 
Choose another bigger positive integer, $m_2 > m_1$;

\item 
Choose random gate sequence $s_2$ from gate sequence set $\mathbb{S}_{m_2}$ and obtain an estimate of the $P_j$-weighted probability, $\hat{f}_j(m_2, s_2)$;

\item Repeat step 6 $K_2$ times to estimate
\begin{equation}
\hat{f}_j(m_2) = \frac{1}{K_2}\sum_{s_2} \hat{f}_j(m_2, s_2),
\end{equation}
where the $K_2$ gate sequences form a gate sequence set, $\mathbb{S}_2 \subset \mathbb{S}_{m_2}$.

\item Estimate the Pauli fidelity
\begin{equation}\label{eq:ratio_estimator}
\hat{\lambda}_j = \left(\frac{\hat{f}_j(m_2)}{\hat{f}_j(m_1)}\right)^{\frac{1}{m_2 - m_1}}.
\end{equation}

\item 
Sample $M$ Pauli operators $\{P_j\}$ in step 1, and for each $P_j$, repeat steps 2-8. Finally estimate the CCB fidelity
\begin{equation}\label{eq:F_estimator}
\hat{F}_\mathrm{ccb} = \frac{1}{M}\sum_{\{P_j\}} \hat{\lambda}_j.
\end{equation}

\end{enumerate}

The main statistical errors are divided into two parts in the above protocol. The first comes from the sampling randomness of the Pauli fidelity estimation $\hat{\lambda}_j$ for each $P_j$ in Eq.~\eqref{eq:ratio_estimator}. The second comes from the sampling randomness of the Pauli operators $\{P_j\}$ in Eq.~\eqref{eq:F_estimator}. We shall calculate the confidence levels for these two sampling randomness respectively.

We first calculate the bias of the Pauli fidelity estimation
\begin{equation}\label{eq:f_bias}
\Delta \lambda_j := \abs{\mathop{\mathbb{E}}\limits_{\mathbb{S}_1, \mathbb{S}_2}[\hat{\lambda}_j]-\bar{\lambda}_j},
\end{equation}
where $\bar{\lambda}_j$ denotes the theoretical value of the Pauli fidelity,
\begin{equation}
\bar{\lambda}_j = \left(\frac{\bar{f}_j(m_2)}{\bar{f}_j(m_1)}\right)^{\frac{1}{m_2 - m_1}}.
\end{equation}
Here, the expectations of the probability estimators are taken over the gate sequences,
\begin{equation}
\begin{split}
\bar{f}_j(m_1) = \mathop{\mathbb{E}}\limits_{s_1\in\mathbb{S}_{m_1}} [\hat{f}_j(m_1, s_1)], \\
\bar{f}_j(m_2) = \mathop{\mathbb{E}}\limits_{s_2\in\mathbb{S}_{m_2}} [\hat{f}_j(m_2, s_2)].
\end{split}
\end{equation}

In order to calculate the expectation value of the ratio estimator in Eq.~\eqref{eq:f_bias}, we take the expectations over the gate sequence sets $\mathbb{S}_1, \mathbb{S}_2$ on both sides of Eq.~\eqref{eq:ratio_estimator},
\begin{equation}
\mathop{\mathbb{E}}\limits_{\mathbb{S}_1, \mathbb{S}_2}[\hat{\lambda}_j] = \mathop{\mathbb{E}}\limits_{\mathbb{S}_2} [\hat{f}_j(m_2)^{\frac{1}{m_2 - m_1}}]   \mathop{\mathbb{E}}\limits_{\mathbb{S}_1} [\hat{f}_j(m_1)^{-\frac{1}{m_2 - m_1}}].
\end{equation}
The expectations for $\hat{f}_j(m_1)$ and $\hat{f}_j(m_2)$ can be separated since the random variables $s_1$ and $s_2$ are independent.
Denote
\begin{equation}
\begin{split}
& \hat{b} := \hat{f}_j(m_2), \quad \bar{b} := \bar{f}_j(m_2), \\
& \hat{a} := \hat{f}_j(m_1), \quad \bar{a} := \bar{f}_j(m_1), \\
& t := \frac{1}{m_2 - m_1},
\end{split}
\end{equation}
we have
\begin{equation}
\mathop{\mathbb{E}}\limits_{\mathbb{S}_1, \mathbb{S}_2}[\hat{\lambda}_j] = \mathop{\mathbb{E}}\limits_{\mathbb{S}_1} [\hat{a}^{-t}] \mathop{\mathbb{E}}\limits_{\mathbb{S}_2} [\hat{b}^t].
\end{equation}
Supposing
\begin{equation}\label{eq:ab_assumptions}
\begin{split}
\delta_{a} &:= \frac{\hat{a} - \bar{a}}{\bar{a}} \ll 1, \\
\delta_{b} &:= \frac{\hat{b} - \bar{b}}{\bar{b}} \ll 1,
\end{split}
\end{equation}
and using the second-order approximation of the Taylor expansion at $\delta_{a} = 0, \delta_{b} = 0$ for $\hat{a}^{-t}, \hat{b}^t$, respectively, we have
\begin{equation}\label{eq:Taylor_ab}
\begin{split}
\hat{a}^{-t} &= \bar{a}^{-t} (1 + \delta_{a})^{-t} \\
&= \bar{a}^{-t} \left[1 - t\delta_{a} + t(t + 1)\delta_{a}^2 + \mathcal{O}(\delta_{a}^3)\right], \\
\hat{b}^t &= \bar{b}^t (1 + \delta_{b})^t \\
&= \bar{b}^t \left[1 + t\delta_{b} + t(t - 1)\delta_{b}^2 + \mathcal{O}(\delta_{b}^3)\right].
\end{split}
\end{equation}
Take expectations over $s_1, s_2$ for Eq.~\eqref{eq:Taylor_ab},
\begin{equation}
\begin{split}
& \mathop{\mathbb{E}}\limits_{\mathbb{S}_1}[\hat{a}^{-t}] = \bar{a}^{-t} \left[1 + t(t + 1)\frac{\mathrm{Var}[\hat{a}]}{\bar{a}^2} + \mathcal{O}(\delta_{a}^3)\right], \\
& \mathop{\mathbb{E}}\limits_{\mathbb{S}_2}[\hat{b}^t] = \bar{b}^t \left[1 + t(t - 1)\frac{\mathrm{Var}[\hat{b}]}{\bar{b}^2} + \mathcal{O}(\delta_{b}^3)\right].
\end{split}
\end{equation}
Then the expectation value of the ratio estimator is given by,
\begin{equation}\label{eq:f_expect}
\begin{split}
\mathop{\mathbb{E}}\limits_{\mathbb{S}_1, \mathbb{S}_2}[\hat{\lambda}_j] &= \left(\frac{\bar{b}}{\bar{a}}\right)^t \left[1 + t(t + 1)\frac{\mathrm{Var}[\hat{a}]}{\bar{a}^2}\right] \left[1 + t(t - 1)\frac{\mathrm{Var}[\hat{b}]}{\bar{b}^2}\right] + \mathcal{O}(\delta_{a}^3, \delta_{b}^3) \\
&= \bar{\lambda}_j \left[1 + t(t + 1)\frac{\mathrm{Var}[\hat{a}]}{\bar{a}^2} + t(t - 1)\frac{\mathrm{Var}[\hat{b}]}{\bar{b}^2}\right] + \mathcal{O}(\delta_{a}^3, \delta_{b}^3).
\end{split}
\end{equation}
Substituting Eq.~\eqref{eq:f_expect} into Eq.~\eqref{eq:f_bias}, we can derive the bias of the Pauli fidelity estimation,
\begin{equation}\label{eq:lambda_bias}
\Delta \lambda_j = \abs{\bar{\lambda}_j \left[t(t + 1)\frac{\mathrm{Var}[\hat{a}]}{\bar{a}^2} + t(t - 1)\frac{\mathrm{Var}[\hat{b}]}{\bar{b}^2}\right]} + \mathcal{O}(\delta_{a}^3, \delta_{b}^3).
\end{equation}


Recall that the assumptions of Eq.~\eqref{eq:ab_assumptions} are established with specific failure probabilities. In order to calculate the confidence intervals for the aforementioned assumptions, we apply Bernstein’s Inequality,
\begin{lemma}[Bernstein’s inequality \cite{BernsteinIneq}]
Consider a set of $n$ independent random variables $\{X_1, \cdots, X_n\}$ with $\forall i, X_i \leq b$. Let $X = \frac{1}{n}\sum_i X_i$, denote $\mathbb{V}^2 = n^{-1} \sum_{i = 1}^n\mathrm{Var}(X_i)$, then $\forall\epsilon > 0$,
\begin{equation}
\Pr(|X - \mathbb{E}[X]| > \epsilon) \leq 2\exp(-\frac{n\epsilon^2 / 2}{\mathbb{V}^2 + b\epsilon / 3}).
\end{equation}
\end{lemma}

Assume $\forall m, \hat{f}_j(m) \leq 1$, we derive the following proposition.
\begin{proposition}
Let $\hat{f}_j(m)$ be the estimator of $P_j$-weighted probability $\bar{f}_j(m)$ for fixed $K$, $m$, and $j$. Then $\forall\epsilon > 0$,
\begin{equation}
\Pr(|\hat{f}_j(m) - \bar{f}_j(m)| > \epsilon) \leq 2\exp(-\frac{K \epsilon^2 / 2}{\mathrm{Var}[\hat{f}_j(m, s)] + \epsilon / 3}).
\end{equation}
\end{proposition}
Then for the estimators $\hat{a}$ and $\hat{b}$, we can derive the corresponding confidence intervals
\begin{equation}\label{eq:ab_confidence}
\begin{split}
& 
\Pr(|\hat{a} - \bar{a}| > \epsilon_1) \leq 2\exp(-\frac{K_1 \epsilon_1^2 / 2}{\mathrm{Var}[\hat{f}_j(m_1, s_1)] + \epsilon_1 / 3}) := \Upsilon_a(K_1, \epsilon_1), \\
& 
\Pr(|\hat{b} - \bar{b}| > \epsilon_2) \leq 2\exp(-\frac{K_2 \epsilon_2^2}{\mathrm{Var}[\hat{f}_j(m_2, s_2)] + \epsilon_2 / 3}) := \Upsilon_b(K_2, \epsilon_2).
\end{split}
\end{equation}
Substituting Eq.~\eqref{eq:ab_confidence} into Eq.~\eqref{eq:lambda_bias}, we have following lemma

\begin{lemma}\label{thm:ratio_estimator}
For any given Pauli operator $P_j \in \textsf{P}_n$ and some fixed $m_1, m_2, K_1, K2$, the bias of the Pauli fidelity estimation is upper bounded by
\begin{equation}\label{eq:biaseb}
\begin{split}
\Delta\lambda_j 
& \leq \frac{1}{(m_2 - m_1)^2} \abs{
\frac{1 + m_2 - m_1}{K_1\bar{f}_j^2(m_1)} \mathrm{Var}[\hat{f}_j(m_1, s_1)] + \frac{1 - m_2 + m_1}{K_2\bar{f}_j^2(m_2)} \mathrm{Var}[\hat{f}_j(m_2, s_2)]} + \mathcal{O}(\epsilon^3_1, \epsilon^3_2) \\
& := \epsilon_b(m_1, m_2; K_1, K_2; \epsilon_1, \epsilon_2)
\end{split}
\end{equation}
with a failure probability bounded by $\Upsilon_j(K_1, K_2; \epsilon_1, \epsilon_2) = \Upsilon_a(K_1, \epsilon_1) + \Upsilon_b(K_2, \epsilon_2)$.
\end{lemma}

The notation $\epsilon_b(m_1, m_2; K_1, K2; \epsilon_1, \epsilon_2)$ is abbreviated as $\epsilon_b$ in the following analysis. Next we will compute the confidence interval for the CCB fidelity $\hat{F}_\mathrm{ccb}$ defined in Eq.~\eqref{eq:F_estimator}. Denote
\begin{equation}
\begin{split}
\bar{F}'_\mathrm{ccb} &= \frac{1}{M}\sum_{\{P_j\}} \bar{\lambda}_j, \\
\bar{F}_\mathrm{ccb} &= \frac{1}{4^n} \sum_j \bar{\lambda}_j.
\end{split}
\end{equation}

Assume the Pauli fidelity $0 \leq \bar{\lambda}_j \leq 1$ for all $P_j$, one can apply the Hoeffding's inequality directly, given by
\begin{lemma}[Hoeffding's inequality \cite{HoeffdingIneq}]
Consider a set of $n$ independent random variables $\{X_1, \cdots, X_n\}$ and $\forall i, a_i \leq X_i \leq b_i$. Let $X = \frac{1}{n}\sum_i X_i$, then $\forall\epsilon > 0$,
\begin{equation}
\Pr(|X - \mathbb{E}[X]| > \epsilon) \leq 2\exp(-\frac{2n^2\epsilon^2}{\sum_{i = 1}^n (b_i - a_i)^2}).
\end{equation}
\end{lemma}

Then we can derive that $\forall\epsilon_M > 0$,
\begin{equation}\label{eq:F_confidence1}
\Pr(|\bar{F}_\mathrm{ccb} - \bar{F}'_\mathrm{ccb}| > \epsilon_M) \leq 2 \exp(-2M \epsilon_M^2).
\end{equation}

According to Lemma \ref{thm:ratio_estimator}, we have
\begin{equation}\label{eq:F_confidence2}
\Pr(|\hat{F}_\mathrm{ccb} - \bar{F}'_\mathrm{ccb}| > \epsilon_b) \leq M \Upsilon_j(K_1, K_2; \epsilon_1, \epsilon_2),
\end{equation}
where $\hat{F}_\mathrm{ccb}$ is the estimator of $\bar{F}_\mathrm{ccb}$, as defined in Eq.~\eqref{eq:F_estimator}. Combine Eq.~\eqref{eq:F_confidence1} and \eqref{eq:F_confidence2} and apply the union bound, we can compute the confidence interval for the process fidelity estimator
\begin{equation}\label{eq:F_confidence_intercal}
\Pr(|\hat{F}_\mathrm{ccb} - \bar{F}_\mathrm{ccb}| > \epsilon_M + \epsilon_b) \leq 2\exp(-M \epsilon_M^2 / 2) + M  \Upsilon_j(K_1, K_2; \epsilon_1, \epsilon_2).
\end{equation}
Let us assume that
\begin{equation}
\begin{split}
& 0 \leq \mathrm{Var}[\hat{f}_j(m_1, s_1)] \leq 1, \\ 
& 0 \leq \mathrm{Var}[\hat{f}_j(m_2, s_2)] \leq 1, 
\end{split}
\end{equation}
then Eq.~\eqref{eq:F_confidence_intercal} can be simplified to
\begin{equation}
\Pr(|\hat{F}_\mathrm{ccb} - \bar{F}_\mathrm{ccb}| > \epsilon_M + \epsilon_b) \leq  2\exp(-M \epsilon_M^2 / 2) + 2M \exp(-\frac{K_1 \epsilon_1^2 / 2}{1 + \epsilon_1 / 3}) + 2M \exp(-\frac{K_2 \epsilon_2^2 / 2}{1 + \epsilon_2 / 3}).
\end{equation}

We further assume that
\begin{equation}
\begin{split}
\frac{1}{2} &< \bar{f}_j(m_1) < 1, \\
\end{split}
\end{equation}
then we can derive the following theorem.

\begin{theorem}\label{thm:M_number}
Consider a CCB implementation with sampling numbers $K_1, K_2$ at sequence lengths $m_1, m_2$ with estimation errors $\epsilon_1, \epsilon_2$, respectively, and the expected survival probability $\bar{f}_j(m_1)$ that satisfies $\forall j, 1/2 < \bar{f}_j(m_1) < 1$. The estimated CCB fidelity, $\hat{F}_\mathrm{ccb}$, is given by the average over $M$ Pauli fidelities $\{\lambda_j\}$. The confidence probability for the CCB fidelity falling into the estimated interval $[\hat{F}_\mathrm{ccb} - \epsilon_M - \epsilon_b, \hat{F}_\mathrm{ccb} + \epsilon_M + \epsilon_b]$ is greater than $1 - \delta$, 
\begin{equation}
\begin{split}
&\Pr(|\hat{F}_\mathrm{ccb} - \bar{F}_\mathrm{ccb}| \leq \epsilon_M + \epsilon_b)\geq 1 - \delta, \\
\end{split}
\end{equation}
with
\begin{equation}
	\begin{split}
&\epsilon_b \leq \frac{4}{K_1} \frac{1}{m_2 - m_1} (\frac{1}{m_2 - m_1} + 1) + \mathcal{O}(\epsilon^3_1, \epsilon^3_2),
	\end{split}
\end{equation}
and 
\begin{equation}
2\exp(-2M \epsilon_M^2) + 2M \exp(-\frac{K_1 \epsilon_1^2 / 2}{1 + \epsilon_1 / 3}) + 2M \exp(-\frac{K_2 \epsilon_2^2 / 2}{1 + \epsilon_2 / 3}) = \delta.
\end{equation}
\end{theorem}

An simplified informal version of the above theorem is shown in the main text as Theorem \ref{thm:CCB_sample}.

\section{Simulation}\label{sec:simulation}
Here, we shall present the noise model for the simulation in the main text and give more details on the 2-qubit controlled-$(TX)$ gate and the 5-qubit error correcting circuit. In addition, we compare a CAB process and an interleaved character randomized benchmarking (ICRB) process~\cite{Helsen2019characterRB,Xue2019CRB} for benchmarking the 2-qubit CZ gate. We provide its simulation details and the results in the end of this part.

\subsection{Error model}\label{append:errormodel}
The noise channel, $\Lambda_t$, we consider here for the target gate is composed of Pauli channel $\Lambda_0$, amplitude damping channel $\Lambda_1$, and qubit-qubit correlation channel $\Lambda_2$, $\Lambda_t = \Lambda_0 \circ \Lambda_1 \circ \Lambda_2$. The noise channel, $\Lambda_\mathrm{ref}$, we consider for the twirling gate set $\textsf{P}_n$ is a gate-independent Pauli channel, which is negligible compared with $\Lambda_t$.

\begin{enumerate}
\item Stochastic Pauli channel $\Lambda_0$.

Pauli channel $\Lambda_0$ can be written as
\begin{equation}
\Lambda_0(\rho) = \sum_i p_i P_i \rho P_i^{-1},
\end{equation}
where $p_i$ is the Pauli error rate related to the Pauli operator $P_i$. As for $\Lambda_0$, it is equivalent to say that operator $P_i \in \textsf{P}_n$ applies on the density matrix $\rho$ with probability $p_i$. We can further rewrite $\Lambda_0$ in the PTM representation,
\begin{equation}
\Lambda_0 = \sum_i \lambda_i \lket{P_i}\lbra{P_i},
\end{equation}
where $\lambda_i$ is the Pauli fidelity. Pauli channel $\Lambda_0$ in the simulation contains dephasing errors and cross-talk errors. In reality, the fidelity of error channel $\Lambda_t$ is mainly determined by $\Lambda_0$.

\item Amplitude damping channel $\Lambda_1$.

Each qubit in the simulation is subject to an amplitude damping channel,
\begin{equation}
\Lambda_1 = \bigotimes_{i=1}^n \Lambda_i^d,
\end{equation}
where $\Lambda_i^d$ is the  single-qubit damping channel for qubit $i$,
\begin{equation}
\begin{split}
\Lambda_i^d &= K_i^{(0)} \rho K_i^{(0)\dag} + K_i^{(1)} \rho K_i^{(1)\dag}, \\
K_i^{(0)} &= \begin{pmatrix} 1 & 0\\ 0 & \sqrt{1-\alpha_i} \end{pmatrix}, \\
K_i^{(1)} &= \begin{pmatrix} 0 & \sqrt{\alpha_i}\\ 0 & 0 \end{pmatrix},
\end{split}
\end{equation}
with damping parameter $\alpha_i$.

\item Qubit-qubit correlation channel $\Lambda_2$.

The qubit-qubit correlation channel $\Lambda_2$ is a coherent error channel in the simulation,
\begin{equation}
\Lambda_2 = \bigotimes_{i<j} e^{i \beta_{ij} \mathrm{SWAP}_{ij}},
\end{equation}
where $\beta_{ij}$ is the correlation parameter describing the interacting strength between qubits $i$ and $j$.
\end{enumerate}

\subsection{Simulations for the controlled-$(TX)$ gate}\label{append:2qubit}
The controlled-$(TX)$ gate can be decomposed as
\begin{equation}
\mathrm{CTX} = (I \otimes \sqrt{T}) \mathrm{CNOT} (I \otimes \sqrt{T}^{-1}).
\end{equation}
Then, we can take $I \otimes \sqrt{T}$ as the local gauge transformation and the twirling gate set turns to
\begin{equation}
\textsf{P}_\mathrm{ctx} = (I \otimes \sqrt{T}) \textsf{P}_2 (I \otimes \sqrt{T}^{-1}).
\end{equation}
Consider noise channel $\Lambda_t = \Lambda_0 \circ \Lambda_1 \circ \Lambda_2$ in Section \ref{append:errormodel} for the noisy controlled-$(TX)$ gate. We randomly sample the Pauli fidelities of $\Lambda_0$ from a normal distribution $\mathcal{N}(\mu, \sigma)$, denoted as a $\mathcal{N}(\mu, \sigma)$-Pauli channel, where $\mu$ and $\sigma$ are the mean value and standard deviation. The parameters for $\Lambda_1$ is set to $\alpha_1 = \alpha_2 = 0.005$. The parameter for $\Lambda_2$ is set to $\beta_{12} = 0.01$. In the following discussions, we label the noise channel for the controlled-$(TX)$ with $\Lambda_t(\mu, \sigma)$, since $\Lambda_1$ and $\Lambda_2$ remain the same in all the simulations.

We simulate the CAB and CCB protocols for the noisy  controlled-$(TX)$ gate with 8 different noise channels $\{\Lambda_t(\mu_i, \sigma_i)\}$, the noisy implementations of $\mathrm{CTX}$ and $\mathrm{CTX}^{-1}$ are given by
\begin{equation}
\begin{split}
&\tilde{\mathrm{CTX}} = \mathrm{CTX} \circ \Lambda_t(\mu, \sigma), \\
&\tilde{\mathrm{CTX}^{-1}} = \mathrm{CTX}^{-1} \circ \Lambda_t(\mu, \sigma).
\end{split}
\end{equation}
The error parameters are taken as $\{(\mu_i, \sigma_i)\} =$ \{(0.995, 0.001), (0.990, 0.002), (0.980, 0.003), (0.970, 0.004), (0.960, 0.005), (0.950, 0.006), (0.940, 0.007), (0.930, 0.008)\}. Take the $\mathcal{N}(0.998, 0.001)$-Pauli channel as for the noise channel of the twirling gate set $\Lambda_\mathrm{ref}$ and then denote the noisy implementation of the twirling gate set as,
\begin{equation}
\tilde{\textsf{P}}_\mathrm{ctx} = \Lambda_\mathrm{ref} \circ \textsf{P}_\mathrm{ctx}.
\end{equation}
Take the $\mathcal{N}(0.998, 0.001)$-Pauli channel as for SPAM error channel $\Lambda_\mathrm{spam}$ and then  denote the noisy implementations of the initial state $\ket{\psi}$ and measurement $Q$ as
\begin{equation}
\begin{split}
\rho_\psi &= \Lambda_\mathrm{spam}(\ketbra{\psi}), \\
\tilde{Q} &= Q \circ \Lambda_\mathrm{spam}.
\end{split}
\end{equation}

The simulation procedures for CAB run as follows.
\begin{enumerate}
\item 
For each noise channel $\Lambda_t(\mu_i, \sigma_i)$, select a set of sequence length $\{m\} = \{1, 2, \cdots, m_\mathrm{max}\}$, where $m_\mathrm{max}$ satisfies $\mu_i^{m_\mathrm{max}} \approx \mu_i / 3$.

\item For each sequence length $m$, sample $K = 50$ random gate sequences $\{(C, P_\mathrm{ctx}^{(1)}, \cdots, P_\mathrm{ctx}^{(2m)})\}$, where $C$ and $P_\mathrm{ctx}^{(i)} (1 \leq i \leq 2m)$ are sampled uniformly at random from $\textsf{C}_1^{\otimes 2}$ and $\textsf{P}_\mathrm{ctx}$, respectively. For each gate sequence, the noisy implementation in PTM is given by
\begin{equation}
\tilde{\mathcal{S}}_\mathrm{cab} = \Lambda_\mathrm{ref}\mathcal{C}^{-1}  \Lambda_\mathrm{ref}  \mathcal{U}_{\mathrm{inv}} \mathrm{CTX}^{-1}  \Lambda_t \Lambda_\mathrm{ref}  \mathcal{P}_\mathrm{ctx}^{(2m)} \cdots \mathcal{P}_\mathrm{ctx}^{(2)}\mathrm{CTX}  \Lambda_t \Lambda_\mathrm{ref}  \mathcal{P}_\mathrm{ctx}^{(1)}\mathcal{C},
\end{equation}
where the inverse gate is given by $\mathcal{U}_{\mathrm{inv}} = (\mathrm{CTX}^{-1} P_\mathrm{ctx}^{(2m)} \cdots \mathrm{CTX}^{-1} P_\mathrm{ctx}^{(2)} \mathrm{CTX} P_\mathrm{ctx}^{(1)})^{-1}$.

\item Compute the survival probability over the $K = 50$ gate sequences for each measurement observable $Q_k \in \{II, IZ, ZI, ZZ\}$
\begin{equation}
f_k(m) = \frac{1}{K}\sum_{\mathcal{S}_\mathrm{cab}} \lbra{\tilde{Q}_k} \tilde{\mathcal{S}}_\mathrm{cab} \lket{\rho_\psi},
\end{equation}
where $\ket{\psi} = \ket{0}^{\otimes 2}$.

\item For each $Q_k$, fit $f_k(m)$ to the function
\begin{equation}
f_k(m) = A_k\mu_k^{2m}.
\end{equation}

\item Estimate the CAB fidelity as
\begin{equation}
F_\mathrm{cab} = \frac{1}{16} (\mu_{II} + 3\mu_{IZ} + 3\mu_{ZI} + 9\mu_{ZZ}).
\end{equation}
\end{enumerate}

The simulation procedures for CCB run as follows.
\begin{enumerate}
\item For each noise channel $\Lambda_t(\mu_i, \sigma_i)$, select a set of sequence length $\{m\} = \{1, 2, \cdots, m_\mathrm{max}\}$, where $m_\mathrm{max}$ satisfies $\mu_i^{m_\mathrm{max}} \approx \mu_i / 3$.

\item Sample $M = 10$ operators $\{P_j\}$ uniformly at random from $\textsf{P}_\mathrm{ctx}$.

\item For each $P_j$ at each sequence length $m$, sample $K = 50$ random gate sequences $\{(P_\mathrm{ctx}^{(0)}, P_\mathrm{ctx}^{(1)}, \cdots, P_\mathrm{ctx}^{(2m)})\}$, where $P_\mathrm{ctx}^{(i)}$ are sampled uniformly at random from $\textsf{P}_\mathrm{ctx}$. For each gate sequence, the noisy implementation in PTM is given by
\begin{equation}
\tilde{\mathcal{S}}_\mathrm{ccb} = \Lambda_\mathrm{ref}  \mathcal{U}_{\mathrm{inv}} \mathrm{CTX}^{-1}  \Lambda_t \Lambda_\mathrm{ref}  \mathcal{P}_\mathrm{ctx}^{(2m)} \cdots \mathcal{P}_\mathrm{ctx}^{(2)}\mathrm{CTX}  \Lambda_t \Lambda_\mathrm{ref}  \mathcal{P}_\mathrm{ctx}^{(1)}\mathcal{P}_\mathrm{ctx}^{(0)},
\end{equation}
where the inverse gate is given by $\mathcal{U}_{\mathrm{inv}} = (\mathrm{CTX}^{-1} P_\mathrm{ctx}^{(2m)} \cdots \mathrm{CTX}^{-1} P_\mathrm{ctx}^{(2)} \mathrm{CTX} P_\mathrm{ctx}^{(1)})^{-1}$.

\item 
Compute the survival probability over the $K = 50$ gate sequences for each measurement $P_j$,
\begin{equation}
f_j(m) = \frac{1}{K}\sum_{\mathcal{S}_\mathrm{ccb}} \chi_j(P^{(0)}) \lbra{P_j} \Lambda_\mathrm{spam}  \tilde{\mathcal{S}}_\mathrm{ccb} \lket{\rho_\psi},
\end{equation}
where $\chi_j(P^{(0)}) = 1$ if $P^{(0)}$ commutes with $P_j$ and -1 otherwise, and $\ket{\psi}$ is the $+1$ eigenstate of $P_j$.

\item Fit $f_j(m)$ to the function
\begin{equation}
f_j(m) = A_j\lambda_j^{2m}.
\end{equation}

\item Estimate the CCB fidelity as
\begin{equation}
F_\mathrm{ccb} = \frac{1}{M}\sum_{\{P_j\}} \lambda_j.
\end{equation}
\end{enumerate}

\subsection{Simulations for the 5-qubit error correcting circuit}\label{append:5qubit}
Here, we take the 5-qubit stabilizer encoding circuit shown in Fig.~\ref{fig:5QubitCircuit} as the target gate $U$, which only contains Clifford gates.

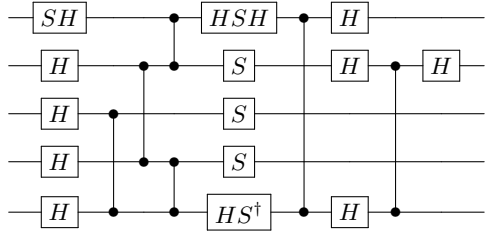
\begin{figure}[!htbp]
	\[
	\Qcircuit @C=1em @R=.7em {
		\lstick{} & \gate{SH} & \qw & \qw & \ctrl{1} & \gate{HSH} & \ctrl{4} & \gate{H} & \qw & \qw & \qw \\
		\lstick{} & \gate{H} & \qw & \ctrl{2} & \control \qw & \gate{S} & \qw & \gate{H} & \ctrl{3} & \gate{H} & \qw \\
		\lstick{} & \gate{H} & \ctrl{2} & \qw & \qw & \gate{S}  & \qw & \qw & \qw & \qw & \qw \\
		\lstick{} & \gate{H} & \qw & \control \qw & \ctrl{1} & \gate{S} & \qw & \qw & \qw & \qw & \qw \\
		\lstick{} & \gate{H} & \control \qw & \qw & \control \qw & \gate{HS^{\dagger}} & \control \qw & \gate{H} & \control \qw & \qw & \qw
	}
	\]
\caption{5-qubit stabilizer encoding circuit.} \label{fig:5QubitCircuit}
\end{figure}

The noise channel we consider here has the form $\Lambda_t = \Lambda_0 \circ \Lambda_1 \circ \Lambda_2$, as presented in Section \ref{append:errormodel}. We set the $\mathcal{N}(0.98, 0)$-Pauli channel as $\Lambda_0$, which can be seen as a depolarizing channel, $\Lambda_\mathrm{dep}(\rho) = p\rho + (1-p)I/d$ where $p = 0.98$. The error parameters $\{\alpha_i\}$ and $\{\beta_{ij}\}$ are sampled uniformly at random from the intervals [0, 0.02] and [0, 0.01], respectively. Then the noisy implementation of the target gate is given by
\begin{equation}
\tilde{\mathcal{U}} = \mathcal{U} \circ \Lambda_t.
\end{equation}

Since the 5-qubit encoding circuit is a Clifford gate, we take the 5-qubit Pauli group $\textsf{P}_5$ as the twirling gate set. For simplicity, we ignore the errors from the SPAM and twirling gates, $\Lambda_\mathrm{spam} = I$ and $\Lambda_\mathrm{ref} = I$.

The simulation procedures for CAB run as follows.
\begin{enumerate}
\item 
For 5-qubit noise channel $\Lambda_t$, select a set of sequence lengths, $\{m\} = \{1, 2, \cdots, 20\}$.

\item For each sequence length $m$, sample $K$ random gate sequences $\{(C, P^{(1)}, \cdots, P^{(2m)})\}$, where $C$ and $P^{(i)}$ are sampled uniformly at random from $\textsf{C}_1^{\otimes 5}$ and $\textsf{P}_5$, respectively. For each gate sequence, the noisy implementation in PTM is given by
\begin{equation}
\tilde{\mathcal{S}}_\mathrm{cab} = \mathcal{C}^{-1} \mathcal{U}_{\mathrm{inv}} \mathcal{U}^{-1} \Lambda_t \mathcal{P}^{(2m)} \cdots \mathcal{P}^{(2)}\mathcal{U} \Lambda_t \mathcal{P}^{(1)}\mathcal{C},
\end{equation}
where the inverse gate is given by $\mathcal{U}_{\mathrm{inv}} = (\mathcal{U}^{-1} P^{(2m)} \cdots \mathcal{U}^{-1} P^{(2)} \mathcal{U} P^{(1)})^{-1}$.

\item 
Compute the survival probability over the $K$ gate sequences for each measurement observable $Q_k \in \{I, Z\}^{\otimes 5}$
\begin{equation}\label{eq:survprob}
f_k(m) = \frac{1}{K}\sum_{\mathcal{S}_\mathrm{cab}} \lbra{\tilde{Q}_k} \tilde{\mathcal{S}}_\mathrm{cab} \lket{\rho_\psi},
\end{equation}
where $\ket{psi} = \ket{0}^{\otimes 5}$.

\item For each $Q_k$, fit $f_k(m)$ to the function
\begin{equation}
f_k(m) = A_k\mu_k^{2m}.
\end{equation}

\item 
Estimate the CAB fidelity as
\begin{equation}
F_\mathrm{cab} = \frac{1}{4^5}\sum_k d_k\mu_k,
\end{equation}
where $d_k$ is the dimension of the $k$th irrep, defined in Eq.~\eqref{eq:CABfidelitysum}.
\end{enumerate}

Note that the target gate is a Clifford gate. In the simulation, we obtain the inverse gate, $U_{inv} = \circ_{i=1}^{20} U^{-1}P_{2i} UP_{2i-1}$, by considering the permutation relation of a Clifford gate acting on Pauli operators. Furthermore, the survival probability for observable $Q_k$ can be obtained from the $Z$ basis measurement result. To be specific, $Q_k$ can be expressed as the linear combination of $\ketbra{z}$, $Q_k = \sum_z t_{kz} \ketbra{z}$, where $z\in \{0,1\}^5$. Then the survival probability in Eq.~\eqref{eq:survprob} is given by
\begin{equation}
\begin{split}
f_k(m) &= \frac{1}{K}\sum_{\mathcal{S}_\mathrm{cab}} \lbra{\tilde{Q}_k} \tilde{\mathcal{S}}_\mathrm{cab} \lket{\rho_\psi}\\
&= \frac{1}{K}\sum_{S_\mathrm{cab}} \Tr[Q_k \tilde{S}_\mathrm{cab} (\ketbra{0})]\\
&= \frac{1}{K}\sum_{S_\mathrm{cab}} \sum_z t_{kz} \Tr[\ketbra{z} \tilde{S}_\mathrm{cab} (\ketbra{0})],
\end{split}
\end{equation}
where $S_\mathrm{cab}$ represents the channel of $\mathcal{S}_\mathrm{cab}$.

For comparison, we also simulate the XEB protocol to benchmark the fidelity of the 5-qubit error correcting circuit with the same noise channel. The simulation procedures for XEB run as follows.
\begin{enumerate}
\item 
For 5-qubit noise channel $\Lambda_t$, select a set of sequence lengths, $\{m\} = \{1, 2, \cdots, 20\}$.
	
\item 
For each sequence length $m$, sample $K$ random gate sequences $\{(C^{(1)}, \cdots, C^{(2m)})\}$, where $C^{(i)}$ are sampled uniformly at random from $\textsf{C}_1^{\otimes 5}$. For each gate sequence, the noisy implementation in PTM is given by
\begin{equation}
\tilde{\mathcal{S}}_\mathrm{xeb} = \mathcal{U}  \Lambda_t  \mathcal{C}^{(2m)} \cdots \Lambda_t  \mathcal{C}^{(2)}\mathcal{U}  \Lambda_t  \mathcal{C}^{(1)}.
\end{equation}

\item Compute the $Z$ basis measurement result for each gate sequence $\mathcal{S}_\mathrm{xeb}$, for $z\in \{0,1\}^5$,
\begin{equation}
f_\mathrm{xeb}(m, z) = \lbra{z}\tilde{\mathcal{S}}_\mathrm{xeb} \lket{\rho_\psi},
\end{equation}
where $\ket{\psi} = \ket{0}^{\otimes 5}$.
	
\item Compute the ideal $Z$ basis measurement result if the gate is ideal for each gate sequence, for $z\in \{0,1\}^5$,
\begin{equation}
f'_\mathrm{xeb}(m, z) =  \lbra{z} \mathcal{U} \mathcal{C}^{(2m)} \cdots \mathcal{C}^{(2)}\mathcal{U} \mathcal{C}^{(1)} \lket{\rho_\psi}.
\end{equation}

\item Compute the average XEB fidelity over the $K$ gate sequences
\begin{equation}
f_\mathrm{xeb}(m) = \frac{1}{K}\sum_{\mathcal{S}_\mathrm{xeb}} \frac{2^n\sum_z f_\mathrm{xeb}(m, z)f'_\mathrm{xeb}(m, z)-1}{2^n \sum_z f_\mathrm{xeb}(m, z)^2-1}.
\end{equation}
	
\item Fit $f_\mathrm{xeb}(m)$ to the function
\begin{equation}
f_\mathrm{xeb}(m) = A_\mathrm{xeb}p^{2m} + B_\mathrm{xeb}.
\end{equation}
The estimation of process fidelity with XEB is given by $p + (1-p)/d^2$.  
\end{enumerate}

\subsection{Comparison between CAB and ICRB for benchmarking a CZ gate}
In this part, we present the comparison of CAB and ICRB for benchmarking a CZ gate. The noise model is the same as the one in Appendix~\ref{append:errormodel}, including a Pauli channel, amplitude damping channel, and a qubit-qubit correlation channel. Specifically, we set the $\mathcal{N}(0.99,0)$-Pauli channel as $\Lambda_0$. For the amplitude damping channel $\Lambda_1$ and the qubit-qubit correlation channel $\Lambda_2$, the error parameters $\{\alpha_i\}$ and $\{\beta_{ij}\}$ are sampled uniformly at random from the intervals [0, 0.01] and [0, 0.01], respectively. Then the noisy implementation of CZ gate is set as
\begin{equation}
    \widetilde{\mathrm{CZ}} = \mathrm{CZ}\circ \Lambda_2 \circ \Lambda_1 \circ \Lambda_0.
\end{equation}
The fidelity of the generated noise channel $\Lambda_2 \circ \Lambda_1 \circ \Lambda_0$ is 0.9864 in our simulation, matching the fidelity of a real CZ gate in experiments. We consider an error probability of 0.02 for state preparation. That is, for each qubit, the actual prepared state is $0.98\ketbra{0}+0.02\ketbra{1}$ for $\ket{0}$ and is $0.02\ketbra{0}+0.98\ketbra{1}$ for $\ket{1}$. For simplicity, we set the measurement to be perfect computational basis measurement.

The simulation of the CAB process is similar to previous cases. The only difference is that we consider the effect of finite number of measurements in this case. That means, after computing the computational basis measurement probability for one sequence, we use this probability to generate a measurement frequency associated with a fixed number of single-shot measurements. Then the measurement frequency is used to compute the survival probability as well as fidelity. For simulating ICRB~\cite{Helsen2019characterRB,Xue2019CRB}, we consider a same measurement setting and take the same number of single-shot measurements for one sample sequence.

In this simulation, the circuit depth of CAB is set as $\{1, 2, 5, 10, 20, 50\}$, corresponding to an overall CZ gate number $\{2, 4, 10, 20, 40, 100\}$. In order to ensure that the numbers of implemented CZ gates in the two protocols are the same, the circuit depth of ICRB is set as $\{2, 4, 10, 20, 40, 100\}$. The sample complexity for each circuit depth is taken from $\{5, 10, 25, 50, 100, 200\}$ and the number of single shots for one sample sequence is taken from $\{100, 200\}$. Then the amount of single-shot measurements for CAB is
\begin{equation}
\text{the number of circuit depths}\times \text{sample complexity for one depth}\times \text{single shots for one sequence}.    
\end{equation}
Due to the additional sample complexity brought by the character gate, the amount of single-shot measurements for ICRB is
\begin{equation}
2^n\times \text{the number of circuit depths}\times \text{sample complexity for one depth}\times \text{single shots for one sequence}.    
\end{equation}
It can be seen that the factor $2^n$ associated with the character gate is an exponential overhead. We separately simulate 50 experiments for each setting and compute the mean and standard deviation of the fidelity of these 50 experiments. The simulation results of the two protocols are shown in Fig.~\ref{fig:CABvsICRB}. The results show that CAB and ICRB have a similar accuracy in benchmarking CZ and CAB is a little better than ICRB in the sense of standard deviation. That means, to estimate the process fidelity of CZ within a given precision, the amount of single-shot measurements for CAB is less than that for ICRB, showing an advantage of CAB in benchmarking experiments.

\begin{figure}[htbp!]
\centering

\subfigure[]{
	\centering
	\includegraphics[width=8cm]{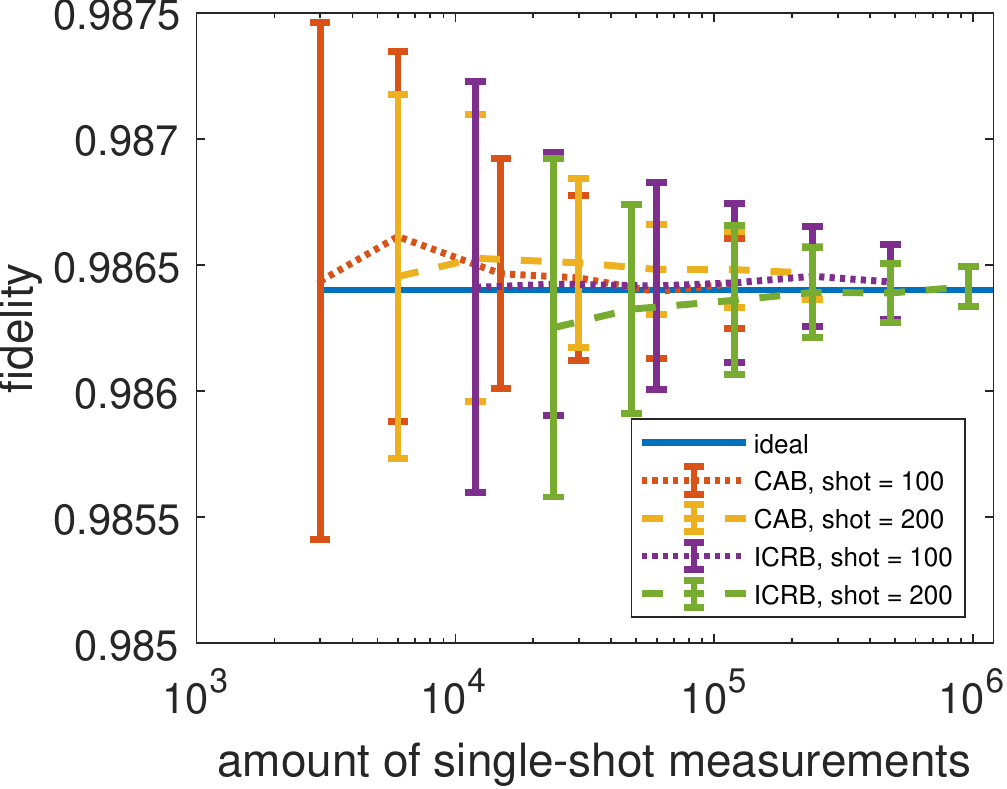}
	\label{fig:CABvsICRBa}
}
\subfigure[]{
	\centering
	\includegraphics[width=8cm]{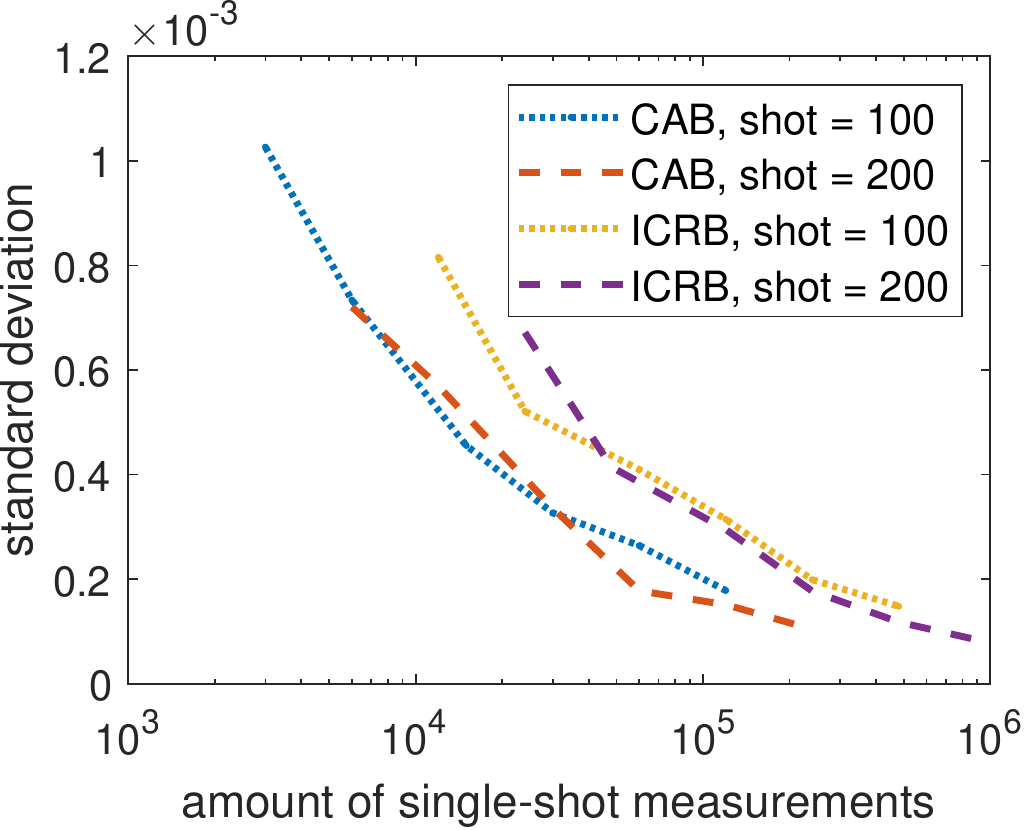}
	\label{fig:CABvsICRBb}
}
\caption{Simulation results of CAB and ICRB for benchmarking a CZ gate. (a) The error bar of the fidelity obtained from 50 experiments via CAB and ICRB protocols. The x-axis represents the amount of single-shot measurements and the y-axis represents the fidelity. The term ``shot" represents the number of single shots associated with one sample sequence. One can see that CAB and ICRB can both obtain an accurate estimation of the process fidelity of a CZ gate. (b) The standard deviation of the fidelity obtained from 50 experiments via CAB and ICRB protocols. The x-axis represents the amount of single-shot measurements and the y-axis represents the standard deviation of the fidelity. The figure shows that CAB has an advantage than ICRB in the sense of standard deviation.}
\label{fig:CABvsICRB}
\end{figure}

\bibliography{bibScalableRB.bib}

\begin{thebibliography}{36}%
\makeatletter
\providecommand \@ifxundefined [1]{%
 \@ifx{#1\undefined}
}%
\providecommand \@ifnum [1]{%
 \ifnum #1\expandafter \@firstoftwo
 \else \expandafter \@secondoftwo
 \fi
}%
\providecommand \@ifx [1]{%
 \ifx #1\expandafter \@firstoftwo
 \else \expandafter \@secondoftwo
 \fi
}%
\providecommand \natexlab [1]{#1}%
\providecommand \enquote  [1]{``#1''}%
\providecommand \bibnamefont  [1]{#1}%
\providecommand \bibfnamefont [1]{#1}%
\providecommand \citenamefont [1]{#1}%
\providecommand \href@noop [0]{\@secondoftwo}%
\providecommand \href [0]{\begingroup \@sanitize@url \@href}%
\providecommand \@href[1]{\@@startlink{#1}\@@href}%
\providecommand \@@href[1]{\endgroup#1\@@endlink}%
\providecommand \@sanitize@url [0]{\catcode `\\12\catcode `\$12\catcode
  `\&12\catcode `\#12\catcode `\^12\catcode `\_12\catcode `\%12\relax}%
\providecommand \@@startlink[1]{}%
\providecommand \@@endlink[0]{}%
\providecommand \url  [0]{\begingroup\@sanitize@url \@url }%
\providecommand \@url [1]{\endgroup\@href {#1}{\urlprefix }}%
\providecommand \urlprefix  [0]{URL }%
\providecommand \Eprint [0]{\href }%
\providecommand \doibase [0]{https://doi.org/}%
\providecommand \selectlanguage [0]{\@gobble}%
\providecommand \bibinfo  [0]{\@secondoftwo}%
\providecommand \bibfield  [0]{\@secondoftwo}%
\providecommand \translation [1]{[#1]}%
\providecommand \BibitemOpen [0]{}%
\providecommand \bibitemStop [0]{}%
\providecommand \bibitemNoStop [0]{.\EOS\space}%
\providecommand \EOS [0]{\spacefactor3000\relax}%
\providecommand \BibitemShut  [1]{\csname bibitem#1\endcsname}%
\let\auto@bib@innerbib\@empty
\bibitem [{\citenamefont {Preskill}(2018)}]{Preskill2018NISQ}%
  \BibitemOpen
  \bibfield  {author} {\bibinfo {author} {\bibfnamefont {J.}~\bibnamefont
  {Preskill}},\ }\bibfield  {title} {\bibinfo {title} {Quantum {C}omputing in
  the {NISQ} era and beyond},\ }\href
  {https://doi.org/10.22331/q-2018-08-06-79} {\bibfield  {journal} {\bibinfo
  {journal} {{Quantum}}\ }\textbf {\bibinfo {volume} {2}},\ \bibinfo {pages}
  {79} (\bibinfo {year} {2018})}\BibitemShut {NoStop}%
\bibitem [{\citenamefont {Vandersypen}\ and\ \citenamefont
  {Chuang}(2005)}]{RevModPhys.76.1037}%
  \BibitemOpen
  \bibfield  {author} {\bibinfo {author} {\bibfnamefont {L.~M.~K.}\
  \bibnamefont {Vandersypen}}\ and\ \bibinfo {author} {\bibfnamefont {I.~L.}\
  \bibnamefont {Chuang}},\ }\bibfield  {title} {\bibinfo {title} {Nmr
  techniques for quantum control and computation},\ }\href
  {https://doi.org/10.1103/RevModPhys.76.1037} {\bibfield  {journal} {\bibinfo
  {journal} {Rev. Mod. Phys.}\ }\textbf {\bibinfo {volume} {76}},\ \bibinfo
  {pages} {1037} (\bibinfo {year} {2005})}\BibitemShut {NoStop}%
\bibitem [{\citenamefont {Chu}(2002)}]{Chu2002}%
  \BibitemOpen
  \bibfield  {author} {\bibinfo {author} {\bibfnamefont {S.}~\bibnamefont
  {Chu}},\ }\bibfield  {title} {\bibinfo {title} {Cold atoms and quantum
  control},\ }\href {https://doi.org/10.1038/416206a} {\bibfield  {journal}
  {\bibinfo  {journal} {Nature}\ }\textbf {\bibinfo {volume} {416}},\ \bibinfo
  {pages} {206} (\bibinfo {year} {2002})}\BibitemShut {NoStop}%
\bibitem [{\citenamefont {Terhal}(2015)}]{RevModPhys.87.307}%
  \BibitemOpen
  \bibfield  {author} {\bibinfo {author} {\bibfnamefont {B.~M.}\ \bibnamefont
  {Terhal}},\ }\bibfield  {title} {\bibinfo {title} {Quantum error correction
  for quantum memories},\ }\href {https://doi.org/10.1103/RevModPhys.87.307}
  {\bibfield  {journal} {\bibinfo  {journal} {Rev. Mod. Phys.}\ }\textbf
  {\bibinfo {volume} {87}},\ \bibinfo {pages} {307} (\bibinfo {year}
  {2015})}\BibitemShut {NoStop}%
\bibitem [{\citenamefont {Campbell}\ \emph {et~al.}(2017)\citenamefont
  {Campbell}, \citenamefont {Terhal},\ and\ \citenamefont
  {Vuillot}}]{Campbell2017}%
  \BibitemOpen
  \bibfield  {author} {\bibinfo {author} {\bibfnamefont {E.~T.}\ \bibnamefont
  {Campbell}}, \bibinfo {author} {\bibfnamefont {B.~M.}\ \bibnamefont
  {Terhal}},\ and\ \bibinfo {author} {\bibfnamefont {C.}~\bibnamefont
  {Vuillot}},\ }\bibfield  {title} {\bibinfo {title} {Roads towards
  fault-tolerant universal quantum computation},\ }\href
  {https://doi.org/10.1038/nature23460} {\bibfield  {journal} {\bibinfo
  {journal} {Nature}\ }\textbf {\bibinfo {volume} {549}},\ \bibinfo {pages}
  {172} (\bibinfo {year} {2017})}\BibitemShut {NoStop}%
\bibitem [{\citenamefont {Chuang}\ and\ \citenamefont
  {Nielsen}(1997)}]{Chuang1997tomo}%
  \BibitemOpen
  \bibfield  {author} {\bibinfo {author} {\bibfnamefont {I.~L.}\ \bibnamefont
  {Chuang}}\ and\ \bibinfo {author} {\bibfnamefont {M.~A.}\ \bibnamefont
  {Nielsen}},\ }\bibfield  {title} {\bibinfo {title} {Prescription for
  experimental determination of the dynamics of a quantum black box},\ }\href
  {https://doi.org/10.1080/09500349708231894} {\bibfield  {journal} {\bibinfo
  {journal} {Journal of Modern Optics}\ }\textbf {\bibinfo {volume} {44}},\
  \bibinfo {pages} {2455} (\bibinfo {year} {1997})},\ \Eprint
  {https://arxiv.org/abs/https://www.tandfonline.com/doi/pdf/10.1080/09500349708231894}
  {https://www.tandfonline.com/doi/pdf/10.1080/09500349708231894} \BibitemShut
  {NoStop}%
\bibitem [{\citenamefont {Gross}\ \emph {et~al.}(2010)\citenamefont {Gross},
  \citenamefont {Liu}, \citenamefont {Flammia}, \citenamefont {Becker},\ and\
  \citenamefont {Eisert}}]{Gross2010prlCompressedSensing}%
  \BibitemOpen
  \bibfield  {author} {\bibinfo {author} {\bibfnamefont {D.}~\bibnamefont
  {Gross}}, \bibinfo {author} {\bibfnamefont {Y.-K.}\ \bibnamefont {Liu}},
  \bibinfo {author} {\bibfnamefont {S.~T.}\ \bibnamefont {Flammia}}, \bibinfo
  {author} {\bibfnamefont {S.}~\bibnamefont {Becker}},\ and\ \bibinfo {author}
  {\bibfnamefont {J.}~\bibnamefont {Eisert}},\ }\bibfield  {title} {\bibinfo
  {title} {Quantum state tomography via compressed sensing},\ }\href
  {https://doi.org/10.1103/PhysRevLett.105.150401} {\bibfield  {journal}
  {\bibinfo  {journal} {Phys. Rev. Lett.}\ }\textbf {\bibinfo {volume} {105}},\
  \bibinfo {pages} {150401} (\bibinfo {year} {2010})}\BibitemShut {NoStop}%
\bibitem [{\citenamefont {Flammia}\ \emph {et~al.}(2012)\citenamefont
  {Flammia}, \citenamefont {Gross}, \citenamefont {Liu},\ and\ \citenamefont
  {Eisert}}]{Flammia2012CompressedSensing}%
  \BibitemOpen
  \bibfield  {author} {\bibinfo {author} {\bibfnamefont {S.~T.}\ \bibnamefont
  {Flammia}}, \bibinfo {author} {\bibfnamefont {D.}~\bibnamefont {Gross}},
  \bibinfo {author} {\bibfnamefont {Y.-K.}\ \bibnamefont {Liu}},\ and\ \bibinfo
  {author} {\bibfnamefont {J.}~\bibnamefont {Eisert}},\ }\bibfield  {title}
  {\bibinfo {title} {Quantum tomography via compressed sensing: error bounds,
  sample complexity and efficient estimators},\ }\href
  {https://doi.org/10.1088/1367-2630/14/9/095022} {\bibfield  {journal}
  {\bibinfo  {journal} {New Journal of Physics}\ }\textbf {\bibinfo {volume}
  {14}},\ \bibinfo {pages} {095022} (\bibinfo {year} {2012})}\BibitemShut
  {NoStop}%
\bibitem [{\citenamefont {Flammia}\ and\ \citenamefont
  {Liu}(2011)}]{Flammia2011prlDirectFidelity}%
  \BibitemOpen
  \bibfield  {author} {\bibinfo {author} {\bibfnamefont {S.~T.}\ \bibnamefont
  {Flammia}}\ and\ \bibinfo {author} {\bibfnamefont {Y.-K.}\ \bibnamefont
  {Liu}},\ }\bibfield  {title} {\bibinfo {title} {Direct fidelity estimation
  from few pauli measurements},\ }\href
  {https://doi.org/10.1103/PhysRevLett.106.230501} {\bibfield  {journal}
  {\bibinfo  {journal} {Phys. Rev. Lett.}\ }\textbf {\bibinfo {volume} {106}},\
  \bibinfo {pages} {230501} (\bibinfo {year} {2011})}\BibitemShut {NoStop}%
\bibitem [{\citenamefont {Emerson}\ \emph {et~al.}(2005)\citenamefont
  {Emerson}, \citenamefont {Alicki},\ and\ \citenamefont
  {Zyczkowski}}]{Emerson2005}%
  \BibitemOpen
  \bibfield  {author} {\bibinfo {author} {\bibfnamefont {J.}~\bibnamefont
  {Emerson}}, \bibinfo {author} {\bibfnamefont {R.}~\bibnamefont {Alicki}},\
  and\ \bibinfo {author} {\bibfnamefont {K.}~\bibnamefont {Zyczkowski}},\
  }\bibfield  {title} {\bibinfo {title} {Scalable noise estimation with random
  unitary operators},\ }\href {https://doi.org/10.1088/1464-4266/7/10/021}
  {\bibfield  {journal} {\bibinfo  {journal} {Journal of Optics B: Quantum and
  Semiclassical Optics}\ }\textbf {\bibinfo {volume} {7}},\ \bibinfo {pages}
  {S347} (\bibinfo {year} {2005})}\BibitemShut {NoStop}%
\bibitem [{\citenamefont {Emerson}\ \emph {et~al.}(2007)\citenamefont
  {Emerson}, \citenamefont {Silva}, \citenamefont {Moussa}, \citenamefont
  {Ryan}, \citenamefont {Laforest}, \citenamefont {Baugh}, \citenamefont
  {Cory},\ and\ \citenamefont {Laflamme}}]{Emerson2007science}%
  \BibitemOpen
  \bibfield  {author} {\bibinfo {author} {\bibfnamefont {J.}~\bibnamefont
  {Emerson}}, \bibinfo {author} {\bibfnamefont {M.}~\bibnamefont {Silva}},
  \bibinfo {author} {\bibfnamefont {O.}~\bibnamefont {Moussa}}, \bibinfo
  {author} {\bibfnamefont {C.}~\bibnamefont {Ryan}}, \bibinfo {author}
  {\bibfnamefont {M.}~\bibnamefont {Laforest}}, \bibinfo {author}
  {\bibfnamefont {J.}~\bibnamefont {Baugh}}, \bibinfo {author} {\bibfnamefont
  {D.~G.}\ \bibnamefont {Cory}},\ and\ \bibinfo {author} {\bibfnamefont
  {R.}~\bibnamefont {Laflamme}},\ }\bibfield  {title} {\bibinfo {title}
  {Symmetrized characterization of noisy quantum processes},\ }\href
  {https://doi.org/10.1126/science.1145699} {\bibfield  {journal} {\bibinfo
  {journal} {Science}\ }\textbf {\bibinfo {volume} {317}},\ \bibinfo {pages}
  {1893} (\bibinfo {year} {2007})},\ \Eprint
  {https://arxiv.org/abs/https://www.science.org/doi/pdf/10.1126/science.1145699}
  {https://www.science.org/doi/pdf/10.1126/science.1145699} \BibitemShut
  {NoStop}%
\bibitem [{\citenamefont {Knill}\ \emph {et~al.}(2008)\citenamefont {Knill},
  \citenamefont {Leibfried}, \citenamefont {Reichle}, \citenamefont {Britton},
  \citenamefont {Blakestad}, \citenamefont {Jost}, \citenamefont {Langer},
  \citenamefont {Ozeri}, \citenamefont {Seidelin},\ and\ \citenamefont
  {Wineland}}]{Knill2008pra}%
  \BibitemOpen
  \bibfield  {author} {\bibinfo {author} {\bibfnamefont {E.}~\bibnamefont
  {Knill}}, \bibinfo {author} {\bibfnamefont {D.}~\bibnamefont {Leibfried}},
  \bibinfo {author} {\bibfnamefont {R.}~\bibnamefont {Reichle}}, \bibinfo
  {author} {\bibfnamefont {J.}~\bibnamefont {Britton}}, \bibinfo {author}
  {\bibfnamefont {R.~B.}\ \bibnamefont {Blakestad}}, \bibinfo {author}
  {\bibfnamefont {J.~D.}\ \bibnamefont {Jost}}, \bibinfo {author}
  {\bibfnamefont {C.}~\bibnamefont {Langer}}, \bibinfo {author} {\bibfnamefont
  {R.}~\bibnamefont {Ozeri}}, \bibinfo {author} {\bibfnamefont
  {S.}~\bibnamefont {Seidelin}},\ and\ \bibinfo {author} {\bibfnamefont
  {D.~J.}\ \bibnamefont {Wineland}},\ }\bibfield  {title} {\bibinfo {title}
  {Randomized benchmarking of quantum gates},\ }\href
  {https://doi.org/10.1103/PhysRevA.77.012307} {\bibfield  {journal} {\bibinfo
  {journal} {Phys. Rev. A}\ }\textbf {\bibinfo {volume} {77}},\ \bibinfo
  {pages} {012307} (\bibinfo {year} {2008})}\BibitemShut {NoStop}%
\bibitem [{\citenamefont {Magesan}\ \emph {et~al.}(2011)\citenamefont
  {Magesan}, \citenamefont {Gambetta},\ and\ \citenamefont
  {Emerson}}]{Emerson2011prl}%
  \BibitemOpen
  \bibfield  {author} {\bibinfo {author} {\bibfnamefont {E.}~\bibnamefont
  {Magesan}}, \bibinfo {author} {\bibfnamefont {J.~M.}\ \bibnamefont
  {Gambetta}},\ and\ \bibinfo {author} {\bibfnamefont {J.}~\bibnamefont
  {Emerson}},\ }\bibfield  {title} {\bibinfo {title} {Scalable and robust
  randomized benchmarking of quantum processes},\ }\href
  {https://doi.org/10.1103/PhysRevLett.106.180504} {\bibfield  {journal}
  {\bibinfo  {journal} {Phys. Rev. Lett.}\ }\textbf {\bibinfo {volume} {106}},\
  \bibinfo {pages} {180504} (\bibinfo {year} {2011})}\BibitemShut {NoStop}%
\bibitem [{\citenamefont {Magesan}\ \emph
  {et~al.}(2012{\natexlab{a}})\citenamefont {Magesan}, \citenamefont
  {Gambetta},\ and\ \citenamefont {Emerson}}]{Emerson2012pra}%
  \BibitemOpen
  \bibfield  {author} {\bibinfo {author} {\bibfnamefont {E.}~\bibnamefont
  {Magesan}}, \bibinfo {author} {\bibfnamefont {J.~M.}\ \bibnamefont
  {Gambetta}},\ and\ \bibinfo {author} {\bibfnamefont {J.}~\bibnamefont
  {Emerson}},\ }\bibfield  {title} {\bibinfo {title} {Characterizing quantum
  gates via randomized benchmarking},\ }\href
  {https://doi.org/10.1103/PhysRevA.85.042311} {\bibfield  {journal} {\bibinfo
  {journal} {Phys. Rev. A}\ }\textbf {\bibinfo {volume} {85}},\ \bibinfo
  {pages} {042311} (\bibinfo {year} {2012}{\natexlab{a}})}\BibitemShut
  {NoStop}%
\bibitem [{\citenamefont {Flammia}\ and\ \citenamefont
  {Wallman}(2020)}]{10.1145/3408039}%
  \BibitemOpen
  \bibfield  {author} {\bibinfo {author} {\bibfnamefont {S.~T.}\ \bibnamefont
  {Flammia}}\ and\ \bibinfo {author} {\bibfnamefont {J.~J.}\ \bibnamefont
  {Wallman}},\ }\bibfield  {title} {\bibinfo {title} {Efficient estimation of
  pauli channels},\ }\bibfield  {journal} {\bibinfo  {journal} {ACM
  Transactions on Quantum Computing}\ }\textbf {\bibinfo {volume} {1}},\ \href
  {https://doi.org/10.1145/3408039} {10.1145/3408039} (\bibinfo {year}
  {2020})\BibitemShut {NoStop}%
\bibitem [{\citenamefont {Harper}\ \emph {et~al.}(2021)\citenamefont {Harper},
  \citenamefont {Yu},\ and\ \citenamefont {Flammia}}]{PRXQuantum.2.010322}%
  \BibitemOpen
  \bibfield  {author} {\bibinfo {author} {\bibfnamefont {R.}~\bibnamefont
  {Harper}}, \bibinfo {author} {\bibfnamefont {W.}~\bibnamefont {Yu}},\ and\
  \bibinfo {author} {\bibfnamefont {S.~T.}\ \bibnamefont {Flammia}},\
  }\bibfield  {title} {\bibinfo {title} {Fast estimation of sparse quantum
  noise},\ }\href {https://doi.org/10.1103/PRXQuantum.2.010322} {\bibfield
  {journal} {\bibinfo  {journal} {PRX Quantum}\ }\textbf {\bibinfo {volume}
  {2}},\ \bibinfo {pages} {010322} (\bibinfo {year} {2021})}\BibitemShut
  {NoStop}%
\bibitem [{\citenamefont {Chow}\ \emph {et~al.}(2009)\citenamefont {Chow},
  \citenamefont {Gambetta}, \citenamefont {Tornberg}, \citenamefont {Koch},
  \citenamefont {Bishop}, \citenamefont {Houck}, \citenamefont {Johnson},
  \citenamefont {Frunzio}, \citenamefont {Girvin},\ and\ \citenamefont
  {Schoelkopf}}]{Chow2009prlRB}%
  \BibitemOpen
  \bibfield  {author} {\bibinfo {author} {\bibfnamefont {J.~M.}\ \bibnamefont
  {Chow}}, \bibinfo {author} {\bibfnamefont {J.~M.}\ \bibnamefont {Gambetta}},
  \bibinfo {author} {\bibfnamefont {L.}~\bibnamefont {Tornberg}}, \bibinfo
  {author} {\bibfnamefont {J.}~\bibnamefont {Koch}}, \bibinfo {author}
  {\bibfnamefont {L.~S.}\ \bibnamefont {Bishop}}, \bibinfo {author}
  {\bibfnamefont {A.~A.}\ \bibnamefont {Houck}}, \bibinfo {author}
  {\bibfnamefont {B.~R.}\ \bibnamefont {Johnson}}, \bibinfo {author}
  {\bibfnamefont {L.}~\bibnamefont {Frunzio}}, \bibinfo {author} {\bibfnamefont
  {S.~M.}\ \bibnamefont {Girvin}},\ and\ \bibinfo {author} {\bibfnamefont
  {R.~J.}\ \bibnamefont {Schoelkopf}},\ }\bibfield  {title} {\bibinfo {title}
  {Randomized benchmarking and process tomography for gate errors in a
  solid-state qubit},\ }\href {https://doi.org/10.1103/PhysRevLett.102.090502}
  {\bibfield  {journal} {\bibinfo  {journal} {Phys. Rev. Lett.}\ }\textbf
  {\bibinfo {volume} {102}},\ \bibinfo {pages} {090502} (\bibinfo {year}
  {2009})}\BibitemShut {NoStop}%
\bibitem [{\citenamefont {Gaebler}\ \emph {et~al.}(2012)\citenamefont
  {Gaebler}, \citenamefont {Meier}, \citenamefont {Tan}, \citenamefont
  {Bowler}, \citenamefont {Lin}, \citenamefont {Hanneke}, \citenamefont {Jost},
  \citenamefont {Home}, \citenamefont {Knill}, \citenamefont {Leibfried},\ and\
  \citenamefont {Wineland}}]{Gaebler2012prlRB}%
  \BibitemOpen
  \bibfield  {author} {\bibinfo {author} {\bibfnamefont {J.~P.}\ \bibnamefont
  {Gaebler}}, \bibinfo {author} {\bibfnamefont {A.~M.}\ \bibnamefont {Meier}},
  \bibinfo {author} {\bibfnamefont {T.~R.}\ \bibnamefont {Tan}}, \bibinfo
  {author} {\bibfnamefont {R.}~\bibnamefont {Bowler}}, \bibinfo {author}
  {\bibfnamefont {Y.}~\bibnamefont {Lin}}, \bibinfo {author} {\bibfnamefont
  {D.}~\bibnamefont {Hanneke}}, \bibinfo {author} {\bibfnamefont {J.~D.}\
  \bibnamefont {Jost}}, \bibinfo {author} {\bibfnamefont {J.~P.}\ \bibnamefont
  {Home}}, \bibinfo {author} {\bibfnamefont {E.}~\bibnamefont {Knill}},
  \bibinfo {author} {\bibfnamefont {D.}~\bibnamefont {Leibfried}},\ and\
  \bibinfo {author} {\bibfnamefont {D.~J.}\ \bibnamefont {Wineland}},\
  }\bibfield  {title} {\bibinfo {title} {Randomized benchmarking of multiqubit
  gates},\ }\href {https://doi.org/10.1103/PhysRevLett.108.260503} {\bibfield
  {journal} {\bibinfo  {journal} {Phys. Rev. Lett.}\ }\textbf {\bibinfo
  {volume} {108}},\ \bibinfo {pages} {260503} (\bibinfo {year}
  {2012})}\BibitemShut {NoStop}%
\bibitem [{\citenamefont {Moussa}\ \emph {et~al.}(2012)\citenamefont {Moussa},
  \citenamefont {da~Silva}, \citenamefont {Ryan},\ and\ \citenamefont
  {Laflamme}}]{Laflamme2012prl}%
  \BibitemOpen
  \bibfield  {author} {\bibinfo {author} {\bibfnamefont {O.}~\bibnamefont
  {Moussa}}, \bibinfo {author} {\bibfnamefont {M.~P.}\ \bibnamefont
  {da~Silva}}, \bibinfo {author} {\bibfnamefont {C.~A.}\ \bibnamefont {Ryan}},\
  and\ \bibinfo {author} {\bibfnamefont {R.}~\bibnamefont {Laflamme}},\
  }\bibfield  {title} {\bibinfo {title} {Practical experimental certification
  of computational quantum gates using a twirling procedure},\ }\href
  {https://doi.org/10.1103/PhysRevLett.109.070504} {\bibfield  {journal}
  {\bibinfo  {journal} {Phys. Rev. Lett.}\ }\textbf {\bibinfo {volume} {109}},\
  \bibinfo {pages} {070504} (\bibinfo {year} {2012})}\BibitemShut {NoStop}%
\bibitem [{\citenamefont {Barends}\ \emph {et~al.}(2014)\citenamefont
  {Barends}, \citenamefont {Kelly}, \citenamefont {Megrant}, \citenamefont
  {Veitia}, \citenamefont {Sank}, \citenamefont {Jeffrey}, \citenamefont
  {White}, \citenamefont {Mutus}, \citenamefont {Fowler}, \citenamefont
  {Campbell}, \citenamefont {Chen}, \citenamefont {Chen}, \citenamefont
  {Chiaro}, \citenamefont {Dunsworth}, \citenamefont {Neill}, \citenamefont
  {O'Malley}, \citenamefont {Roushan}, \citenamefont {Vainsencher},
  \citenamefont {Wenner}, \citenamefont {Korotkov}, \citenamefont {Cleland},\
  and\ \citenamefont {Martinis}}]{Barends2014surface}%
  \BibitemOpen
  \bibfield  {author} {\bibinfo {author} {\bibfnamefont {R.}~\bibnamefont
  {Barends}}, \bibinfo {author} {\bibfnamefont {J.}~\bibnamefont {Kelly}},
  \bibinfo {author} {\bibfnamefont {A.}~\bibnamefont {Megrant}}, \bibinfo
  {author} {\bibfnamefont {A.}~\bibnamefont {Veitia}}, \bibinfo {author}
  {\bibfnamefont {D.}~\bibnamefont {Sank}}, \bibinfo {author} {\bibfnamefont
  {E.}~\bibnamefont {Jeffrey}}, \bibinfo {author} {\bibfnamefont {T.~C.}\
  \bibnamefont {White}}, \bibinfo {author} {\bibfnamefont {J.}~\bibnamefont
  {Mutus}}, \bibinfo {author} {\bibfnamefont {A.~G.}\ \bibnamefont {Fowler}},
  \bibinfo {author} {\bibfnamefont {B.}~\bibnamefont {Campbell}}, \bibinfo
  {author} {\bibfnamefont {Y.}~\bibnamefont {Chen}}, \bibinfo {author}
  {\bibfnamefont {Z.}~\bibnamefont {Chen}}, \bibinfo {author} {\bibfnamefont
  {B.}~\bibnamefont {Chiaro}}, \bibinfo {author} {\bibfnamefont
  {A.}~\bibnamefont {Dunsworth}}, \bibinfo {author} {\bibfnamefont
  {C.}~\bibnamefont {Neill}}, \bibinfo {author} {\bibfnamefont
  {P.}~\bibnamefont {O'Malley}}, \bibinfo {author} {\bibfnamefont
  {P.}~\bibnamefont {Roushan}}, \bibinfo {author} {\bibfnamefont
  {A.}~\bibnamefont {Vainsencher}}, \bibinfo {author} {\bibfnamefont
  {J.}~\bibnamefont {Wenner}}, \bibinfo {author} {\bibfnamefont {A.~N.}\
  \bibnamefont {Korotkov}}, \bibinfo {author} {\bibfnamefont {A.~N.}\
  \bibnamefont {Cleland}},\ and\ \bibinfo {author} {\bibfnamefont {J.~M.}\
  \bibnamefont {Martinis}},\ }\bibfield  {title} {\bibinfo {title}
  {Superconducting quantum circuits at the surface code threshold for fault
  tolerance},\ }\href {https://doi.org/10.1038/nature13171} {\bibfield
  {journal} {\bibinfo  {journal} {Nature}\ }\textbf {\bibinfo {volume} {508}},\
  \bibinfo {pages} {500} (\bibinfo {year} {2014})}\BibitemShut {NoStop}%
\bibitem [{\citenamefont {Lu}\ \emph {et~al.}(2015)\citenamefont {Lu},
  \citenamefont {Li}, \citenamefont {Trottier}, \citenamefont {Li},
  \citenamefont {Brodutch}, \citenamefont {Krismanich}, \citenamefont
  {Ghavami}, \citenamefont {Dmitrienko}, \citenamefont {Long}, \citenamefont
  {Baugh},\ and\ \citenamefont {Laflamme}}]{Lu2015prl}%
  \BibitemOpen
  \bibfield  {author} {\bibinfo {author} {\bibfnamefont {D.}~\bibnamefont
  {Lu}}, \bibinfo {author} {\bibfnamefont {H.}~\bibnamefont {Li}}, \bibinfo
  {author} {\bibfnamefont {D.-A.}\ \bibnamefont {Trottier}}, \bibinfo {author}
  {\bibfnamefont {J.}~\bibnamefont {Li}}, \bibinfo {author} {\bibfnamefont
  {A.}~\bibnamefont {Brodutch}}, \bibinfo {author} {\bibfnamefont {A.~P.}\
  \bibnamefont {Krismanich}}, \bibinfo {author} {\bibfnamefont
  {A.}~\bibnamefont {Ghavami}}, \bibinfo {author} {\bibfnamefont {G.~I.}\
  \bibnamefont {Dmitrienko}}, \bibinfo {author} {\bibfnamefont
  {G.}~\bibnamefont {Long}}, \bibinfo {author} {\bibfnamefont {J.}~\bibnamefont
  {Baugh}},\ and\ \bibinfo {author} {\bibfnamefont {R.}~\bibnamefont
  {Laflamme}},\ }\bibfield  {title} {\bibinfo {title} {Experimental estimation
  of average fidelity of a clifford gate on a 7-qubit quantum processor},\
  }\href {https://doi.org/10.1103/PhysRevLett.114.140505} {\bibfield  {journal}
  {\bibinfo  {journal} {Phys. Rev. Lett.}\ }\textbf {\bibinfo {volume} {114}},\
  \bibinfo {pages} {140505} (\bibinfo {year} {2015})}\BibitemShut {NoStop}%
\bibitem [{\citenamefont {Ballance}\ \emph {et~al.}(2016)\citenamefont
  {Ballance}, \citenamefont {Harty}, \citenamefont {Linke}, \citenamefont
  {Sepiol},\ and\ \citenamefont {Lucas}}]{Ballance2016prlRB}%
  \BibitemOpen
  \bibfield  {author} {\bibinfo {author} {\bibfnamefont {C.~J.}\ \bibnamefont
  {Ballance}}, \bibinfo {author} {\bibfnamefont {T.~P.}\ \bibnamefont {Harty}},
  \bibinfo {author} {\bibfnamefont {N.~M.}\ \bibnamefont {Linke}}, \bibinfo
  {author} {\bibfnamefont {M.~A.}\ \bibnamefont {Sepiol}},\ and\ \bibinfo
  {author} {\bibfnamefont {D.~M.}\ \bibnamefont {Lucas}},\ }\bibfield  {title}
  {\bibinfo {title} {High-fidelity quantum logic gates using trapped-ion
  hyperfine qubits},\ }\href {https://doi.org/10.1103/PhysRevLett.117.060504}
  {\bibfield  {journal} {\bibinfo  {journal} {Phys. Rev. Lett.}\ }\textbf
  {\bibinfo {volume} {117}},\ \bibinfo {pages} {060504} (\bibinfo {year}
  {2016})}\BibitemShut {NoStop}%
\bibitem [{\citenamefont {Gaebler}\ \emph {et~al.}(2016)\citenamefont
  {Gaebler}, \citenamefont {Tan}, \citenamefont {Lin}, \citenamefont {Wan},
  \citenamefont {Bowler}, \citenamefont {Keith}, \citenamefont {Glancy},
  \citenamefont {Coakley}, \citenamefont {Knill}, \citenamefont {Leibfried},\
  and\ \citenamefont {Wineland}}]{Gaebler2016RBion}%
  \BibitemOpen
  \bibfield  {author} {\bibinfo {author} {\bibfnamefont {J.~P.}\ \bibnamefont
  {Gaebler}}, \bibinfo {author} {\bibfnamefont {T.~R.}\ \bibnamefont {Tan}},
  \bibinfo {author} {\bibfnamefont {Y.}~\bibnamefont {Lin}}, \bibinfo {author}
  {\bibfnamefont {Y.}~\bibnamefont {Wan}}, \bibinfo {author} {\bibfnamefont
  {R.}~\bibnamefont {Bowler}}, \bibinfo {author} {\bibfnamefont {A.~C.}\
  \bibnamefont {Keith}}, \bibinfo {author} {\bibfnamefont {S.}~\bibnamefont
  {Glancy}}, \bibinfo {author} {\bibfnamefont {K.}~\bibnamefont {Coakley}},
  \bibinfo {author} {\bibfnamefont {E.}~\bibnamefont {Knill}}, \bibinfo
  {author} {\bibfnamefont {D.}~\bibnamefont {Leibfried}},\ and\ \bibinfo
  {author} {\bibfnamefont {D.~J.}\ \bibnamefont {Wineland}},\ }\bibfield
  {title} {\bibinfo {title} {High-fidelity universal gate set for
  ${^{9}\mathrm{Be}}^{+}$ ion qubits},\ }\href
  {https://doi.org/10.1103/PhysRevLett.117.060505} {\bibfield  {journal}
  {\bibinfo  {journal} {Phys. Rev. Lett.}\ }\textbf {\bibinfo {volume} {117}},\
  \bibinfo {pages} {060505} (\bibinfo {year} {2016})}\BibitemShut {NoStop}%
\bibitem [{\citenamefont {Proctor}\ \emph {et~al.}(2021)\citenamefont
  {Proctor}, \citenamefont {Seritan}, \citenamefont {Rudinger}, \citenamefont
  {Nielsen}, \citenamefont {Blume-Kohout},\ and\ \citenamefont
  {Young}}]{proctor2021scalable}%
  \BibitemOpen
  \bibfield  {author} {\bibinfo {author} {\bibfnamefont {T.}~\bibnamefont
  {Proctor}}, \bibinfo {author} {\bibfnamefont {S.}~\bibnamefont {Seritan}},
  \bibinfo {author} {\bibfnamefont {K.}~\bibnamefont {Rudinger}}, \bibinfo
  {author} {\bibfnamefont {E.}~\bibnamefont {Nielsen}}, \bibinfo {author}
  {\bibfnamefont {R.}~\bibnamefont {Blume-Kohout}},\ and\ \bibinfo {author}
  {\bibfnamefont {K.}~\bibnamefont {Young}},\ }\bibfield  {title} {\bibinfo
  {title} {Scalable randomized benchmarking of quantum computers using mirror
  circuits},\ }\href@noop {} {\bibfield  {journal} {\bibinfo  {journal} {arXiv
  preprint arXiv:2112.09853}\ } (\bibinfo {year} {2021})}\BibitemShut {NoStop}%
\bibitem [{\citenamefont {Magesan}\ \emph
  {et~al.}(2012{\natexlab{b}})\citenamefont {Magesan}, \citenamefont
  {Gambetta}, \citenamefont {Johnson}, \citenamefont {Ryan}, \citenamefont
  {Chow}, \citenamefont {Merkel}, \citenamefont {da~Silva}, \citenamefont
  {Keefe}, \citenamefont {Rothwell}, \citenamefont {Ohki}, \citenamefont
  {Ketchen},\ and\ \citenamefont {Steffen}}]{Magesan2012interleavedRB}%
  \BibitemOpen
  \bibfield  {author} {\bibinfo {author} {\bibfnamefont {E.}~\bibnamefont
  {Magesan}}, \bibinfo {author} {\bibfnamefont {J.~M.}\ \bibnamefont
  {Gambetta}}, \bibinfo {author} {\bibfnamefont {B.~R.}\ \bibnamefont
  {Johnson}}, \bibinfo {author} {\bibfnamefont {C.~A.}\ \bibnamefont {Ryan}},
  \bibinfo {author} {\bibfnamefont {J.~M.}\ \bibnamefont {Chow}}, \bibinfo
  {author} {\bibfnamefont {S.~T.}\ \bibnamefont {Merkel}}, \bibinfo {author}
  {\bibfnamefont {M.~P.}\ \bibnamefont {da~Silva}}, \bibinfo {author}
  {\bibfnamefont {G.~A.}\ \bibnamefont {Keefe}}, \bibinfo {author}
  {\bibfnamefont {M.~B.}\ \bibnamefont {Rothwell}}, \bibinfo {author}
  {\bibfnamefont {T.~A.}\ \bibnamefont {Ohki}}, \bibinfo {author}
  {\bibfnamefont {M.~B.}\ \bibnamefont {Ketchen}},\ and\ \bibinfo {author}
  {\bibfnamefont {M.}~\bibnamefont {Steffen}},\ }\bibfield  {title} {\bibinfo
  {title} {Efficient measurement of quantum gate error by interleaved
  randomized benchmarking},\ }\href
  {https://doi.org/10.1103/PhysRevLett.109.080505} {\bibfield  {journal}
  {\bibinfo  {journal} {Phys. Rev. Lett.}\ }\textbf {\bibinfo {volume} {109}},\
  \bibinfo {pages} {080505} (\bibinfo {year} {2012}{\natexlab{b}})}\BibitemShut
  {NoStop}%
\bibitem [{\citenamefont {Wallman}(2018)}]{Wallman2018quantum}%
  \BibitemOpen
  \bibfield  {author} {\bibinfo {author} {\bibfnamefont {J.~J.}\ \bibnamefont
  {Wallman}},\ }\bibfield  {title} {\bibinfo {title} {Randomized benchmarking
  with gate-dependent noise},\ }\href
  {https://doi.org/10.22331/q-2018-01-29-47} {\bibfield  {journal} {\bibinfo
  {journal} {{Quantum}}\ }\textbf {\bibinfo {volume} {2}},\ \bibinfo {pages}
  {47} (\bibinfo {year} {2018})}\BibitemShut {NoStop}%
\bibitem [{\citenamefont {Merkel}\ \emph {et~al.}(2021)\citenamefont {Merkel},
  \citenamefont {Pritchett},\ and\ \citenamefont {Fong}}]{merkel2021RB}%
  \BibitemOpen
  \bibfield  {author} {\bibinfo {author} {\bibfnamefont {S.~T.}\ \bibnamefont
  {Merkel}}, \bibinfo {author} {\bibfnamefont {E.~J.}\ \bibnamefont
  {Pritchett}},\ and\ \bibinfo {author} {\bibfnamefont {B.~H.}\ \bibnamefont
  {Fong}},\ }\bibfield  {title} {\bibinfo {title} {Randomized {B}enchmarking as
  {C}onvolution: {F}ourier {A}nalysis of {G}ate {D}ependent {E}rrors},\ }\href
  {https://doi.org/10.22331/q-2021-11-16-581} {\bibfield  {journal} {\bibinfo
  {journal} {{Quantum}}\ }\textbf {\bibinfo {volume} {5}},\ \bibinfo {pages}
  {581} (\bibinfo {year} {2021})}\BibitemShut {NoStop}%
\bibitem [{\citenamefont {Helsen}\ \emph {et~al.}(2019)\citenamefont {Helsen},
  \citenamefont {Xue}, \citenamefont {Vandersypen},\ and\ \citenamefont
  {Wehner}}]{Helsen2019characterRB}%
  \BibitemOpen
  \bibfield  {author} {\bibinfo {author} {\bibfnamefont {J.}~\bibnamefont
  {Helsen}}, \bibinfo {author} {\bibfnamefont {X.}~\bibnamefont {Xue}},
  \bibinfo {author} {\bibfnamefont {L.~M.~K.}\ \bibnamefont {Vandersypen}},\
  and\ \bibinfo {author} {\bibfnamefont {S.}~\bibnamefont {Wehner}},\
  }\bibfield  {title} {\bibinfo {title} {A new class of efficient randomized
  benchmarking protocols},\ }\href {https://doi.org/10.1038/s41534-019-0182-7}
  {\bibfield  {journal} {\bibinfo  {journal} {npj Quantum Information}\
  }\textbf {\bibinfo {volume} {5}},\ \bibinfo {pages} {71} (\bibinfo {year}
  {2019})}\BibitemShut {NoStop}%
\bibitem [{\citenamefont {Erhard}\ \emph {et~al.}(2019)\citenamefont {Erhard},
  \citenamefont {Wallman}, \citenamefont {Postler}, \citenamefont {Meth},
  \citenamefont {Stricker}, \citenamefont {Martinez}, \citenamefont
  {Schindler}, \citenamefont {Monz}, \citenamefont {Emerson},\ and\
  \citenamefont {Blatt}}]{Erhard2019cycleRB}%
  \BibitemOpen
  \bibfield  {author} {\bibinfo {author} {\bibfnamefont {A.}~\bibnamefont
  {Erhard}}, \bibinfo {author} {\bibfnamefont {J.~J.}\ \bibnamefont {Wallman}},
  \bibinfo {author} {\bibfnamefont {L.}~\bibnamefont {Postler}}, \bibinfo
  {author} {\bibfnamefont {M.}~\bibnamefont {Meth}}, \bibinfo {author}
  {\bibfnamefont {R.}~\bibnamefont {Stricker}}, \bibinfo {author}
  {\bibfnamefont {E.~A.}\ \bibnamefont {Martinez}}, \bibinfo {author}
  {\bibfnamefont {P.}~\bibnamefont {Schindler}}, \bibinfo {author}
  {\bibfnamefont {T.}~\bibnamefont {Monz}}, \bibinfo {author} {\bibfnamefont
  {J.}~\bibnamefont {Emerson}},\ and\ \bibinfo {author} {\bibfnamefont
  {R.}~\bibnamefont {Blatt}},\ }\bibfield  {title} {\bibinfo {title}
  {Characterizing large-scale quantum computers via cycle benchmarking},\
  }\href {https://doi.org/10.1038/s41467-019-13068-7} {\bibfield  {journal}
  {\bibinfo  {journal} {Nature Communications}\ }\textbf {\bibinfo {volume}
  {10}},\ \bibinfo {pages} {5347} (\bibinfo {year} {2019})}\BibitemShut
  {NoStop}%
\bibitem [{\citenamefont {Arute}\ \emph {et~al.}(2019)\citenamefont {Arute},
  \citenamefont {Arya}, \citenamefont {Babbush}, \citenamefont {Bacon},
  \citenamefont {Bardin}, \citenamefont {Barends}, \citenamefont {Biswas},
  \citenamefont {Boixo}, \citenamefont {Brandao}, \citenamefont {Buell},
  \citenamefont {Burkett}, \citenamefont {Chen}, \citenamefont {Chen},
  \citenamefont {Chiaro}, \citenamefont {Collins}, \citenamefont {Courtney},
  \citenamefont {Dunsworth}, \citenamefont {Farhi}, \citenamefont {Foxen},
  \citenamefont {Fowler}, \citenamefont {Gidney}, \citenamefont {Giustina},
  \citenamefont {Graff}, \citenamefont {Guerin}, \citenamefont {Habegger},
  \citenamefont {Harrigan}, \citenamefont {Hartmann}, \citenamefont {Ho},
  \citenamefont {Hoffmann}, \citenamefont {Huang}, \citenamefont {Humble},
  \citenamefont {Isakov}, \citenamefont {Jeffrey}, \citenamefont {Jiang},
  \citenamefont {Kafri}, \citenamefont {Kechedzhi}, \citenamefont {Kelly},
  \citenamefont {Klimov}, \citenamefont {Knysh}, \citenamefont {Korotkov},
  \citenamefont {Kostritsa}, \citenamefont {Landhuis}, \citenamefont
  {Lindmark}, \citenamefont {Lucero}, \citenamefont {Lyakh}, \citenamefont
  {Mandr{\`a}}, \citenamefont {McClean}, \citenamefont {McEwen}, \citenamefont
  {Megrant}, \citenamefont {Mi}, \citenamefont {Michielsen}, \citenamefont
  {Mohseni}, \citenamefont {Mutus}, \citenamefont {Naaman}, \citenamefont
  {Neeley}, \citenamefont {Neill}, \citenamefont {Niu}, \citenamefont {Ostby},
  \citenamefont {Petukhov}, \citenamefont {Platt}, \citenamefont {Quintana},
  \citenamefont {Rieffel}, \citenamefont {Roushan}, \citenamefont {Rubin},
  \citenamefont {Sank}, \citenamefont {Satzinger}, \citenamefont {Smelyanskiy},
  \citenamefont {Sung}, \citenamefont {Trevithick}, \citenamefont
  {Vainsencher}, \citenamefont {Villalonga}, \citenamefont {White},
  \citenamefont {Yao}, \citenamefont {Yeh}, \citenamefont {Zalcman},
  \citenamefont {Neven},\ and\ \citenamefont {Martinis}}]{XEB2019google}%
  \BibitemOpen
  \bibfield  {author} {\bibinfo {author} {\bibfnamefont {F.}~\bibnamefont
  {Arute}}, \bibinfo {author} {\bibfnamefont {K.}~\bibnamefont {Arya}},
  \bibinfo {author} {\bibfnamefont {R.}~\bibnamefont {Babbush}}, \bibinfo
  {author} {\bibfnamefont {D.}~\bibnamefont {Bacon}}, \bibinfo {author}
  {\bibfnamefont {J.~C.}\ \bibnamefont {Bardin}}, \bibinfo {author}
  {\bibfnamefont {R.}~\bibnamefont {Barends}}, \bibinfo {author} {\bibfnamefont
  {R.}~\bibnamefont {Biswas}}, \bibinfo {author} {\bibfnamefont
  {S.}~\bibnamefont {Boixo}}, \bibinfo {author} {\bibfnamefont {F.~G. S.~L.}\
  \bibnamefont {Brandao}}, \bibinfo {author} {\bibfnamefont {D.~A.}\
  \bibnamefont {Buell}}, \bibinfo {author} {\bibfnamefont {B.}~\bibnamefont
  {Burkett}}, \bibinfo {author} {\bibfnamefont {Y.}~\bibnamefont {Chen}},
  \bibinfo {author} {\bibfnamefont {Z.}~\bibnamefont {Chen}}, \bibinfo {author}
  {\bibfnamefont {B.}~\bibnamefont {Chiaro}}, \bibinfo {author} {\bibfnamefont
  {R.}~\bibnamefont {Collins}}, \bibinfo {author} {\bibfnamefont
  {W.}~\bibnamefont {Courtney}}, \bibinfo {author} {\bibfnamefont
  {A.}~\bibnamefont {Dunsworth}}, \bibinfo {author} {\bibfnamefont
  {E.}~\bibnamefont {Farhi}}, \bibinfo {author} {\bibfnamefont
  {B.}~\bibnamefont {Foxen}}, \bibinfo {author} {\bibfnamefont
  {A.}~\bibnamefont {Fowler}}, \bibinfo {author} {\bibfnamefont
  {C.}~\bibnamefont {Gidney}}, \bibinfo {author} {\bibfnamefont
  {M.}~\bibnamefont {Giustina}}, \bibinfo {author} {\bibfnamefont
  {R.}~\bibnamefont {Graff}}, \bibinfo {author} {\bibfnamefont
  {K.}~\bibnamefont {Guerin}}, \bibinfo {author} {\bibfnamefont
  {S.}~\bibnamefont {Habegger}}, \bibinfo {author} {\bibfnamefont {M.~P.}\
  \bibnamefont {Harrigan}}, \bibinfo {author} {\bibfnamefont {M.~J.}\
  \bibnamefont {Hartmann}}, \bibinfo {author} {\bibfnamefont {A.}~\bibnamefont
  {Ho}}, \bibinfo {author} {\bibfnamefont {M.}~\bibnamefont {Hoffmann}},
  \bibinfo {author} {\bibfnamefont {T.}~\bibnamefont {Huang}}, \bibinfo
  {author} {\bibfnamefont {T.~S.}\ \bibnamefont {Humble}}, \bibinfo {author}
  {\bibfnamefont {S.~V.}\ \bibnamefont {Isakov}}, \bibinfo {author}
  {\bibfnamefont {E.}~\bibnamefont {Jeffrey}}, \bibinfo {author} {\bibfnamefont
  {Z.}~\bibnamefont {Jiang}}, \bibinfo {author} {\bibfnamefont
  {D.}~\bibnamefont {Kafri}}, \bibinfo {author} {\bibfnamefont
  {K.}~\bibnamefont {Kechedzhi}}, \bibinfo {author} {\bibfnamefont
  {J.}~\bibnamefont {Kelly}}, \bibinfo {author} {\bibfnamefont {P.~V.}\
  \bibnamefont {Klimov}}, \bibinfo {author} {\bibfnamefont {S.}~\bibnamefont
  {Knysh}}, \bibinfo {author} {\bibfnamefont {A.}~\bibnamefont {Korotkov}},
  \bibinfo {author} {\bibfnamefont {F.}~\bibnamefont {Kostritsa}}, \bibinfo
  {author} {\bibfnamefont {D.}~\bibnamefont {Landhuis}}, \bibinfo {author}
  {\bibfnamefont {M.}~\bibnamefont {Lindmark}}, \bibinfo {author}
  {\bibfnamefont {E.}~\bibnamefont {Lucero}}, \bibinfo {author} {\bibfnamefont
  {D.}~\bibnamefont {Lyakh}}, \bibinfo {author} {\bibfnamefont
  {S.}~\bibnamefont {Mandr{\`a}}}, \bibinfo {author} {\bibfnamefont {J.~R.}\
  \bibnamefont {McClean}}, \bibinfo {author} {\bibfnamefont {M.}~\bibnamefont
  {McEwen}}, \bibinfo {author} {\bibfnamefont {A.}~\bibnamefont {Megrant}},
  \bibinfo {author} {\bibfnamefont {X.}~\bibnamefont {Mi}}, \bibinfo {author}
  {\bibfnamefont {K.}~\bibnamefont {Michielsen}}, \bibinfo {author}
  {\bibfnamefont {M.}~\bibnamefont {Mohseni}}, \bibinfo {author} {\bibfnamefont
  {J.}~\bibnamefont {Mutus}}, \bibinfo {author} {\bibfnamefont
  {O.}~\bibnamefont {Naaman}}, \bibinfo {author} {\bibfnamefont
  {M.}~\bibnamefont {Neeley}}, \bibinfo {author} {\bibfnamefont
  {C.}~\bibnamefont {Neill}}, \bibinfo {author} {\bibfnamefont {M.~Y.}\
  \bibnamefont {Niu}}, \bibinfo {author} {\bibfnamefont {E.}~\bibnamefont
  {Ostby}}, \bibinfo {author} {\bibfnamefont {A.}~\bibnamefont {Petukhov}},
  \bibinfo {author} {\bibfnamefont {J.~C.}\ \bibnamefont {Platt}}, \bibinfo
  {author} {\bibfnamefont {C.}~\bibnamefont {Quintana}}, \bibinfo {author}
  {\bibfnamefont {E.~G.}\ \bibnamefont {Rieffel}}, \bibinfo {author}
  {\bibfnamefont {P.}~\bibnamefont {Roushan}}, \bibinfo {author} {\bibfnamefont
  {N.~C.}\ \bibnamefont {Rubin}}, \bibinfo {author} {\bibfnamefont
  {D.}~\bibnamefont {Sank}}, \bibinfo {author} {\bibfnamefont {K.~J.}\
  \bibnamefont {Satzinger}}, \bibinfo {author} {\bibfnamefont {V.}~\bibnamefont
  {Smelyanskiy}}, \bibinfo {author} {\bibfnamefont {K.~J.}\ \bibnamefont
  {Sung}}, \bibinfo {author} {\bibfnamefont {M.~D.}\ \bibnamefont
  {Trevithick}}, \bibinfo {author} {\bibfnamefont {A.}~\bibnamefont
  {Vainsencher}}, \bibinfo {author} {\bibfnamefont {B.}~\bibnamefont
  {Villalonga}}, \bibinfo {author} {\bibfnamefont {T.}~\bibnamefont {White}},
  \bibinfo {author} {\bibfnamefont {Z.~J.}\ \bibnamefont {Yao}}, \bibinfo
  {author} {\bibfnamefont {P.}~\bibnamefont {Yeh}}, \bibinfo {author}
  {\bibfnamefont {A.}~\bibnamefont {Zalcman}}, \bibinfo {author} {\bibfnamefont
  {H.}~\bibnamefont {Neven}},\ and\ \bibinfo {author} {\bibfnamefont {J.~M.}\
  \bibnamefont {Martinis}},\ }\bibfield  {title} {\bibinfo {title} {Quantum
  supremacy using a programmable superconducting processor},\ }\href
  {https://doi.org/10.1038/s41586-019-1666-5} {\bibfield  {journal} {\bibinfo
  {journal} {Nature}\ }\textbf {\bibinfo {volume} {574}},\ \bibinfo {pages}
  {505} (\bibinfo {year} {2019})}\BibitemShut {NoStop}%
\bibitem [{\citenamefont {Fulton}\ and\ \citenamefont
  {Harris}(2004)}]{fultonRepresentation}%
  \BibitemOpen
  \bibfield  {author} {\bibinfo {author} {\bibfnamefont {W.}~\bibnamefont
  {Fulton}}\ and\ \bibinfo {author} {\bibfnamefont {J.}~\bibnamefont
  {Harris}},\ }\href {https://doi.org/10.1007/978-1-4612-0979-9} {\emph
  {\bibinfo {title} {Representation Theory: A First Course}}}\ (\bibinfo
  {publisher} {Springer New York},\ \bibinfo {address} {New York, NY},\
  \bibinfo {year} {2004})\BibitemShut {NoStop}%
\bibitem [{\citenamefont {Horodecki}\ \emph {et~al.}(1999)\citenamefont
  {Horodecki}, \citenamefont {Horodecki},\ and\ \citenamefont
  {Horodecki}}]{horodecki1999general}%
  \BibitemOpen
  \bibfield  {author} {\bibinfo {author} {\bibfnamefont {M.}~\bibnamefont
  {Horodecki}}, \bibinfo {author} {\bibfnamefont {P.}~\bibnamefont
  {Horodecki}},\ and\ \bibinfo {author} {\bibfnamefont {R.}~\bibnamefont
  {Horodecki}},\ }\bibfield  {title} {\bibinfo {title} {General teleportation
  channel, singlet fraction, and quasidistillation},\ }\href@noop {} {\bibfield
   {journal} {\bibinfo  {journal} {Physical Review A}\ }\textbf {\bibinfo
  {volume} {60}},\ \bibinfo {pages} {1888} (\bibinfo {year}
  {1999})}\BibitemShut {NoStop}%
\bibitem [{\citenamefont {Hardy}\ \emph {et~al.}(1952)\citenamefont {Hardy},
  \citenamefont {Littlewood}, \citenamefont {P{\'o}lya}, \citenamefont
  {P{\'o}lya} \emph {et~al.}}]{Ineq1952}%
  \BibitemOpen
  \bibfield  {author} {\bibinfo {author} {\bibfnamefont {G.~H.}\ \bibnamefont
  {Hardy}}, \bibinfo {author} {\bibfnamefont {J.~E.}\ \bibnamefont
  {Littlewood}}, \bibinfo {author} {\bibfnamefont {G.}~\bibnamefont
  {P{\'o}lya}}, \bibinfo {author} {\bibfnamefont {G.}~\bibnamefont
  {P{\'o}lya}}, \emph {et~al.},\ }\href@noop {} {\emph {\bibinfo {title}
  {Inequalities}}}\ (\bibinfo  {publisher} {Cambridge university press},\
  \bibinfo {year} {1952})\BibitemShut {NoStop}%
\bibitem [{\citenamefont {Bernstein}(1946)}]{BernsteinIneq}%
  \BibitemOpen
  \bibfield  {author} {\bibinfo {author} {\bibfnamefont {S.}~\bibnamefont
  {Bernstein}},\ }\href@noop {} {\bibinfo {title} {The theory of
  probabilities}},\ \bibinfo {howpublished} {Gastehizdat Publishing House,
  Moscow} (\bibinfo {year} {1946})\BibitemShut {NoStop}%
\bibitem [{\citenamefont {Hoeffding}(1963)}]{HoeffdingIneq}%
  \BibitemOpen
  \bibfield  {author} {\bibinfo {author} {\bibfnamefont {W.}~\bibnamefont
  {Hoeffding}},\ }\bibfield  {title} {\bibinfo {title} {Probability
  inequalities for sums of bounded random variables},\ }\href
  {https://doi.org/10.1080/01621459.1963.10500830} {\bibfield  {journal}
  {\bibinfo  {journal} {Journal of the American Statistical Association}\
  }\textbf {\bibinfo {volume} {58}},\ \bibinfo {pages} {13} (\bibinfo {year}
  {1963})},\ \Eprint
  {https://arxiv.org/abs/https://www.tandfonline.com/doi/pdf/10.1080/01621459.1963.10500830}
  {https://www.tandfonline.com/doi/pdf/10.1080/01621459.1963.10500830}
  \BibitemShut {NoStop}%
\bibitem [{\citenamefont {Xue}\ \emph {et~al.}(2019)\citenamefont {Xue},
  \citenamefont {Watson}, \citenamefont {Helsen}, \citenamefont {Ward},
  \citenamefont {Savage}, \citenamefont {Lagally}, \citenamefont {Coppersmith},
  \citenamefont {Eriksson}, \citenamefont {Wehner},\ and\ \citenamefont
  {Vandersypen}}]{Xue2019CRB}%
  \BibitemOpen
  \bibfield  {author} {\bibinfo {author} {\bibfnamefont {X.}~\bibnamefont
  {Xue}}, \bibinfo {author} {\bibfnamefont {T.~F.}\ \bibnamefont {Watson}},
  \bibinfo {author} {\bibfnamefont {J.}~\bibnamefont {Helsen}}, \bibinfo
  {author} {\bibfnamefont {D.~R.}\ \bibnamefont {Ward}}, \bibinfo {author}
  {\bibfnamefont {D.~E.}\ \bibnamefont {Savage}}, \bibinfo {author}
  {\bibfnamefont {M.~G.}\ \bibnamefont {Lagally}}, \bibinfo {author}
  {\bibfnamefont {S.~N.}\ \bibnamefont {Coppersmith}}, \bibinfo {author}
  {\bibfnamefont {M.~A.}\ \bibnamefont {Eriksson}}, \bibinfo {author}
  {\bibfnamefont {S.}~\bibnamefont {Wehner}},\ and\ \bibinfo {author}
  {\bibfnamefont {L.~M.~K.}\ \bibnamefont {Vandersypen}},\ }\bibfield  {title}
  {\bibinfo {title} {Benchmarking gate fidelities in a
  $\mathrm{Si}/\mathrm{SiGe}$ two-qubit device},\ }\href
  {https://doi.org/10.1103/PhysRevX.9.021011} {\bibfield  {journal} {\bibinfo
  {journal} {Phys. Rev. X}\ }\textbf {\bibinfo {volume} {9}},\ \bibinfo {pages}
  {021011} (\bibinfo {year} {2019})}\BibitemShut {NoStop}%
\end{thebibliography}%

\end{document}